\documentclass[12pt,oneside,letter]{article}
\usepackage[left=1in, right=1in, top=1in, bottom=1in]{geometry}

\usepackage{amsfonts}
\usepackage{amssymb}
\usepackage{amsmath}
\usepackage{amsthm}
\usepackage{bbm}
\usepackage{eurosym}
\usepackage{mathtools}
\usepackage{bm}
\usepackage[extra]{tipa}

\usepackage[toc]{appendix}
\usepackage[font=footnotesize, labelsep=period, labelfont=bf]{caption}
\usepackage{changepage}
\usepackage{color, colortbl, soul}
\usepackage{xcolor}
\usepackage[bottom]{footmisc}
\usepackage{graphicx}
\usepackage{array}
\usepackage{graphics}
\usepackage{epstopdf}
\usepackage{epsfig}
\usepackage{tabularx}
\usepackage[font=large]{caption}

\usepackage{epstopdf} %
\usepackage{setspace} %

\usepackage[hyperfootnotes=true]{hyperref}
\usepackage[latin1,utf8]{inputenc}
\usepackage{longtable}
\usepackage{lscape}
\usepackage{multirow} %
\usepackage[round]{natbib}
\usepackage{pdflscape}
\usepackage{pdfpages}
\usepackage{rotating}
\usepackage{booktabs,makecell}
\usepackage[detect-all]{siunitx}[=v2]
\usepackage{etoolbox}
\usepackage [autostyle, english = american]{csquotes}
\MakeOuterQuote{"}
\usepackage{epigraph}
\usepackage[normalem]{ulem}
\usepackage{colortbl}
\usepackage{hyperref}
\usepackage{comment}
\usepackage{subfig}
\usepackage{tikz}
\usepackage{titlesec}
\usepackage[USenglish]{babel}
\newtheorem{theorem}{Theorem}

\newtheorem{lemma}[theorem]{Lemma}

\newtheorem{hypothesis} {Hypothesis}
\newtheorem{proposition}{Result}

\newcolumntype{d}[1]{D{.}{.}{#1}}

\newcommand{\fn}{\footnotesize}

\newcommand{\rc}{\rowcolor{black!5}}

\definecolor{blue}{rgb}{0.00,0.07,1.00}
\definecolor{red}{rgb}{1.00,0.10,0.00}

\definecolor{black}{rgb}{0.00,0.00,0.00}
\definecolor{LightCyan}{rgb}{0.88,1,1}
\definecolor{White}{rgb}{1.00,1.00,1.00}
\definecolor{Gray}{gray}{0.95}
\definecolor{Oran}{rgb}{0.97,1.00,0.67}
\definecolor{darkred}{rgb}{0.6,0.0,0.0}
\definecolor{orange}{rgb}{1,0.5,0}
\definecolor{w}{rgb}{1.00,1.00,1.00}
\definecolor{b}{rgb}{0.00,0.00,0.00}
\def\sym#1{\ifmmode^{#1}\else\(^{#1}\)\fi}
\newcommand{\bl}{\color{blue}}

\newcommand{\E}{\operatorname{\E}}
\newcommand*{\mytab}[1]{\hyperref[{#1}]{Table~\ref*{#1}}}
\newcommand*{\myfig}[1]{\hyperref[{#1}]{Figure~\ref*{#1}}}
\newcommand*{\myfigs}[1]{\hyperref[{#1}]{Figures~\ref*{#1}}}
\newcommand*{\myfigx}[1]{\hyperref[{#1}]{\ref*{#1}}}
\newcommand*{\mysec}[1]{\hyperref[{#1}]{Section~\ref*{#1}}}
\newcommand*{\myeq}[1]{\hyperref[{#1}]{Equation~(\ref*{#1})}}
\newcommand*{\myeqs}[1]{\hyperref[{#1}]{Equations~(\ref*{#1})}}
\newcommand*{\myeqx}[1]{\hyperref[{#1}]{(\ref*{#1})}}
\newcommand*{\myhyp}[1]{\hyperref[{#1}]{Hypothesis~\ref*{#1}}}
\newcommand*{\mypred}[1]{\hyperref[{#1}]{Prediction~\ref*{#1}}}
\newcommand*{\mylem}[1]{\hyperref[{#1}]{Lemma~\ref*{#1}}}
\newcommand*{\myres}[1]{\hyperref[{#1}]{Result~\ref*{#1}}}

\newcommand*{\myIA}[1]{\hyperref[{#1}]{Appendix}}
\newcommand*{\mysecIA}[1]{\hyperref[{#1}]{Appendix~\ref*{#1}}}
\newcommand*{\mytabIA}[1]{\hyperref[{#1}]{Appendix Table~\ref*{#1}}}
\newcommand*{\myfigIA}[1]{\hyperref[{#1}]{Appendix Figure~\ref*{#1}}}
\newcommand*{\mytabIAs}[1]{\hyperref[{#1}]{Appendix Tables~\ref*{#1}}}
\titlespacing\section{0pt}{12pt plus 4pt minus 2pt}{0pt plus 2pt minus 2pt}
\titlespacing\subsection{0pt}{6pt plus 4pt minus 2pt}{0pt plus 2pt minus 2pt}
\titlespacing\subsubsection{0pt}{6pt plus 4pt minus 2pt}{0pt plus 2pt minus 2pt}

\usepackage{hyperref}
\hypersetup{
	colorlinks=true,
	raiselinks=true,
	breaklinks=false,
	linkcolor=blue,
	anchorcolor=darkred,
	citecolor=blue,
	urlcolor=blue}
\defcitealias{Diamond1987}{DV}

\DeclareUnicodeCharacter{2212}{-}
\begin{document}

\thispagestyle{empty}

\title{\LARGE{\bf The Double-Edged Sword \\ of Short-Selling Bans}\protect
\footnote{The authors would like to thank Dante Amengual, Wolfgang Bessler, Patrick Bolton, Gilles Chemla,  J\'{e}r\^{o}me Dugast, Marco Pagano, Evgenia Passari, Rafael Repullo,  Enrique Sentana, Ravi Shukla, Javier Suarez as well as seminar participants at the 2023 NFA meeting, the 2024 FMA Applied Finance Conference,  AUEB, CEMFI, Durham, Universit\'e Paris-Dauphine, and University of St.\! Andrews for their helpful and constructive comments. Pasquale Della Corte is with Imperial College London and the Centre for Economic Policy Research (CEPR), e-mail: \texttt{p.dellacorte@imperial.ac.uk}. Robert Kosowski is with Imperial College London, the Centre for Economic Policy Research, and the Oxford-Man Institute of Quantitative Finance, e-mail: \texttt{r.kosowski@imperial.ac.uk}. Dimitris Papadimitriou is with King's College London and the University of Cyprus, email: \texttt{dimitris.papadimitriou@kcl.ac.uk}. Nikolaos P. Rapanos is with Imperial College London, e-mail: \texttt{nikolaos.rapanos@imperial.ac.uk}.
}}

\author{
	\begin{tabular}[t]{c@{\extracolsep{1cm}}cc}
		{\large \bf  Pasquale Della Corte}          &       & {\large \bf Robert Kosowski}  \\
		{\large Imperial College London}                   &       &  {\large Imperial College London} \\
		{\large   \& CEPR}                         &       & {\large   \& CEPR} \\
		\addlinespace[20pt]
		{\large \bf Dimitris Papadimitriou}       &       & {\large \bf Nikolaos P. Rapanos} \\
		{\large King's College London} &       & {\large Imperial College London}  \\
        {\large   \& University of Cyprus}                         &       &  \\
	\end{tabular}
	\bigskip
}

\maketitle
\thispagestyle{empty}

\newpage

\thispagestyle{empty}
\vspace*{2cm}

\begin{center}
\title{\LARGE \bf{The Double-Edged Sword \\\vspace{10pt} of Short-Selling Bans}}
\end{center}
\maketitle

\vspace*{1cm}

\begin{abstract}
\noindent We develop a theoretical model that endogenizes the regulator's decision to impose short-selling bans to prevent large stock price declines. Empirically, we test the model's predictions using the cross-sectional variation in short-selling restrictions implemented across European countries in 2020. Consistent with our model, we find that bans had a detrimental effect on liquidity and failed to support the average price levels, but were effective in limiting large price drawdowns. Finally, we show that the effectiveness of the bans depends on the share of informed stockholders, a central variable in our framework, thus informing the design of more effective regulatory responses.

\vspace{1cm}
\noindent\textit{Keywords:} Short Selling, Ban, Liquidity, Price Discovery, Covid.\\
\medskip
\noindent\textit{JEL Classification:} G01, G12, G14, G18.
	
\end{abstract}

\newpage

\thispagestyle{empty}
\clearpage

\newpage
\setcounter{page}{1}

\onehalfspacing
\section{Introduction}
The effects of short-selling bans on market liquidity and share prices %
are part of a long-standing debate among regulators and academics around the world \citep[e.g.,][]{Beber2013,ESMA:2022}. Disagreement about the ability of these measures to restore market quality during volatile markets reemerged at the peak of the COVID-19 pandemic, when some European regulators chose to pull the brake used during previous financial crises and curb short-selling activity. Other regulators, however, decided against these restrictions as they can reduce market liquidity while having no positive effects on price levels. This divergence is well reflected in a statement released by the Financial Conduct Authority, the UK national watchdog, on March 23, 2020: \enquote{\it A great many investment and risk management strategies rely on the ability to take `long' and `short' positions. These benefit a wide range of ordinary investors ... The loss of these benefits would need to be carefully balanced before determining that any intervention to prevent short selling was appropriate.}%

In this paper, we contribute to this debate and provide novel insights on the potential costs and benefits of short-selling bans. We start our analysis by building a stylized model with {\it informed} and {\it noise} investors while endogenizing a {\it regulator}'s decision to impose a ban on short sales in order to prevent a large drop in asset prices. This model generates several testable predictions for both liquidity and asset returns, which we then empirically evaluate by exploiting the introduction of temporary short-selling bans in some but not all European countries in March 2020, an ideal setting to focus on the causal effect of bans. Consistent with our model, we show that the likelihood of imposing short sale restrictions is inversely related to the institutional ownership of each stock, a proxy for the share of {\it informed} stockholders. Moreover, short-selling bans widen bid-ask spreads, fail to support average returns, but help limit large price drawdowns. Overall, the potential costs of short-selling bans (deterioration in liquidity) increase for stocks with high institutional ownership, whereas the possible benefits (supporting the left tail of asset returns) improve for stocks with low institutional ownership.

Our findings offer valuable guidance to policy makers, as we show that the efficacy of short-selling bans in mitigating significant short-term price falls depends on the composition of stockholders and their level of informativeness. Specifically, our study indicates that restrictions on short sales can be effective in markets with a higher percentage of uninformed stockholders but may be detrimental in markets with a larger share of informed stockholders.  We uncover these results through a simple model that extends the work of \citet{Diamond1987} while using institutional ownership to quantify the relative proportion of informedness among stockholders. The use of this measure is supported by recent research \citep[e.g.,][]{Boehmer2009, bai2016have, davila2022identifying}, which finds that stocks owned by institutional investors exhibit higher price informativeness and concludes that these investors can be construed as informed traders.  Moreover, we corroborate the evidence of recent literature on the adverse effects of short-selling bans on liquidity and their inability to enhance average returns \citep[e.g.,][]{Beber2013, Enriques2020}, thus providing a framework to assess the trade-offs faced by regulators on the impact of short-selling bans using the most recent experience in Europe.

To guide our analysis, \mysec{sec:bans_model} builds on the seminal work of \citet{Diamond1987} and considers a single-period economy populated with {\it informed} and {\it noise} investors trading a stock against competitive {\it market makers}. However, we add two extensions to this setup. First, we introduce a {\it regulator} whose goal is to avert a sharp decline in asset prices and has to decide on whether to prohibit short-selling activity, while being uncertain about the liquidity needs of {\it noise} traders. The higher this uncertainty, the more likely the regulator decides to take action. Indeed, uncertainty is likely to have played an important role in the recent decision of various European countries to implement short-selling bans, {as corroborated by our discussions with representatives from several European regulatory bodies}. For example, an ESMA opinion issued in March 2020 concerning the French short-selling bans states: \textit{``AMF reports to have observed examples of disinformation, rumours and false news ... these rumours may affect listed companies and may damage the confidence of investors on an efficient market''}. Second, each group of traders owns the stock with a given probability and the effectiveness of a restrictive short-selling policy depends on the ratio of {\it informed} stockholders relative to {\it noise} stockholders. In the context of our model, imposing a short-selling ban is not always an optimal solution from the perspective of the {\it regulator}. In particular, the probability of a sharp decline in prices depends on two factors: the probability of a low bid price posted by {\it market makers} and the probability of a sell order coming from active traders given that the bid price is low. The latter always decreases when short-selling is not allowed as traders without the stock are unable to submit a sell order. However, the former may decrease or increase depending on the share of {\it informed} stockholders. The larger this fraction, the higher the adverse selection faced by {\it market makers} when short-selling bans are in place, and they may set lower bid prices in proportion to their adverse selection. Thus, we get the following testable hypotheses. First, the likelihood of imposing short-selling bans increases when institutional ownership is lower. Second, bans lead to a deterioration of liquidity with the effect being more pronounced for high institutional ownership stocks. Finally, the left tail of returns is supported when bans are implemented and this support is larger for stocks with low institutional ownership.

In our empirical analysis, we exploit a quasi-natural experiment associated with the introduction of short-selling restrictions in some European countries in March 2020, as described in \mysec{sec:bans_data}.  The power to temporarily restrict the short sales of a financial instrument in European trading venues is granted to national authorities by the European Union Short Selling Regulation in the case of adverse circumstances. Starting on March 17, 2020, Austria, Belgium, France, Greece, Italy, and Spain introduced temporary short-selling bans on shares admitted to their trading venues. The restrictions applied to any natural or legal person, regardless of where they were located, covered all stocks traded in cash and derivatives markets, and lasted for about two months.  Our identification strategy follows the existing literature and relies on a difference-in-differences method with pre- and post-intervention periods capturing a  month of observations before and after the regulatory change. Additionally, our treatment group consists of individual stocks traded in the European countries that adopted short-selling bans, while the control group includes individual stocks traded in the other twelve European countries. While this approach helps control for time-invariant stock-level heterogeneity and common shocks, we recognize that the decision to implement a ban and its timing may be endogenous, potentially responding to deteriorating market conditions. To address this concern and mitigate potential omitted variable bias, we also complement our baseline analysis with an instrumental variable exercise to isolate exogenous variation in the likelihood of a ban.

Using data from 17 European stock markets, we test various predictions implied from our model in \mysec{sec:bans_results}.  In our hypotheses, we compare stocks in countries with and without short-selling bans, and focus on the mechanism that our theoretical model identifies and on the role of institutional ownership. Overall, we find the following results. First, we find that institutional ownership is lower in countries that imposed short-selling bans. Second, in line with existing literature \citep[e.g.,][]{Beber2013}, we find that short-selling bans during the COVID-19 pandemic were associated with a deterioration of liquidity by up to 13 basis points, as measured by average bid-ask spreads. To strengthen the causal interpretation, we follow an approach similar to \citet{Beber2013}, using sovereign CDS spreads and the Financial Stress Index of \citet{Duprey2017} as instruments. The results remain qualitatively unchanged. Third, we demonstrate that stocks with higher institutional ownership suffered from a greater deterioration in liquidity. Indicatively, bid-ask spreads of stocks with low institutional ownership increase on average by 5-8 basis points (which is statistically insignificant) as a result of the bans, whereas the bid-ask spreads of stocks with high institutional ownership increase by 22-29 basis points. Fourth, we show that short-selling bans offered left-tail support, as measured by the maximum drawdown. Finally, stocks with lower institutional ownership benefit more in terms of limiting extreme negative outcomes (i.e., left-tail support). We estimate that short-selling bans lower the maximum drawdown (in absolute terms) by 310 basis points relative to the stocks without short-selling bans, and the effect is even stronger for stocks with low institutional ownership. Hence, we quantify the trade-off that policy makers face when imposing short-term short-selling bans: bid-ask spreads widen by around 13 basis points, and maximum drawdown is reduced by 310 basis points. In \mysec{sec:robustness}, we conduct a variety of additional exercises that confirm our main results, such as matching stocks by market capitalization and industry classification, running a placebo outcome exercise, and extending the pre- and post-intervention windows.

Our paper is closely related to three strands of the literature. First, we contribute to the growing literature on the effects of short-selling bans on market quality. For example, \citet{Beber2013} investigate the impact of short-selling bans around the world during the 2007-2009 financial crisis and conclude that short-selling bans are detrimental to liquidity, slow price discovery, and fail to support prices except for US financial stocks. More recently, and most closely related to our study, \citet{bessler20212020} analyze the 2020 European short-selling bans and their impact on several measures of market quality, confirming that the bans were detrimental for market liquidity. \citet*{Beber2017}, moreover, find that bans increase the probability of default and volatility. \citet*{Boehmer2013} study the response of liquidity to the short-selling ban imposed during the global financial crisis in the US by exploiting the difference between financial stocks that were targeted by the ban and those that were not, and find that liquidity worsened in all but the smallest stocks. Similarly, \citet{Marsh2012} examine the UK's short sales ban in 2008 and document a deterioration of liquidity on affected stocks and an overall negative effect on the quality of the market; they also recognize, however, that ``if the goal of the FSA was to arrest sharp declines in financials' stock prices'', then their results may signify that the policy of banning short sales was successful. Finally, \citet{battalio2011} investigate the impact of 2008 bans on the option market in the US and document a dramatic increase in the bid-ask spreads of affected options. While we build on this literature and replicate some of its results as a preliminary step, our paper offers a more nuanced perspective on short-selling bans. Beyond providing an analysis of market outcomes once bans were implemented, we focus on the underlying mechanism that drives these effects, their causal interpretation, and ultimately the regulator's decision to impose these bans.

In addition, we speak to a growing theoretical literature that evaluates the effects of short-selling bans. \citet{Miller1977} predicts that short-selling bans lead to overpricing, while \citet{Diamond1987} build on  \citet{Glosten1985} and  show that this is not true in a rational expectations framework; stocks are not systematically overpriced when short sales are prohibited, but liquidity and price discovery are compromised. \citet{Hong2003} build a heterogeneous agent model and find that short-selling bans may aggravate price declines, while \citet*{Bai2006} point out that short-selling constraints can increase uncertainty about the asset in a model with risk-averse investors, and thus also lead to a decline in prices. On the other hand, \citet{Brunnermeier2014} show how short selling impacts the fundamentals of firms rather than just the price discovery process.
They argue that financial institutions may be vulnerable to predatory short selling, which may lead to a bank-run equilibrium, and their model provides a potential justification for temporary restrictions on short selling. Finally, \citet{Dixon2021} endogenizes the incentives for information acquisition when short selling is costly and shows that a ban increases adverse selection on the sell side and reduces it on the buy side. We add to this literature by developing a theoretical model in which the regulator’s decision to impose a short-selling ban is endogenous, aimed at curbing sharp price declines. A key determinant of this decision is the ratio of informed to noise traders holding a stock, empirically proxied with institutional ownership. Consistent with prior studies, we find that bid–ask spreads widen following bans, although the effect is economically modest and concentrated in stocks (and markets) with a high share of informed traders. Importantly, unlike earlier work, we move beyond average price levels and document left-tail support, particularly in stocks (and markets) dominated by noise traders.

Finally, our work contributes to the broader empirical literature on short selling \citep*[e.g.,][]{Boehmer2010, Reed2013}.  In particular, \citet{Saffi2011} show that stocks subject to short-selling constraints, as measured by low lending supply, have lower price efficiency and relaxing those constraints does not lead to instability in the form of a higher probability of left-tail returns. \citet*{muravyev2025anomalies}  provide evidence that many pricing anomalies are eliminated once stock-level shorting costs are taken into account. \citet*{deng2020short} use SEC's Regulation SHO as a natural experiment and find that the lifting of short-sale constraints reduces the stock price crash risk. \citet{barardehi2019short} study the effects of short-lived short-selling restrictions, originating from Rule 201 in the US, which restricts the placement of marketable short-sale orders once an intraday return drops below $−10\%$. They find that prices are supported, volatility decreases, and ask-side depth increases. \citet{Jones2002} document evidence consistent with the overpricing hypothesis when short-selling constraints bind, but \citet*{Diether2009} find that short sellers correctly predict negative future returns. In a similar vein, \citet*{DellaCorte2022} exploit granular net short positions disclosed at the investor-stock level for European markets and show that a measure of short conviction harvests information from better informed investors, while \citet*{AnHuangLouShi2021} show that most mutual funds do not engage in short selling, suggesting that short interest is concentrated in a subset of institutional investors. Finally, \citet*{Bris2007} find that short-selling constraints reduce price efficiency and are associated with less negative skewness at the market level.    

The remainder of this paper is organized as follows. \mysec{sec:bans_model} introduces the model and defines the empirical predictions, \mysec{sec:bans_data} describes the data, \mysec{sec:bans_results} presents the results of our empirical tests, and \mysec{sec:robustness} performs some robustness exercises. We provide our concluding remarks in \mysec{sec:bans_conclusion}. A separate Internet \myIA{app:internet_appendix} contains additional technical details and empirical results.

\section{Model}\label{sec:bans_model}
In this section, we first extend the model of \citet{Diamond1987} by introducing a {\it regulator} who can impose a short-selling ban to avert a sharp decline in prices, and then determine the conditions under which it is optimal for the {\it regulator} to impose restrictions on short sales. We conclude with a set of empirical predictions, tested later in this paper. Model proofs and extensions are reported in \mysecIA{app:proof_simulations}.

\subsection{Setting}
We consider a static model with a single stock in the spirit of \citet{Diamond1987}.\footnote{Note that the main results of our model would not change if we introduced multiple stocks.} The value of the risky asset is denoted by $V$ and follows a Bernoulli distribution that takes the liquidation value of one with probability $p$, or zero with probability $1-p$. While the liquidation value is paid in the future, we abstract from any discounting for simplicity. This model comprises an infinite number of risk-neutral traders who sequentially enter the market and want to trade with probability $g$. They can also decide not to engage in trading with probability $1-g$, a case in which no trade is observed. Active traders initiate a buy or sell order for a unit of the risky asset. When a sell order is submitted by a trader who owns no asset, we have a short sale. Since no market participants can distinguish sell orders from short sales, the set of observed actions includes buy orders, sell-or-short orders, and no trade.

There are two types of traders in our setting: {\it informed} traders and {\it noise} traders. The first group of traders has a probability mass of $\alpha$, privately knows the true liquidating value of the stock, and trades only for information reasons. An active {\it informed} trader submits a buy order when $V$ is high and a sell (or short) order when $V$ is low. The second group of investors has a probability mass of $1-\alpha$, infers the value of the risky asset using publicly available information, and trades only for liquidity reasons exogenous to our model. A {\it noise} trader with a desire to trade submits a buy order with probability $1-\eta_s$ and a sell-or-short order with probability $\eta_s$. Moreover, each group of traders owns the stock with a given probability.  We use $h_I$ to indicate the probability that an {\it informed} trader owns a share of the risky asset prior to submitting a sell-or-short order, and $h_N$ to denote the corresponding quantity for a {\it noise} trader. $h_I$ and $h_N$ can also be interpreted as the fraction of {\it informed} and {\it noise} traders who own the stock, respectively. When a trader faces short-sale restrictions and owns no stock, no trade is observed.

In addition to {\it informed} and {\it noise} traders, our model is also populated by competitive risk-neutral {\it market makers} facing no inventory costs or constraints, so that the expected profit from each trade is zero. The {\it market maker} has no access to private information, but he observes all trades as they take place and knows the probability of a sell order from a {\it noise} trader. He sets bid and ask prices such that profits and losses from transacting with {\it noise} and {\it informed} traders, respectively, offset each other. Ultimately, the {\it market maker} will post a bid (ask) price that equals his expectation of $V$ conditional on public information and depending on whether the transaction is a sell-or-short  (buy) order. Finally, we assume that the demand for the risky asset is bounded, as otherwise the {\it informed} traders would be willing to trade an infinite amount of the risky asset.

Unlike the model of \citet{Diamond1987}, our economy also includes a {\it regulator}, whose main objective is to avert a significant price decline in the short term by ensuring that
\begin{equation*}\label{eq:obj_reg}
	P(q<c)<x,
\end{equation*}
where $q$ is the equilibrium price of the risky asset in the next period, $c$ is a sufficiently low price threshold, $x$ is a confidence level, and $P(q<c)$ denotes the perceived probability that $q$ falls below $c$. We will provide further discussion on the regulator's objective in \mysec{sec:discmodel}.  At the time of her decision, the {\it regulator} only knows the probability distribution of a sell order from a {\it noise} trader, denoted by $f(\eta_s)$, so from her perspective the probability of a sharp decline in price is never zero. For example, if $c=\text{\euro}600$, $x=5\%$, and the price today is $\text{\euro}1,000$, the goal of the {\it regulator} is to prevent a decline in the price of the stock below $\text{\euro}600$ with a confidence interval of $5\%$. When this probability is higher than $x=5\%$, the {\it regulator} must decide whether or not to impose a ban on short sales of the risky asset, as a sharp price drop represents a potential threat to financial stability.\footnote{Even though our model is static, it is important to note that we view the regulator's objective as a short-term goal, since the final payoff, and thus the price of the asset, is totally exogenous ($V$ is $0$ or $1$).} For example, \citet{Brunnermeier2014} point out that a sharp decrease in the share price of financial institutions combined with leverage constraints may lead to a bank run in equilibrium. Similarly, in March 2020, the ESMA expressed its opinion on the short-selling bans implemented by the Autorit\'e des March\'es Financiers (AMF) in France. The opinion stated that: \textit{"the AMF considers that the growth of short positions betting on negative news (be they real or ill-based)... could destabilize markets in a way that could be self-reinforcing, with downward price spirals"}.  %

To preview our results, the equilibrium price of the risky asset can fall below the regulator's price threshold when the {\it market maker} sets a bid price below such threshold and simultaneously an {\it informed} trader submits a sell-or-short order. The joint likelihood of these events depends on the probability of a sell order submitted by a {\it noise} trader. In particular, when the liquidity needs of {\it noise} traders are high, the bid price is unlikely to fall below the price threshold as the {\it market maker} anticipates that any sell-or-short order is likely to come from a {\it noise} trader, and thus be uninformative about the liquidation value of the risky asset. In contrast, the probability of having a bid price falling below the price threshold increases when the liquidity needs of {\it noise} traders are low. In this case, any sell order is likely to come from an {\it informed} trader and thus carry valuable information about the liquidating value of the risky asset. As a result, the {\it regulator} will decide whether to impose a short-sale restriction depending on the effect of $\eta_n$ on the combined likelihood of the {\it market maker} setting a bid price below $q$ and an {\it informed} trader submitting a sell-or-short order.

Compared to \citet{Diamond1987}, we further introduce the possibility that $h_N$ and $h_I$ may not be identical, thus relaxing one of their assumptions. In doing so, the {\it market maker} can set a different bid price depending on whether short-selling activity is allowed or restricted. As a result, the relationship between $h_N$ and $h_I$ will affect the policy decision of the {\it regulator} on whether to introduce a short-selling ban. 

\subsection{Unconstrained Short Selling}
We first study the model when all traders can freely short the asset. Akin to \citet{Diamond1987}, a {\it market maker} sets the bid price equal to his expectation of $V$ conditional on public information and the fact that the transaction is a sale-or-short order. This is equivalent to $Bid =\mathbb{E}\left[V|Sell\right] = \textstyle\sum\nolimits_{v} v\cdot P\left(V=v|Sell\right)= P\left(V=1|Sell\right)$ since $V$ takes either the value of one or zero. Using the Bayes' rule, the price at which a {\it market maker} is willing to buy a share of the risky asset is then given by
\begin{align}\label{eq:bid}
	Bid = \frac{P(Sell|V=1)\times P(V=1)}{P(Sell)}.
\end{align}
When the liquidation value of the risky asset is high, only a {\it noise} trader would submit a sell-or-short order. This means that the conditional probability of a sell-or-short order is given by $P(Sell|V=1) = g(1-\alpha)\eta_s$, while its unconditional probability is equal to $P(Sell) = g(1-\alpha)\eta_s + g\alpha(1-p)$. Here, $g\alpha(1-p)$ accounts for the probability that an {\it informed} trader submits a sell-or-short order, which only occurs when the liquidation value of the risky asset is low. Hence, when traders can freely short the asset, we get
\begin{align}\label{eq:bid_unc}
	Bid = \frac{ (1-\alpha)\eta_sp}{(1-\alpha)\eta_s + \alpha(1-p)}.
\end{align}
Similarly, $Ask = \mathbb{E}\left[V|Buy\right] = P(V=1|Buy)$. 
Since any trader can submit a buy order when the liquidating value of the risky asset is high, the conditional probability of a buy order is computed as $P(Buy|V=1) = g(1-\alpha)(1-\eta_s) + g\alpha$, whereas its unconditional probability follows from $P(Buy) = g(1-\alpha)(1-\eta_s) + g\alpha p$. Therefore, using again Baye's rule, we get
\begin{align}\label{eq:ask_unc}
	Ask = \frac{ (1-\alpha)(1-\eta_s)p + \alpha p}{(1-\alpha)(1-\eta_s) + \alpha p}.
\end{align}
Finally, when there is no trade, the {\it market maker} cannot update his beliefs and sets the price of the risky asset equal to the expectation of $V$ conditional on public information. Therefore, the no trade price is $\text{\it No Trade} = \mathbb{E}\left[V|\text{\it No Trade}\right] = p$, which is independent of $\eta_s$.

When short selling is allowed, the {\it regulator} has to decide whether to introduce a ban on short sales. Since the prices posted by a {\it market maker} are random from her perspective, the equilibrium price $q$ follows a compound probability distribution that depends on $f(\eta_s)$ as well as $\mathbb{E}\left[V|Sell\right]$, $\mathbb{E}\left[V|Buy\right]$, and  $\mathbb{E}\left[V|\text{\it No Trade}\right]$ with probabilities $P(Buy|\eta_s)$, $P(Sell|\eta_s)$ and $P(\text{\it No Trade}|\eta_s)$, respectively. Put differently, conditional on the unknown parameter $\eta_s$, the {\it regulator} knows that $q$ follows a three-point probability distribution. To derive the probability of a sharp decline in prices in the absence of any short-selling restriction, the price threshold $c$ is assumed to be less than the maximum bid price and less than the minimum no-trade price, i.e.,  $c<\frac{(1-\alpha)p}{1-\alpha p}$. We then have the following lemma:

\begin{lemma}{\label{lem:prob}}
When all agents are unconstrained, the regulator's perceived probability that the stock price $q$ falls below a price threshold $c$ is given by:
\begin{equation}\label{eq:probq}
    P(q<c)=\int_{0}^{K}\left((1-p)g\alpha+(1-\alpha)g\eta\right)f(\eta)d\eta
\end{equation}
where $K=\frac{c\alpha(1-p)}{(1-\alpha)(p-c)}$.
If we assume that $\eta_s\sim U(0,1)$, a standard uniform distribution, we then obtain
\begin{equation*}
	P(q<c)=(1-p)\alpha g K+\frac{1-\alpha}{2}gK^2. 
\end{equation*} 
\end{lemma} 
\begin{proof}
	See \mysecIA{app:proof_lemma1}. 
\end{proof}

According to \mylem{lem:prob}, a sharp decline in prices is more likely to happen (i.e., $P(q<c)\,\uparrow$) when {\it noise} traders are unlikely to sell (i.e., $ P(\eta_s<K)\,\uparrow$), and the expected probability of a sell order potentially coming from {\it informed}  traders with negative information is high (i.e., $ \mathbb{E}\left[P(Sell|\eta_s)\mid\eta_s<K\right]\,\uparrow$).
Indeed, when the probability of a sell-or-short order from {\it noise} traders is high, it is unlikely that {\it market makers} will set a low bid price, as any sell order is unlikely to reflect valuable information about $V$. For tractability reasons, the second part of the lemma employs a standard uniform distribution for $\eta_s$, which implies that the {\it regulator} has an uninformative prior about the liquidity needs of {\it noise} traders. In \mylem{lem:SSD} in \mysecIA{app:technical}, moreover, we specify a general family of distributions and show that $P(q<c)$ increases as the variance of $\eta_s$ increases. That is, the  {\it regulator} is more likely to impose a ban when uncertainty about the liquidity needs of the {\it noise} traders increases. We can then obtain a necessary condition for the implementation of a short-selling ban. 
\begin{proposition}\label{prop:zeta}
There exists a threshold $\zeta$ such that the {\it regulator} imposes a ban on short sales \textit{only if}:
$$\frac{c\alpha(1-p)}{(1-\alpha)(p-c)}>\zeta.$$
\end{proposition}
\begin{proof}
	See \mysecIA{app:proof_result1}
\end{proof}

In particular, the {\it regulator} may impose a ban when the probability of a high payoff is low (i.e., $p \downarrow$),  the probability mass of {\it informed} traders is high ($\alpha \uparrow$), or when the desired support for the left tail of prices (captured by the threshold $c$) increases.\footnote{These conditions are likely to arise during a crisis, which can also be driven by country-specific shocks. In our empirical analysis, we control for the latter using sovereign CDS spreads.} It is important to emphasize that, in our setting, the {\it regulator}'s primary concern is preventing a sudden drop in prices. This means that the {\it regulator} will not consider the impact of his policy action on the informational efficiency of the market. For example, the regulator could potentially harm market efficiency by imposing a short-selling ban when the market fundamentals are weak, as indicated by a low $p$. While the conditions from \myres{prop:zeta} give us the potential trigger for the {\it regulator}'s intervention, the actual enforcement of bans will also depend on the impact of the new rules and the subsequent distribution of the price. The following section thus examines the pricing implications of a short-selling ban imposed by the {\it regulator}. In doing so, we can endogenously determine the conditions under which a regulator would optimally choose to restrict short-selling activity. 

\subsection{Restrictions to Short Sales}
When short sales are allowed, the {\it market maker} cannot distinguish a sale order from a short sale. However, when short selling is prohibited, only investors holding the stock can sell it, and {\it market makers} form bid prices taking this information into account. Importantly, the fraction of {\it informed} traders owning the stock ($h_I$) may differ from the fraction of {\it noise} traders possessing the stock ($h_N$), thus affecting the adverse selection faced by {\it market makers}. Here, we study the effect of short-selling bans on the bid-ask spread and get the following lemma:

\begin{lemma}\label{lem:bidask}
When short-selling bans are in place, the bid-ask spread increases, relative to the unconstrained case, if and only if $\frac{h_I}{h_N}>1.$
\end{lemma} 

Intuitively, when the fraction of informed traders owning the asset exceeds that of noise traders ($h_I > h_N$),  a short-selling ban makes any sell order relatively more likely to originate from an {\it informed} trader. That is, conditional on observing a sell order, market makers infer that the order is more likely to originate from an informed trader with negative information, thereby increasing adverse selection and widening spreads.\footnote{A related intuition appears in \citet{Dixon2021}, although their mechanism operates through endogenous information acquisition rather than exogenous differences in stock ownership.} Mathematically, we  first derive bid and ask prices when short selling is restricted and then draw a comparison with bid and ask prices formed when short selling is allowed. The bid price set by {\it market makers} under short-selling bans is given by
\begin{equation}\label{eq:bid_con}
	\widetilde{Bid}=\frac{(1-\alpha)h_N\eta_s p}{\alpha  h_I(1-p)+(1-\alpha)h_N\eta_s},
\end{equation}
whereas the corresponding ask price is equal to
\begin{equation}\label{eq:ask_con}
	\widetilde{Ask}= \frac{ (1-\alpha)(1-\eta_s)p + \alpha p}{(1-\alpha)(1-\eta_s) + \alpha p}.
\end{equation}
The comparison of \myeqs{eq:bid_unc} and \myeqx{eq:bid_con} reveals that the difference between $\widetilde{Bid}$ and $Bid$ depends on the ratio between {\it informed} and {\it noise} traders holding the stock, and $\widetilde{Bid} < Bid$ when $h_I/ h_N > 1$. From \myeqs{eq:ask_unc} and \myeqx{eq:ask_con}, we instead see that $\widetilde{Ask}$ and $Ask$ remain identical as short-selling bans do not affect the buying activity of traders. Therefore, the bid-ask spread, a commonly used measure of liquidity, widens under short-selling bans when $h_I > h_N$. %

When short-selling is prohibited and there is no trade, {\it market makers} will condition on $\eta_s$ to determine their no trade price. When $\eta_s$ is low (high), {\it market makers} may attribute the absence of trade to constrained {\it informed} ({\it noise}) traders who cannot short the stock, thus updating downwards (upwards) their expectation of $V$. As a result, the no trade price is computed as $\widetilde{\text{\it No Trade}}=E[V|\text{\it No Trade},\eta_s]=\frac{((1-g)+g(1-\alpha)\eta_s(1-h_N))p}{1-g+g((1-\alpha)\eta_s(1-h_N)+\alpha(1-p)(1-h_I))}$, which is increasing in $\eta_s$ and is always less than $p$. 
A graphical illustration of the relationship between $\eta_s$ and the prices quoted by {\it market makers} is provided in \myfigIA{fig:simulated_bid_ask}.

We now examine the probability of a sharp decline in price when the {\it regulator} imposes bans on short sales, and use $\tilde{q}$ to refer to the equilibrium price of the risky asset from the {\it regulator}'s perspective.
Similar to \mylem{lem:prob}, we first determine the region of $\eta_s$ that makes a bid price sufficiently low.\footnote{We assume that $\smash{c<\min\left\{\frac{(1-\alpha)h_{N} p}{\alpha  h_I(1-p)+(1-\alpha)h_N},\frac{(1-g)p}{1-g+g\alpha(1-p)(1-h_I)}\right\}}$, where the first term is the maximum bid price and the second term is the minimum possible no trade price. Then, $\smash{\widetilde{Bid}<c}$, {\it iff} $\smash{\eta_s<(h_I/h_N)K}$.} Hence, we obtain the following lemma:
\begin{lemma}{\label{lem:conprob}}
When short-selling bans are in place, the regulator's perceived probability that the price $\tilde{q}$ falls below the threshold $c$ is equal to:
$$P(\tilde{q}<c)=\int_{0}^{\frac{h_I}{h_N}K}((1-p) \alpha g h_I+(1-\alpha)g h_N\eta)f(\eta)d\eta,$$
where $K=\frac{c\alpha(1-p)}{(1-\alpha)(p-c)}$.
 In particular, if $\eta_s\sim U(0,1)$ then: 
\begin{align}\label{eq:probtildeq}
P(\tilde{q}<c)=\frac{h_I^2}{h_N}P(q<c).
\end{align}
\end{lemma}
\begin{proof}
	See \mysecIA{app:proof_lemma3}
\end{proof}
\myfig{fig:probability_q} displays the relationship between $K$ (an increasing function of $c$) and the probability that the equilibrium price falls below the price threshold $c$ under different scenarios using $\alpha=0.5$, $p= 0.5$, $g= 0.9$, $h_N=0.3$, and $h_I = 0.2$ (so that $h_I^2/h_N = 0.13$) or $h_I = 0.6$ (so that $h_I^2/h_N = 1.2$). The black line corresponds to the baseline scenario without any restrictions on short-selling. Based on \myres{prop:zeta}, there exists a value $\zeta$ so that when $K>\zeta$, the likelihood of a significant decrease in price exceeds the confidence level $x$, potentially leading to a regulatory intervention. The red (blue) line, moreover, denotes the scenario with short-selling bans and shows how the introduction of a short-selling ban leads to an increase (decrease) in the probability of a sharp price decline when $h_I^2/h_N > 1$ ($h_I^2/h_N < 1$). In particular, the difference between the baseline scenario (short selling is allowed) and the alternative scenario (short selling is prohibited)  depends on  $h_I^2/h_N$, which can be seen as the product of $h_I$ and $h_I/h_N$. The former affects the conditional likelihood of a sell order, whereas the latter impacts the probability that the bid price is low. 
\begin{center}
	\textsc{\myfig{fig:probability_q} about here} 
\end{center}
\myeq{eq:probtildeq} assumes that the liquidity needs of {\it noise} traders $\eta_s$ follow a standard uniform distribution. However, $P(\tilde{q}<c)$ would remain an increasing function of $\frac{h_I^2}{h_N}$ also for alternative distributions of $\eta_s$ (for example, for $\eta_s\sim Beta(a,1)$, with $a\geq 1$).
Therefore, we can now get a necessary and sufficient condition for the short-selling ban to be successful:
\begin{proposition}\label{prop:condition}
If $\eta_s\sim U(0,1)$, the {\it regulator} successfully manages to decrease the probability that the price falls below the prespecified level $c$ if and only if 
\begin{equation}\label{eq:assum}
    h_I^2<h_N.
\end{equation}
\end{proposition}

According to \myres{prop:condition}, the effectiveness of a restrictive short-selling policy depends on whether the proportion of {\it noise} traders holding the stock is sufficiently larger than the proportion of {\it informed} traders owning the stock. The {\it regulator} can achieve his goal by reducing the probability that there will be a sell order, or by reducing the likelihood that the bid price will be smaller than $c$. In particular, when the fraction of {\it informed} traders owning the asset is relatively low ($h_I<h_N$), a short-selling ban prevents a larger fraction of {\it noise} traders from selling, thereby reducing adverse selection and supporting prices. However, even when $h_I>h_N$, there may still be cases in which a ban reduces the likelihood of a rapid price drawdown by converting some events that would otherwise occur as short sales into less informative no-trade events. In times when irrational exuberance has led many {\it noise} traders to own the asset and has driven {\it informed} investors away, the intervention of the {\it regulator} will be more warranted and more successful. If, however, the proportion of {\it informed} traders who own the stock is large relative to the corresponding fraction of {\it noise} traders, then imposing bans may not even support prices.\footnote{Previous empirical literature has hinted towards a mixed result concerning whether bans succeed in supporting stock prices. We explore this issue further in the next section.} Having examined the effect of bans on the left tail of prices, we now explore their effects on key measures of central tendency. More specifically, assuming for simplicity that the probabilities of a buy and of a sell order are both less than 1/2 (e.g., if $g<1/2$), we have: 

\begin{lemma}\label{lem:meanmedian}
When short-selling bans are implemented, the mean price remains unchanged, while the median price decreases. 
\end{lemma}
\begin{proof}
	See \mysecIA{app:proof_lemma4}
\end{proof}

\subsection{Testable Hypotheses} \label{sec:hypotheses}
We use our simple model to guide our empirical investigation on assessing the impact of short-selling bans on prices and liquidity. Before formalizing our testable hypothesis, we should find a suitable empirical proxy for $h_I$ and $h_N$. 
Recent literature \citep[e.g.,][]{Boehmer2009,bai2016have,davila2022identifying} documents that stocks with high institutional ownership have higher price informativeness, thus suggesting that institutional investors can be regarded as \textit{informed} traders. Building on this finding, let $m$ be a conjugate probability denoting the fraction of the risky stock owned by {\it informed} traders and $1-m$ the corresponding probability for {\it noise} traders. 
If we assume that {\it informed} traders are mostly financial institutions, then $m$ can be quantified with the institutional ownership of a stock.\footnote{While institutional ownership is an imperfect proxy for informed ownership, prior literature documents that stocks with higher institutional ownership exhibit greater price informativeness on average. Moreover, our interpretation does not require all institutions to be informed traders, but only that institutional ownership is positively associated with the relative presence of sophisticated investors.} Using the Bayes' rule and assuming that the fraction of {\it informed} investors (say $\alpha$) in the whole economy is fixed, we first obtain $\frac{h_I}{h_N}=\frac{m}{1-m}\frac{1-\alpha}{\alpha}$ and then in the proportional form 
\begin{equation*}
\frac{h_I}{h_N}\propto\frac{m}{1-m} \qquad \text{and} \qquad \frac{h_I^2}{h_N}\propto\frac{m^2}{1-m}.
\end{equation*}

According to our model, the {\it regulator} should only impose bans when $h_I^2/h_N <1$. Hence, in markets  populated by many sophisticated stock owners, the {\it regulator} should refrain from imposing any short-selling restrictions. We thus formalize the following hypothesis:
\setcounter{hypothesis}{0}
\begin{hypothesis}
\label{h:imposingban}
Short-selling bans are more likely to be imposed by regulators in markets with low institutional ownership.
\end{hypothesis}

Short-selling bans affect the bid-ask spread by altering the bid price. This happens as the composition of the potential sellers changes with $h_I/h_N$, that is, the relative likelihood that an {\it informed} investor owns the stock. Our model predicts that when $h_I/h_N>1$, it is more likely that a sell order is initiated by an {\it informed} trader and the {\it market maker} submits a lower bid. Moreover, stocks with a larger number of sophisticated owners are more likely to experience wider bid-ask spreads after the introduction of short-selling bans. This happens as adverse selection worsens relative to the unconstrained case, as the fraction of {\it informed} sellers is relatively larger. When short selling is allowed, a sell order may arise from any {\it informed} or {\it noise} trader in the economy, and the {\it market maker} adjusts his expectation of the payoff, depending on the overall relative masses of {\it informed}-to-{\it noise} traders. In contrast, when short-selling is prohibited, the pool of potential sellers changes and includes only those who already own the stock. Therefore, the fraction of {\it informed} traders holding the stock becomes relevant. The higher this fraction, the more the {\it market maker} thinks that a sell order contains information and adjusts the bid price downwards. Since the bid-ask spread is commonly used as a measure of liquidity, our model is consistent with the following hypothesis:

\begin{hypothesis}\label{h:liquidity}
Under certain conditions ($h_I>h_N$), short-selling bans lead to a deterioration in liquidity; moreover, the higher the institutional ownership of a stock, the larger the increase in bid-ask spreads.
\end{hypothesis}

In our model, the {\it regulator} can successfully reduce the probability of a sharp price decline as long as $h_I^2/h_N<1$. Since this effect dominates any changes in the bid price, a short-selling ban leads to a thinner left tail of the price distribution, and following a similar logic as before, the lower the fraction of \textit{informed} stockholders the more pronounced this effect becomes. Also, consistent with the findings of \citet{Beber2013}, our model also predicts that short-selling bans have no impact on the mean return. The median return, however, decreases under short-selling bans as no-trade actions are more likely to reflect negative news. This hypothesis is also discussed in \citet{Dixon2021}, while the focus in \citet{Diamond1987} is on the dynamics of the bid-ask spread. %
From the above discussion, we can then derive the following hypothesis on the effect of bans on the distribution of prices:
\begin{hypothesis}\label{h:support}
Under certain conditions ($h_I^2<h_N$), short-selling bans support the left tail of returns; moreover, the lower the institutional ownership of a stock, the higher this support.
Finally, the median return decreases relative to the unconstrained case, while the mean remains the same.
\end{hypothesis}

Overall, it can very well be the case that a {\it regulator} manages to avoid a huge price drop (if $h_I^2/h_N<1$) while causing a deterioration of liquidity (if $h_I/h_N>1$). However, it is also possible that a non-optimal imposition of short-selling bans can have a negative effect on both the left tail of returns and on liquidity. We leave the study of this trade-off of regulators for future work. 

\subsection{Discussion of the model}\label{sec:discmodel}
Our model, despite being highly stylized, offers a number of testable predictions. However, its simplicity also entails certain limitations that are worth discussing. First, our model is static and therefore designed to capture the short-term, transitory effects of short-selling bans during episodes of heightened uncertainty, such as the early stages of the COVID-19 crisis, when market participants were still learning about the economic consequences of the shock and the terminal value of assets was not yet fully understood. In this environment, prices adjust progressively through trading, and short-selling bans may affect economically meaningful intermediate-period outcomes by limiting destabilizing price spirals while uncertainty about fundamentals remains unresolved. As such, the model’s predictions apply to intermediate trading periods and do not extend to the final resolution of uncertainty, when the asset’s payoff is realized.  Similarly, we measure liquidity conditions using average bid-ask spreads but the ban may have a non-trivial effect on the dynamics of spreads over time. Moreover, since we only consider a market with a single risky asset, our cross-sectional predictions do not take into account the interactions between the returns of different assets and the changes in investor portfolios.  We leave such extensions for future work, where one can also study the effects of lifting the bans, distinguish between short-term and longer-term effects, and make further inferences concerning the differential effect of bans on various stocks. 

Second, we exogenously assume that the objective of the {\it regulator} is to avert a sharp decline in prices in the short term. Although this is consistent with the goal of regulators to ensure financial stability and maintain market confidence, in practice, regulators may also consider the effects of short-selling bans on liquidity and price informativeness. Third, we have implicitly made the assumption that the model parameters remain unchanged after the introduction of the ban. This is a simplifying assumption, but it could have important implications if some of these parameters change endogenously. For example, the decision to impose a ban could change the incentive of investors to acquire information about a stock and could, thus, affect the parameters $h_I, h_N$ of the model \citep{Dixon2021}. Finally, in the model, there can be instances where no trade takes place. In these cases, we assume that the price of the asset is equal to the updated expectation of the payoff from the perspective of {\it market makers}. However, in the data,  we only observe transaction prices and to avoid a censored sample bias problem \citep[see, for example, Section 5.3 of][]{Diamond1987}, we exclude from our empirical analysis all micro and nano-cap stocks, which may be less actively traded.

\section{Data Description and Preliminary Analysis}\label{sec:bans_data}
In this section, we begin by reviewing the European regulation on short selling and how restrictions were implemented in some European countries during the COVID-19 outbreak. Next, we describe our datasets before proceeding with a preliminary empirical assessment.

\subsection{Short-Selling Bans in European Markets}\label{subsec:short_selling_ban_periods}
In response to the global financial crisis that peaked with the collapse of Lehman Brothers, supervisory authorities around the world enacted various emergency measures limiting the short sale of certain financial instruments. At the time, the European Union lacked a common regulatory framework, and its member states adopted different short-selling policies to halt the downward spiral in securities prices. The need to implement a harmonized legislation to regulate short-selling activity led the European Union to introduce the Short Selling Regulation in November 2012. This regime applies to any legal person undertaking the short sales of shares, sovereign debt, sovereign credit default swap, and related instruments traded on a trading venue located in the European Economic Area, and gives a wide set of supervisory powers to the European Securities and Markets Authority (ESMA). However, national regulators have the power to impose temporary short selling restrictions on any financial instrument, if there is a serious threat to financial stability or to market confidence.\footnote{The UK was subject to this framework until December 2020, when the Brexit transition period ended.} 

With the outbreak of COVID-19, European equity markets experienced a substantial amount of volatility and a sharp price decline. Under these exceptional circumstances, six European regulatory authorities notified ESMA of their intention to temporarily outlaw the short sales of stocks in a quasi-synchronized fashion. On 17 March 2020, Belgium, France and Spain exercised their right under Article 23 of the Short Selling Regulation and decided to introduce a temporary ban on taking or increasing net short positions with respect to all shares admitted to their trading venues. Identical measures were then adopted by Greece, Italy, and Austria, effective from March 18, 2020.\footnote{{On the same day, the European Central Bank also launched a $\text{\euro}750$ billion temporary asset purchase programme (the Pandemic Emergency Purchase Programme (PEPP)), which focused on the fixed income market.}} Initially, the bans were introduced for a period of approximately one month, with the exception of Italy, which proposed a ban of three months. These emergency measures were deemed appropriate by ESMA given the adverse situation linked to the COVID-19 outbreak. However, on 15 April 2020, Austria, Belgium, France, Greece, and Spain informed ESMA of their intention to extend the ban on short sales for another month. ESMA agreed on the renewal, and the bans remained in place until May 18, 2020.  Meanwhile, the Italian Authority also chose to terminate the short-selling restriction on May 18, 2020, in accordance with the decisions made by the other European countries. The dates on these decisions have been gathered from the websites of the European Regulators. \mytabIA{tab:timeline_restrictions} provides a description of these actions and links to the relevant websites. 

Overall,  while the proposed duration of short-selling bans varied slightly, all countries eventually decided to lift the short-selling restrictions on the same date. The scope of the bans applied to any natural or legal person, regardless of where they were located, and covered all stocks traded in cash and derivatives markets, including American Depository Receipts and bearish intraday operations. The prohibitions did not apply to market-making activities or trading in index-related instruments. Exceptions also included convertible bond arbitrage with a delta-neutral structure and short positions hedged by a purchase that is equivalent in terms of subscription rights. Index-related instruments in which restricted shares represented more than a given country-specific threshold were also exempted. Initially, the thresholds were 20\% for Belgium, Greece, and Italy, and 50\% for France and Spain. From April 15 onward, a uniform threshold of 50\% was adopted for all countries.

We test the predictions of our model by exploiting a quasi-natural experiment associated with the introduction of temporary short-selling bans in Europe. In our empirical analysis, the  pre-intervention period runs from  February 17 to March 16, 2020, whereas the post-intervention period goes from March 17 to April 15, 2020. Both periods are intended to capture a month of observations before and after the regulatory shock. As shown in the robustness section, our results remain both quantitatively and qualitatively similar if we instead employ two months of observations before and after the regulatory change, i.e., from  January 17 to March 16, 2020, for the pre-treatment period and from March 17 to May 15, 2020 for the post-treatment period. Additionally, our treatment group will consist of individual stocks traded in Austria, Belgium, France, Greece, Italy, and Spain, while the control group will include individual stocks traded in other European countries. While the selection within each block is not random, in our robustness section, we will further employ firm characteristics and match stocks between treated and untreated countries so that they can be regarded as largely comparable. This is an important exercise, as causal inference could be compromised if the group of treated stocks was systematically different from the group of control stocks.

\subsection{Data on European Stock Markets}
We collect stock market data for 17 European markets, that is, Austria, Belgium, Denmark, Finland, France, Germany, Greece, Ireland, Italy, Netherlands, Norway, Poland, Portugal, Spain, Sweden, Switzerland, and the United Kingdom, from Datastream. Our dataset consists of daily bid and ask prices and total return indices between January 2 and June 2, 2020. After dropping micro and nano-cap stocks (i.e., stocks with a market capitalization below \$250 million) and removing observations with a negative spread between bid and ask prices, we end up with a sample of 1,922 stocks corresponding to more than 200 thousand daily observations. We convert all prices to US dollars using daily spot exchange rates sourced from the same database, to ensure comparability across markets.
\begin{center}
	\textsc{\mytab{tab:data_description} about here}
\end{center}
\mytab{tab:data_description} summarizes the key aspects of this dataset at the country level and shows that, after our data-cleaning, the stock markets with the largest number of stocks are located in the United Kingdom (448 stocks), France (225 stocks), and Germany (209 stocks). In contrast, the stock markets of Portugal (15 stocks), Ireland (20 stocks), and Greece (28 stocks) have the fewest number of stocks overall. Further, we categorize companies based on their market capitalization and distinguish between small-cap (market capitalization between 250 million and 2 billion dollars), mid-cap (market capitalization between 2 billion and 10 billion dollars), and large-cap companies (market capitalization of at least 10 billion dollars). We find that 59\% of the stocks are small-cap companies, while the remaining 41\% of the stocks are either mid-cap or large-cap companies. 

\subsection{Other Data}
We combine our stock market data with other datasets. First, we collect data on institutional ownership, i.e., the percentage of shares outstanding held by institutions, as of December 2019 from Bloomberg. Institutions include 13Fs, US and international mutual funds, schedule Ds (US insurance companies), institutional stake holdings that appear on the aggregate level, and other holdings collected by Bloomberg. When using institutional ownership data, we remove securities for which institutional ownership is either negative or greater than 100\%. These are potentially erroneous observations stemming from reporting lags or double counting due to short selling. 
Second, to capture the causal relationship between short-selling bans and returns or liquidity, we need to ensure that unobserved time-varying factors do not cloud the interpretation of our results. As noted by \citet{Beber2013}, sovereign credit default swaps (CDS) may correlate with the regulator's decision to impose short-selling restrictions. Thus, we obtain sovereign CDS spreads from Datastream and use them as a key control variable in all our empirical exercises. Finally, we retrieve the Industry Classification Benchmark (ICB) code for the stocks in our sample via {\it Bloomberg}.

\subsection{Summary Statistics}
Before taking the predictions of our model to the data, we offer a preliminary analysis on bid-ask spreads and stock returns for countries that adopted bans, as well as countries that did not introduce bans, before and after the introduction of short-selling restrictions. As described in \mysec{subsec:short_selling_ban_periods}, the pre-ban period runs between February 17 and March 16 2020, whereas the post-ban window spans the period from March 17 to April 15, 2020. Panel A of \myfig{fig:summary_bar_chart} shows the average bid-ask spreads and documents a major increase during the post-ban period relative to the pre-ban period, especially in countries subject to short-selling bans. For instance, during the pre-ban period, the average bid-ask spread was about $0.72\%$ in countries that implemented restrictions and $0.62\%$ in other countries. Moving to the post-ban period, the average bid-ask spread reached $1.09\%$ in countries with bans and $0.87\%$ in countries without bans. In relative terms, the average bid-ask spreads increased by more than $50\%$ in countries with bans and less than $40\%$ in countries without bans. We also report country-level statistics for bid-ask spreads in \mytabIA{tab:data_summary_bas} and show that bid-ask spreads increased in all countries, with the exception of Poland.
\begin{center}
	\textsc{\myfig{fig:summary_bar_chart} about here} 
\end{center}
Panel B of \myfig{fig:summary_bar_chart} plots the average stock returns before and after the short-selling bans were introduced, and shows virtually no differences between countries with restrictions and countries without. Specifically, the average return was $-1.94\%$ ($-1.98\%$) per day in countries with (without) bans during the pre-ban period, and  $0.61\%$ ($0.65\%$) per day in countries with (without) bans during the post-ban period. We also report country-level statistics for stock returns in \mytabIA{tab:data_summary_returns} but observe no differences across countries. Taken together, these figures are in line with the findings of \citet{Beber2013} that short-selling restrictions are generally detrimental to liquidity and do not support the average level of prices. The next section tests our predictions using difference-in-differences regressions.

\section{Our Main Results}\label{sec:bans_results}
This section provides an empirical evaluation of our model's predictions using the introduction of short-sale bans in certain European countries during the COVID-19 pandemic as a quasi-experimental exercise. In particular, we assess the impact of short-selling bans on market liquidity and stock returns while conditioning on the heterogeneous effect of institutional ownership.

\subsection{Short-Selling Bans and Institutional Ownership}
We start our empirical investigation with \myhyp{h:imposingban}, which suggests that the regulator is more likely to impose a short-selling ban to avert a sharp drop in prices when the fraction of {\it informed} investors holding the stock is low. As explained in \mysec{sec:hypotheses}, moreover, we quantify the share of sophisticated stock owners using stock-level institutional ownership. 
\begin{center}
	\textsc{\myfig{fig:institutional_ownership} about here} 
\end{center}
Panel A of \myfig{fig:institutional_ownership} takes this prediction to the data by plotting the average stock-level institutional ownership by country as of December 2019. We observe that countries that imposed short-sale bans during the COVID-19 pandemic are on the lower end of the institutional ownership compared to countries that did not impose any restrictions.\footnote{There are three exceptions to this general observation: Switzerland, Germany, and Denmark. We reviewed the quality and sources of institutional ownership data in these three countries, and it seems comparable to that of the rest of the countries. Therefore, the idiosyncrasies related to the collection of data in these countries do not appear to be obvious explanations for these exceptions.} In Panel B of \myfig{fig:institutional_ownership}, we pool together all countries that restricted short-selling activity and compare their average stock-level institutional ownership with the same quantity for countries that did not impose any limitations on short sales. We show that countries with bans had substantially lower average institutional ownership than otherwise countries. Taken together, we view \myfig{fig:institutional_ownership} as suggestive evidence in support of \myhyp{h:imposingban}, maintaining that institutional ownership is an important factor in the decision-making process of regulators when considering imposing restrictions on short-selling activity.

\subsection{Short-Selling Bans and Market Liquidity}
Next, we study the effect of short-sale restrictions on stock market liquidity using bid-ask spreads, following the seminal paper of \citet{Beber2013}. To avoid replicating existing results, we briefly summarize our findings here and relegate the full analysis to \mysecIA{app:supplement_analysis}. To assess the impact of the ban, we calculate the average daily bid-ask spread over a window that covers one month before and one month after the introduction of short-selling bans. While liquidity was lower during the ban than in the preceding period as previously discussed,  we cannot conclude that the imposition of short-selling bans \emph{caused} the rise in bid-ask spreads. In \myfigIA{fig:DiD_bas} we can see that the bid-ask spreads for the two groups of countries, those that imposed bans and those that refrained from the policy, moved together prior to the implementation of short-selling bans.  However, following the introduction of these restrictions, bid-ask spreads increased more sharply in the countries that imposed bans, suggesting a potential adverse effect of short-selling bans on market liquidity, while after the bans are lifted on May 18, 2020, the divergence in bid-ask spreads gradually narrows.

We then formalize our findings in \mytabIA{tab:DiD_bas} by reporting the results of a difference-in-differences regression that estimates how short-selling bans differentially affected market liquidity. Average bid-ask spreads increased by 25 basis points during the short-selling ban period, while the spreads of stocks that were affected by the ban experienced an additional widening of 11.6 basis points.  Our results are consistent across many specifications (e.g., with time and stock fixed effects, and controlling for CDS spreads), and they thus support \myhyp{h:liquidity} and are consistent with the findings of \citet{Beber2013} during the global financial crisis.

A key concern with the estimates in \mytabIA{tab:DiD_bas} is the potential endogeneity of short-selling bans. If policymakers tend to impose bans during periods of rising volatility and falling liquidity, the relationship between short-selling bans and market liquidity could not be interpreted as a causal relationship. To address this concern, we employ an instrumental variables strategy, where the first stage determines the likelihood of a ban and the second stage its effect on liquidity. Similar to \citet{Beber2013}, we employ the monthly average sovereign CDS spread and the monthly Financial Stress Index of \citet{Duprey2017} in logs as instruments. In \mytabIA{tab:IV_bas}, the first-stage regression shows that our instruments exhibit strong explanatory power. In the second stage regression, we regress percentage bid-ask spreads on the instrumented short-selling ban indicator, and we find that the coefficient on the estimated ban, ranging between $0.126$ and $0.135$,  is statistically significant and closely aligned with the magnitudes reported in \mytabIA{tab:DiD_bas}.

\subsection{The Role of Institutional Ownership for Market Liquidity}
Next, we test the prediction of the second part of \myhyp{h:liquidity} according to which the negative effect of short-selling bans on liquidity is larger on stocks with higher institutional ownership. That is, when short-selling is prohibited, stocks with higher institutional ownership will experience a greater deterioration in liquidity, manifested in larger bid-ask spreads. This is because the adverse selection facing {\it market makers} will be greater when more {\it informed} investors own the stock and can thus submit a sell order despite the short-selling bans. To test this hypothesis, we split our sample into stocks with low/high institutional ownership and estimate the difference-in-differences regressions of \mytabIA{tab:DiD_bas} in these two subsamples.
We choose the bottom tercile (i.e., institutional ownership up to $41\%$) for stocks with low institutional ownership, and the top tercile (i.e., institutional ownership above $68\%$) for stocks with high institutional ownership. 
\begin{center}
	\textsc{\mytab{tab:DiD_bas_IO} about here} 
\end{center}
The results, presented in \mytab{tab:DiD_bas_IO}, are indeed in line with the prediction of our model. Standard errors are clustered by stock and time (calendar date) dimensions. Bid-ask spreads of stocks with low institutional ownership increase by an additional 8 basis points on average as a result of the short-selling bans, whereas bid-ask spreads of stocks with high institutional ownership increase by 23-29 basis points. Moreover, the impact of short-selling bans on the liquidity of stocks with low institutional ownership is not statistically significant at any of the conventional levels, whereas the estimated impact of 22-29 basis points on stocks with high institutional ownership is statistically significant at the 1\% level across all specifications. We conclude that short-selling bans have a negative impact on liquidity, especially for stocks with high institutional ownership. In \mytabIA{tab:DiD_bas_IO_sctcls}, we also report standard errors clustered by stock and country-time dimensions, which account for potential cross-sectional correlation within each country on a given day. The results remain qualitatively unchanged, thus confirming our findings.

\subsection{Short-Selling Bans and Stock Returns}
After examining the impact of short-selling bans on market liquidity, we move to assessing their impact on stock returns. Building on \myhyp{h:support}, we focus on the mean, median, and the maximum drawdown as an intuitive criterion to compare the behaviour of stock returns on the left tail of their distribution. \myfigIA{fig:DiD_mean} visualizes the cumulative mean stock return in countries that imposed bans versus countries that did not follow the same policy. As shown in this graph, on average, stock returns closely track each other during the ban period. \myfigIA{fig:DiD_median}, moreover, plots the cumulative median stock return for the same groups of countries. Unlike the mean returns, we uncover a decrease in the median stock return for countries subject to bans relative to unconstrained countries. 

Our theoretical model suggests that even though short-selling bans may not be effective in supporting the average level of prices, they may be effective in supporting the left tail of the distribution of prices. Of course, this is particularly important during a financial crisis, when a precipitous fall in prices may raise concerns about financial stability. For example, \citet{Brunnermeier2014} show that when a financial institution is sufficiently close to its leverage constraints, a sharp fall in its stock price may trigger a run on the bank. Naturally, regulators may be inclined to impose temporary short-selling bans to prevent that from happening and avert a more generalized market panic.\footnote{Moreover, \citet{geraci2018short} find that, during "extreme events", the correlation between short-selling activity and negative price changes becomes significantly stronger.}
\begin{center}
	\textsc{\myfig{fig:DiD_mdd} about here} 
\end{center}
In \myfig{fig:DiD_mdd}, we plot the cross-sectional average of the maximum drawdown (i.e., the maximum cumulative decline from a peak to a trough) in the pre-ban and post-ban periods for the countries that imposed short-selling bans and for those that did not. As this graph shows, the maximum drawdown is similar for the two sets of countries in the pre-ban period, but once short-selling bans are imposed, stock declines are less pronounced in countries with bans relative to countries without short-selling restrictions. To sum up, in line with our theoretical prediction, the preliminary evidence described above suggests that mean returns remained qualitatively similar while median returns decreased during ban periods, thus confirming the findings of \citet{Beber2013}. Additionally, short-selling restrictions were able to support the left tail of stock returns. While the existing evidence on the impact of bans on skewness is far from being conclusive, a study by \citet{Bris2007} reveals that stock returns experience significantly less negative skewness in markets with restrictions on short selling. Their results, however, are primarily driven by market-level evidence, with limited support at the individual stock level, whereas our focus is on identifying the mechanism that leads to a reduction in the likelihood of extreme left tail returns.
\begin{center}
	\textsc{\mytab{tab:DiD_distribution} about here} 
\end{center}
To formally test the predictions of \myhyp{h:support}, we first calculate the mean, median, volatility, and maximum drawdown for each stock return over the pre-ban period and the post-ban period. For each metric, we thus end up with a balanced panel of two observations per stock, i.e., the first one calculated using a month of observations before the enactment of short-selling bans and the second one based on the subsequent month of observations. We then employ standard difference-in-differences regressions, while controlling for sovereign CDS spreads and firm size. We report our estimates in \mytab{tab:DiD_distribution}, where the variable of interest is the interaction term denoted by \emph{Ban} $\times$ \emph{Country}, which is equal to one for stocks subject to short-selling bans during the post-ban period. Relative to stocks traded in countries with no bans, we observe that the mean returns of stocks subject to bans were 7 basis points lower (statistically insignificant), the median returns were 22 basis points lower (statistically significant at the 1\% level), and the maximum drawdowns were 310 basis points higher (significant at the 1\% level). Note here that an increase in the maximum drawdown should be interpreted as support for the left tail of the return distribution, and the magnitude of the coefficient suggests that our results are also economically large. Thus, putting it all together, short-selling bans seem to avert large declines in prices at the expense of marginally lower mean/median returns, in addition to a deterioration of liquidity.

\subsubsection{The Role of Institutional Ownership for Stock Returns}
Finally, we test the remaining part of \myhyp{h:support}, according to which the effectiveness of short-selling bans in limiting extreme negative outcomes is inversely proportional to the institutional ownership of the affected stock. Institutional ownership is used as a proxy for the fraction of informed traders who own the stock. If short-selling is not allowed, then the fraction of informed traders owning the stock affects the distribution of prices in two ways. On the one hand, the lower this is, the fewer sale orders will be submitted (these will be hidden under a veil of ``no order'' events), as potential investors with negative information will be prohibited from submitting a short-sale order. On the other hand, this fraction determines the adverse selection in the market. When the fraction of informed (relative to noise) traders who own the stock is low, market makers are more likely to perceive sell orders as initiated by noise traders. Thus, they would be less aggressive in revising their expectation of fundamentals and skewing their bid lower in defense against potential adverse selection.
\begin{center}
	\textsc{\mytab{tab:DiD_distribution_IO} about here} 
\end{center}
To test this prediction,  we consider two groups of stocks, i.e., one with low (bottom tercile) and another with high (top tercile) institutional ownership. We, then, repeat the difference-in-differences regressions of \mytab{tab:DiD_distribution} in these two subsamples. As predicted by the model, the results presented in \mytab{tab:DiD_distribution_IO} confirm that short-selling bans are more effective in supporting the left tail of stocks with lower institutional ownership. Specifically, we estimate that in the sample of stocks with low institutional ownership, the maximum drawdown for stocks subject to short-selling bans was 4.3\% above the maximum drawdown experienced by stocks without short-selling restrictions, whereas the effect was more muted at around 2\% wedge in the maximum drawdown between the groups in the subsample of stocks with high institutional ownership.

Overall, our empirical findings support the view that short-selling bans can, under certain conditions, reduce the likelihood of a sharp decline in prices, but this comes at the cost of a deterioration in liquidity. It has often been cited that regulators' motivation to impose short-selling bans is to restore financial stability and market confidence.\footnote{For example, Robert Ophele, the Chairman of the French Authority, stated in an interview with Bloomberg on May 18, 2020, that \textit{``The European regulation is very clear: This restriction [the short-selling ban] is possible in case of adverse developments which constitute a serious threat to financial stability or market confidence. This restriction should be temporary, and taken in order to prevent the disorderly decline in the price of financial instruments $\ldots$''}}  The point of this paper is to shed light to the relevant trade-offs, and employ the regulators with the empirical evidence and quantitative estimates they need in order to make informed decisions when considering the imposition of short-selling bans.

\section{Robustness Analysis} \label{sec:robustness}
In this section, we run a battery of robustness tests to corroborate our findings using alternative specifications. We find that results remain unchanged, thus, confirming our conclusions from the the previous section.

\subsection{Samples matched by Firms Characteristics}
One may be concerned that stocks traded in the treatment countries are not necessarily comparable to those traded in the control countries. This is an important concern to address, as causal inference could be compromised if the group of treated stocks was systematically different from the group of control stocks. Concretely, each stock subject to short-selling bans is matched to an unrestricted stock that is closest in terms of market capitalization and industry classification according to the Industry Classification Benchmark (ICB) code. We then repeat the exercises presented earlier in \mytab{tab:DiD_bas_IO} and \mytab{tab:DiD_distribution_IO}, respectively, using our subsample of matched stocks. 
\begin{center}
	\textsc{\mytab{tab:DiD_bas_IO_matched} and \mytab{tab:DiD_distribution_IO_matched} about here} 
\end{center}
We report our estimates for bid-ask spreads in \mytab{tab:DiD_bas_IO_matched}, and for the distribution of stock returns in \mytab{tab:DiD_distribution_IO_matched}. We find no qualitative differences with our main results, thus suggesting that the composition of stocks in the treatment and control groups does not materially affect our results. Furthermore, we replicate \myfig{fig:DiD_mdd} using our subsample of matched stocks in \myfigIA{fig:DiD_mdd_matched} but find qualitatively similar results. 

\subsection{Placebo Tests}
The difference-in-differences regressions rely on the assumption that outcomes have equal trends between affected and unaffected stocks in the absence of treatment, i.e., confounders varying across groups are time-invariant and confounders varying across time are group invariant. When the underlying common trend assumption is invalid, the difference-in-differences estimates are biased, since the trend for the control group is not a valid estimate of the counterfactual trend that we would have observed for the treatment group in the absence of the regulatory change. We test the validity of the underlying assumption of equal trends by repeating the exercises reported in \mytab{tab:DiD_bas_IO} and \mytab{tab:DiD_distribution_IO} using placebo outcomes. Put differently, we keep the actual regulatory shock on March 17, 2020 but only employ countries not subject to short-selling bans. 
\begin{center}
	\textsc{\mytab{tab:PlaceboOutcome_bas_IO_matched} and \mytab{tab:PlaceboOutcome_distribution_IO_matched} about here} 
\end{center}
We arrange our control group of countries in alphabetic order so that stocks traded in Denmark, Finland, Germany, Ireland, Netherlands, and Norway belong to the placebo group, and stocks traded in Poland, Portugal, Sweden, Switzerland, and the United Kingdom are part of the control group. We also match stocks in terms of market capitalization and industry classification. We report our evidence for bid-ask spreads in \mytab{tab:PlaceboOutcome_bas_IO_matched}, and for the distribution of stock returns in \mytab{tab:PlaceboOutcome_distribution_IO_matched}. We find that our difference-in-differences estimates are statistically and economically insignificant for bid-ask spreads, means, and maximum drawdowns. We thus mitigate concerns about the validity of the parallel trend assumption.

\subsection{Extending the Pre-Ban and Post-Ban Window}
Our analysis employs a one-month window for the pre-treatment and post-treatment period, respectively. In \mytabIAs{tab:DiD_bas_IO_60days}--\ref{tab:DiD_distribution_IO_60days}, we repeat our main exercises while extending the window for the pre-treatment and post-treatment period by an additional month, i.e., from January 17 to March 16, 2020 for the pre-ban sample and from March 17 to May 15, 2020 for the post-ban sample.  We obtain qualitatively and quantitatively similar results, thus illustrating that our choice of a one-month window around the imposition of short-selling bans is innocuous.

\subsection{Downside Risk}
Finally, \myfigIA{fig:DiD_down} replicates the content of \myfig{fig:DiD_mdd} while replacing the maximum drawdown with downside risk, which we measure using the volatility of daily negative returns.  Consistent with our findings on the maximum drawdown, we show that the average downside risk was largely identical before the regulatory shocks for treated and untreated shocks. After the introduction of short-selling bans, we document that downside risk was substantially higher for untreated stocks relative to treated stocks, consistent with our model's prediction. We also use samples of stocks matched by market capitalization and industry classification in \myfigIA{fig:DiD_down_matched} but the evidence is virtually identical.

\section{Conclusion} \label{sec:bans_conclusion}
Since the seminal work of \citet{Beber2013}, the general wisdom is that short-selling bans have a detrimental effect on market liquidity and fail to support prices. Yet regulators in six European countries (i.e., Austria, Belgium, France, Greece, Italy, and Spain) decided to impose a two-month ban on new short sales (in March 2020) in response to the financial crisis caused by the COVID-19 outbreak. In this paper, we examine the costs and benefits of restricting short sales in times of crisis, and provide a framework that can help regulators determine whether, when, and how to implement such restrictions.

In particular, we build a theoretical model endogenizing the regulator's decision to impose a ban on short sales and derive testable predictions for liquidity and prices, which we then verify empirically.  Our model extends \citet{Diamond1987} by introducing a regulator whose goal is to avert a sharp decline in prices, and we show that the effectiveness of short-selling bans depends on the relative ratio of informed to noise traders who own the stock. We identify institutional ownership as a proxy for this model parameter and exploit cross-sectional variation arising from the European 2020 short-selling bans to test the model's predictions. Consistent with the model, we find that tail risk was reduced in countries that implemented short-selling bans and that this effect was more pronounced in stocks with low institutional ownership. Moreover, we also confirm prior evidence that bans were harmful for liquidity and failed to support prices. Overall, our findings highlight the trade-off facing policy makers considering short-selling bans: hampering liquidity in order to mitigate tail risks in the distribution of stock returns.

\begin{spacing}{1}
	\setlength{\bibsep}{\baselineskip}
	\bibliographystyle{aer}
	\bibliography{biblio}
\end{spacing}

\clearpage
\begin{figure}
	\begin{center}
		\includegraphics[scale = 0.90, angle = 0, 
		trim = 00mm 00mm 00mm 00mm]{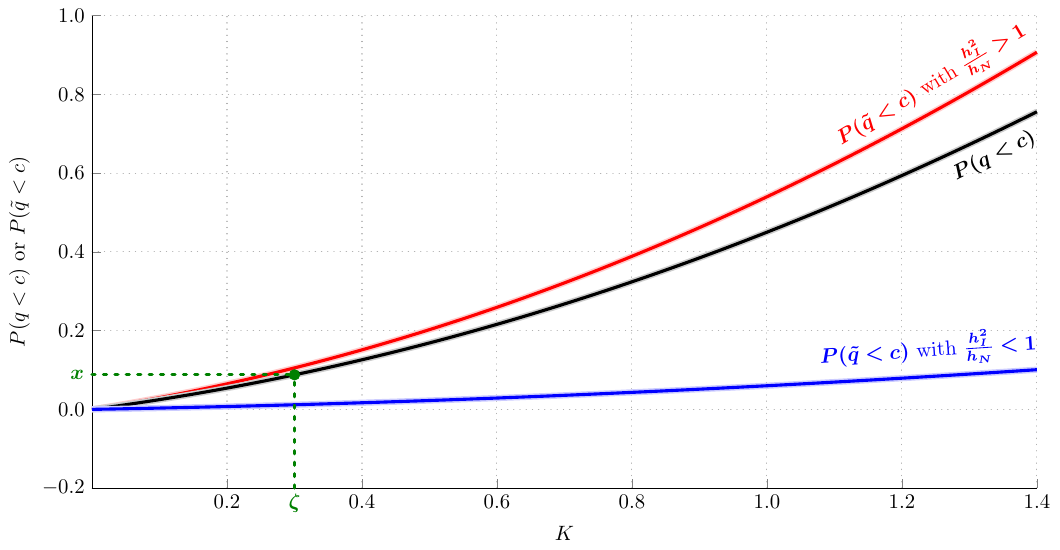}
		\vspace{0.5cm}
		\caption{\bf Simulated Equilibrium Prices under Different Scenarios} 
		\label{fig:probability_q}
	\end{center}
	
	\begin{footnotesize}
		This figure displays, from  the {\it regulator}'s perspective, the probability that the equilibrium price $q$ (short selling is allowed) or $\tilde{q}$ (short selling is prohibited) drops below the threshold $c$ using simulated data. The former scenario is denoted as $P(q<c)$, whereas the latter case is represented as $P(\tilde{q}<c)$. The black line corresponds to the baseline case with unrestricted short selling, the red line denotes the alternative scenario with restricted short-selling while setting $\smash{h_I^2/h_N > 1}$, and the blue line denotes the alternative scenario with restricted short-selling while setting  $\smash{h_I^2/h_N < 1}$. $h_I$ denotes the probability that an {\it informed} trader owns the stock, whereas $h_N$ indicates  the probability that a {\it noise} trader owns the stock. $\smash{K= c\alpha(1-p)/((1-\alpha)(p-c))}$ is an increasing function of $c$. The simulation assumes $\alpha=0.5$, $p= 0.5$, $g= 0.9$, $h_N=0.3$, and $h_I = 0.6$ (so that $\smash{h_I^2/h_N >1}$) or $h_I = 0.2$ (so that $\smash{h_I^2/h_N < 1}$). A description of these parameters is presented in \mytabIA{tab:summary_notation}.		
		
	\end{footnotesize}

\end{figure}

\begin{landscape}
	\begin{figure}

		\begin{center}		
			\includegraphics[scale=0.65, angle = 0, trim = 00mm 00mm 00mm 00mm]{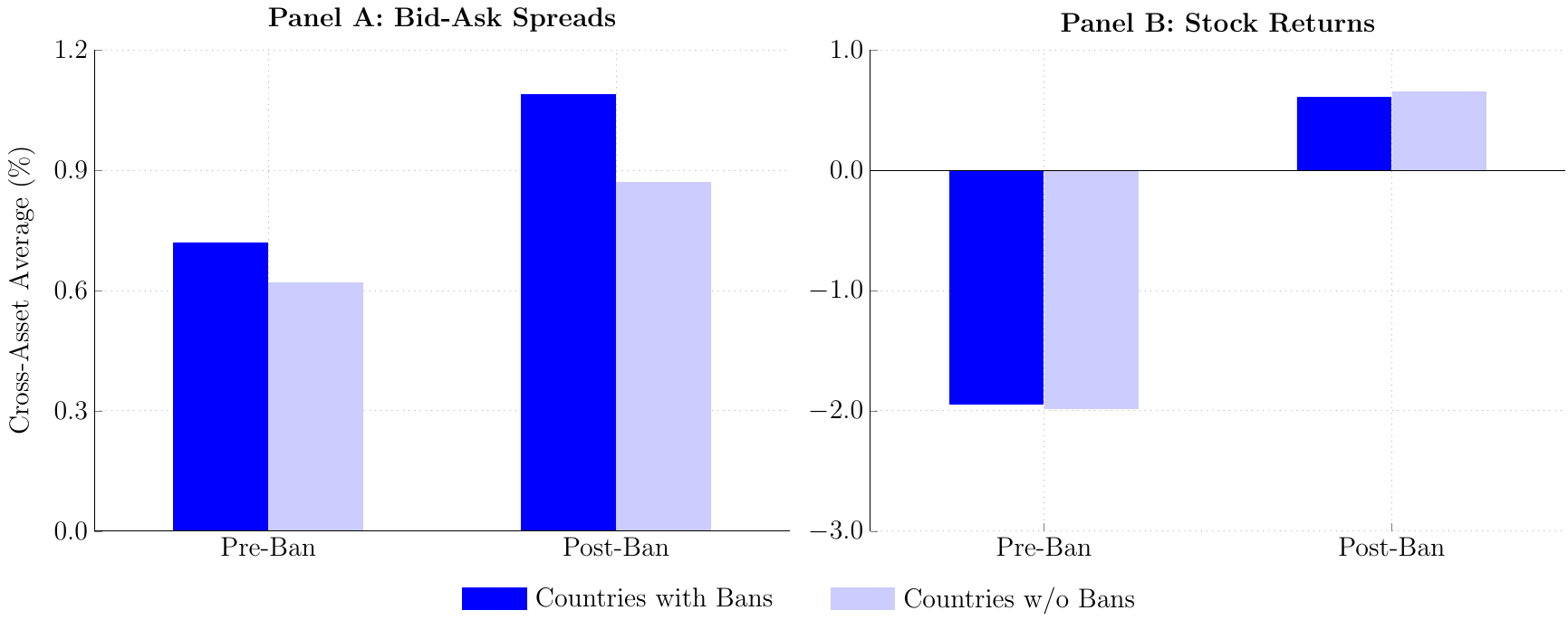}
			\vspace{0.5cm}
			\caption{\bf Bid-Ask Spreads and Stock Returns} 
			\label{fig:summary_bar_chart}
		\end{center}
		
		\begin{footnotesize}
		This figure shows average bid-ask spreads (Panel A) and average stock returns (Panel B) for European stock markets around the introduction of temporary short-selling bans during the COVID-19 pandemic. The group of countries that restricted short sales includes Austria, Belgium, France, Greece, Italy, and Spain. The group of countries that allowed short sales consists of Denmark, Finland, Germany, Ireland, Netherlands, Norway,  Poland, Portugal, Sweden, Switzerland, and the United Kingdom. The pre-ban (post-ban) period covers one month and runs between February 17 and March 16, 2020 (March 17 and April 15, 2020). Bid-ask spreads are computed using daily bid and ask prices before taking the mean value across all stocks in countries with and without short-selling bans, respectively. Returns are based on daily total return indices expressed in dollar terms before taking the mean value across all stocks in countries with and without short-selling bans, respectively. Both bid-ask spreads and returns are expressed in percentages. Data are collected from {\it Datastream}.		
		
		\end{footnotesize}
	\end{figure}
\end{landscape}

\begin{landscape}
	\begin{figure}

		\begin{center}		
			\includegraphics[scale=0.60, angle = 0, trim = 00mm 00mm 00mm 00mm]{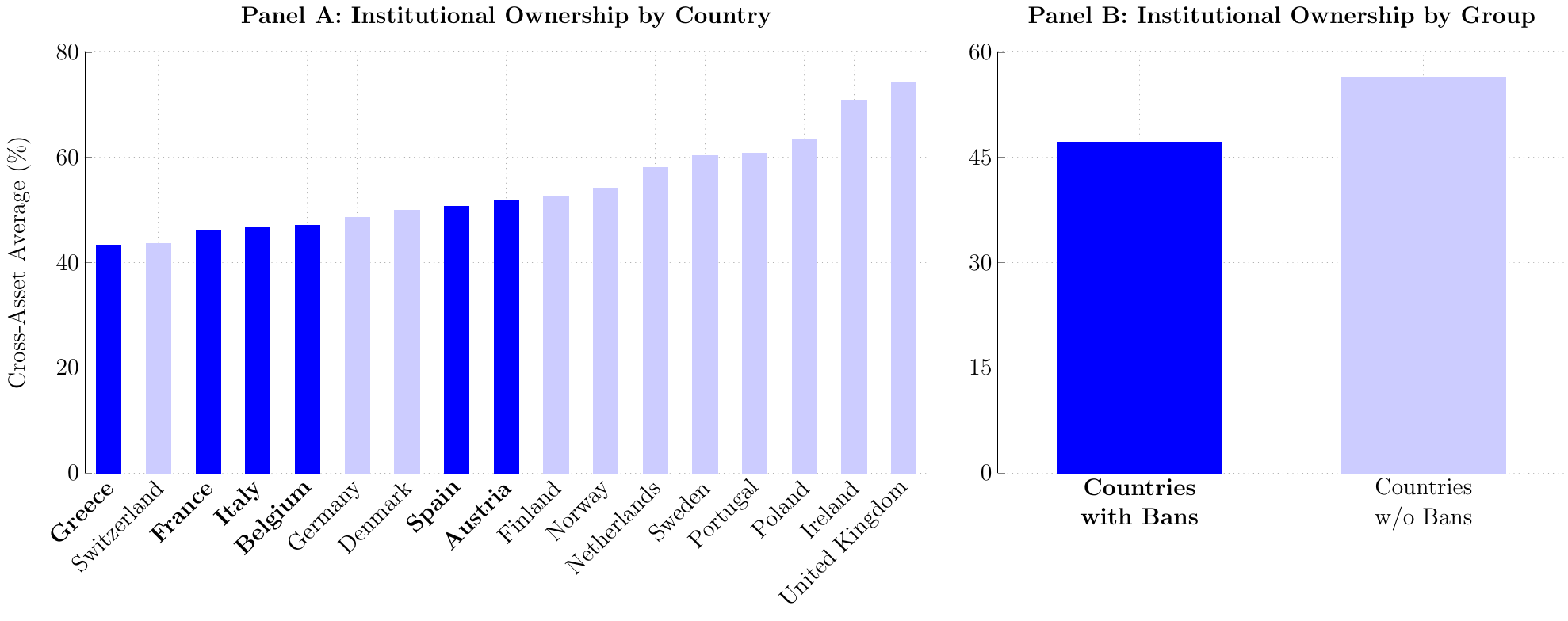}
			\vspace{0.5cm}
			\caption{\bf Institutional Ownership} 
			\label{fig:institutional_ownership}
		\end{center}
		
		\begin{footnotesize}
		This figure shows the average stock-level institutional ownership as of December 2019 for European stock markets. Austria, Belgium, France, Greece, Italy, and Spain denote the countries that introduced temporary short-selling bans during the COVID-19 pandemic. Denmark, Finland, Germany, Ireland, Netherlands, Norway,  Poland, Portugal, Sweden, Switzerland, and the United Kingdom refer to countries that allowed short sales. Panel A presents the average institutional ownership of all stocks in each country whereas Panel B displays the average institutional ownership across all stocks traded in countries with and without short-selling bans, respectively. Institutional ownerships are all expressed in percentages and collected from {\it Bloomberg}.		
		
		\end{footnotesize}
	\end{figure}
\end{landscape}

\begin{landscape}
\begin{figure}
		\begin{center}
		\includegraphics[scale=1.00, angle = 0, trim = 00mm 00mm 00mm 00mm]{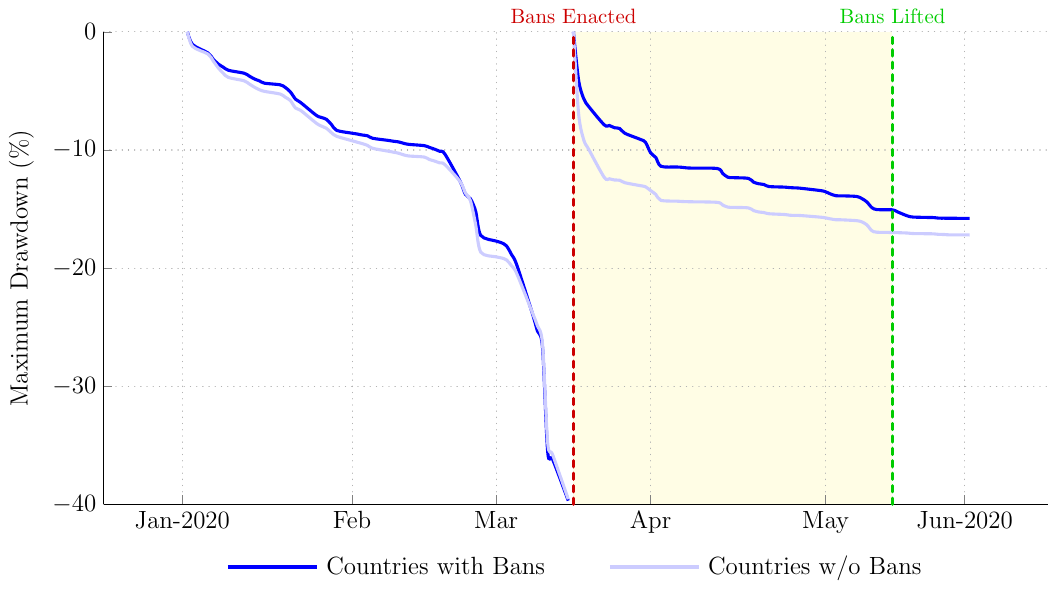}
		\vspace{0.5cm}
		\caption{\bf Maximum Drawdown and Short-Selling Bans} \label{fig:DiD_mdd}
	\end{center}

	\begin{footnotesize}
	This figure shows the average maximum drawdown of stocks traded in European markets around the introduction of  temporary short-selling bans during the COVID-19 pandemic. The group of countries that restricted short sales includes Austria, Belgium, France, Greece, Italy, and Spain, whereas the group of countries that allowed short sales consists of Denmark, Finland, Germany, Ireland, Netherlands, Norway,  Poland, Portugal, Sweden, Switzerland, and the United Kingdom. Short-selling bans were introduced on March 17 and remained in place until May 18, 2020. The maximum drawdown employs an expanding window from January 2 to March 16, 2020, and from March 17 to June 2, 2020, respectively, and is based on daily total return indices expressed in dollar terms before taking the mean value across all stocks in countries with and without short-selling bans, respectively. The sample includes small, mid, and large cap stocks (micro and nano caps are excluded), and runs between January 2 and June 2, 2020. Data are collected from {\it Datastream}.	
	
	\end{footnotesize}

\end{figure}
\end{landscape}

\begin{landscape}
	\begin{table}[ht]
		\caption{\bf Data Description} \label{tab:data_description}
		\begin{footnotesize}
		This table describes stock market data for European venues. The countries highlighted in gray, i.e., Austria, Belgium, France, Greece, Italy, and Spain, introduced temporary short-selling bans between March 17 and May 18, 2020. $\#$Small denotes the number of small-cap firms with an average market capitalization between 250 million and 2 billion US dollars, $\#$Mid refers to the number of mid-cap firms with an average market capitalization between 2 billion and 10 billion US dollars, whereas $\#$Large indicates the number of large-cap firm with an average market capitalization of \$10 billion US dollars or more.  We exclude micro and nano caps from our sample. The sample runs daily between January 2 and June 2, 2020. Data are collected from {\it Datastream}. 		
			
		\end{footnotesize}
		\bigskip

		\newcolumntype{L}[1]{>{\raggedright\let\newline\\\arraybackslash\hspace{0pt}}m{#1}}
		\newcolumntype{C}[1]{>{\centering\let\newline\\\arraybackslash\hspace{0pt}}m{#1}}
		\newcolumntype{R}[1]{>{\raggedleft\let\newline\\\arraybackslash\hspace{0pt}}m{#1}}
		\setlength\extrarowheight{2pt}

		\sisetup{
			input-decimal-markers  =  .,
			input-ignore           = {,},
			table-number-alignment = right,
			group-separator        ={,}, 
			group-four-digits      = true,
			input-symbols         = [()],
			table-text-alignment  = center,
			table-space-text-post = ***,
			table-align-text-post = false,
		}

		\robustify\bfseries
		\centering
		\scalebox{0.90}{
			
			\begin{tabular}{L{3.5cm}
					S[table-column-width = 3.0cm, table-format = 7.0]
					S[table-column-width = 2.2cm, table-format = 6.0]
					S[table-column-width = 2.2cm, table-format = 6.0]
					S[table-column-width = 2.2cm, table-format = 6.0]
					S[table-column-width = 2.2cm, table-format = 6.0]
				}
				\toprule
									   
			   	 &       &       &       \multicolumn{3}{c}{Market Capitalization} \\
			   	 \cmidrule{4-6}
			   	 \multicolumn{1}{l}{ Country}	
			   	 & \multicolumn{1}{c}{$\#$Obs} 
			   	 & \multicolumn{1}{c}{$\#$Stock} 
			   	 & \multicolumn{1}{c}{$\#$Small} 
			   	 & \multicolumn{1}{c}{$\#$Mid}   
			   	 & \multicolumn{1}{c}{$\#$Large}  \\			   	 
			   	 \cmidrule{1-6}

	\rc \bl{Austria} & \bl3,672 & \bl34    & \bl16    & \bl16    & \bl2 \\
	\rc \bl{Belgium} & \bl7,452 & \bl69    & \bl40    & \bl24    & \bl5 \\
	{Denmark} & 5,616 & 52    & 25    & 16    & 11 \\
	{Finland} & 5,800 & 54    & 36    & 11    & 7 \\
	\rc \bl{France} & \bl24,254 & \bl225   & \bl118   & \bl61    & \bl46 \\
	{Germany} & 22,572 & 209   & 107   & 61    & 41 \\
	\rc \bl{Greece} & \bl3,024 & \bl28    & \bl24    & \bl4     &  0\\
	{Ireland} & 2,166 & 20    & 13    & 4     & 3 \\
	\rc \bl{Italy} & \bl12,744 & \bl118   & \bl68    & \bl38    & \bl12 \\
	{Netherlands} & 6,588 & 61    & 28    & 15    & 18 \\
	{Norway} & 9,180 & 85    & 63    & 18    & 4 \\
	{Poland} & 6,372 & 59    & 41    & 18    &  0\\
	{Portugal} & 1,620 & 15    & 9     & 4     & 2 \\
	\rc \bl{Spain} & \bl9,072 & \bl84    & \bl46    & \bl23    & \bl15 \\
	{Sweden} & 19,440 & 180   & 118   & 49    & 13 \\
	{Switzerland} & 19,538 & 181   & 104   & 49    & 28 \\
	{United Kingdom} & 48,308 & 448   & 270   & 134   & 44 \\
	\midrule
	{Total} & 207,418 & 1,922 & 1,126 & 545   & 251 \\ 

	\bottomrule
				
			\end{tabular}
		}
	\end{table}

\end{landscape}

\begin{landscape}
\begin{table}[ht]
	\caption{\bf Bid-Ask Spreads and Institutional Ownership}\label{tab:DiD_bas_IO}
	
	\begin{footnotesize}		
		This table presents difference-in-differences estimates associated with the introduction of temporary short-selling bans in European stock markets during the COVID-19 pandemic. The dependent variable is the percentage bid-ask spread based on daily bid and ask prices for stocks with low institutional ownership (bottom tercile) in Panel A and stocks with high institutional ownership (top tercile) in Panel B, respectively. \emph{Country} is a dummy variable that equals one for the group of countries that introduced short-selling bans (i.e., Austria, Belgium, France, Greece, Italy, and Spain), and zero for the control group of countries (i.e., Denmark, Finland, Germany, Ireland, Netherlands, Norway,  Poland, Portugal, Sweden, Switzerland, and the United Kingdom). \emph{Ban} is a dummy variable that equals one (zero) for a post-treatment (pre-treatment) period of one month that goes from March 17 to April 15, 2020 (February 17 to March 16, 2020). The set of controls includes firm size and sovereign CDS spread. Specifications are complemented with stock and time (calendar date) fixed effects {\it fe}. Standard errors (in parentheses) are clustered by stock and time (calendar date) dimensions. *, **, ***, indicate statistical significance at the 10\%, 5\%, and 1\% level, respectively. The sample includes small, mid, and large cap stocks (micro and nano caps are excluded). Data are collected from {\it Datastream} and {\it Bloomberg}.	
			
   	\end{footnotesize}
	
		\bigskip

	\newcolumntype{L}[1]{>{\raggedright\let\newline\\\arraybackslash\hspace{0pt}}m{#1}}
	\newcolumntype{C}[1]{>{\centering\let\newline\\\arraybackslash\hspace{0pt}}m{#1}}
	\newcolumntype{R}[1]{>{\raggedleft\let\newline\\\arraybackslash\hspace{0pt}}m{#1}}
	\setlength\extrarowheight{0pt}

	\sisetup{
		input-symbols         = [()],
		table-format          = 1.0,
		table-space-text-post = ***,
		table-align-text-post = false,
		table-text-alignment  = center,
		group-digits          = false
	}

	\robustify\bfseries
	\centering
	\scalebox{0.90}{
		
		\begin{tabular}{L{2.9cm}
				S[table-column-width = 1.9cm]
				S[table-column-width = 1.9cm]
				S[table-column-width = 1.9cm]
				S[table-column-width = 1.9cm]
				S[table-column-width = 1.9cm]
				S[table-column-width = 1.9cm]
				S[table-column-width = 1.9cm]
				S[table-column-width = 1.9cm]
			}
	
	\toprule

			& \multicolumn{4}{c}{\bf Panel A: Low Institutional Ownership} 
			& \multicolumn{4}{c}{\bf Panel B: High Institutional Ownership} \\
			\cmidrule(r){2-5} \cmidrule(l){6-9}

		 	& \multicolumn{1}{c}{(1)} 
			& \multicolumn{1}{c}{(2)} 
		    & \multicolumn{1}{c}{(3)} 
		 	& \multicolumn{1}{c}{(4)} 
			& \multicolumn{1}{c}{(5)} 
			& \multicolumn{1}{c}{(6)} 
			& \multicolumn{1}{c}{(7)}
			& \multicolumn{1}{c}{(8)} \\
			
			\midrule
			\addlinespace[5pt]

	\rc{\it Ban $\times$ Country}   & \bl0.079 & \bl0.077 & \bl0.079 & \bl0.058 & \bl0.229\fn{***} & \bl0.285\fn{***} & \bl0.285\fn{***} & \bl0.278\fn{***} \\
	\rc & (0.053) & (0.063) & (0.064) & (0.072) & (0.066) & (0.078) & (0.079) & (0.083) \\
	\addlinespace[5pt]
	
	{\it Ban} & 0.307\fn{***} & 0.310\fn{***} &       &       & 0.206\fn{***} & 0.209\fn{***} &       &  \\
	& (0.072) & (0.077) &       &       & (0.043) & (0.049) &       &  \\
	\addlinespace[5pt]
	
	{\it Country} & 0.063 &       &       &       & 0.181 &       &       &  \\
	& (0.089) &       &       &       & (0.114) &       &       &  \\
	\addlinespace[5pt]

	{\it Constant} & 0.807\fn{***} & 0.833\fn{***} & 0.983\fn{***} & 0.186 & 0.613\fn{***} & 0.653\fn{***} & 0.754\fn{***} & 0.643 \\
	& (0.081) & (0.066) & (0.012) & (1.074) & (0.062) & (0.040) & (0.010) & (0.819) \\
	\addlinespace[5pt]
	
	{$R^2$} & 0.011 & 0.490 & 0.513 & 0.513 & 0.014 & 0.618 & 0.628 & 0.628  \\
	\cmidrule{2-9}
	{\it $\#$ obs} & {20,826} & {20,825} & {20,825} & {20,825} & {20,921} & {20,921} & {20,921} & {20,921}  \\
	\midrule
	{\it Stock fe} &       & {\checkmark}     & {\checkmark}     & {\checkmark}     &       & {\checkmark}     & {\checkmark}     & {\checkmark} \\
	{\it Time fe}  &       &                  & {\checkmark}     & {\checkmark}     &       &                  & {\checkmark}     & {\checkmark} \\
	{\it Controls} &       &                  &                  & {\checkmark}     &       &                  &                  & {\checkmark} \\
	\bottomrule

\end{tabular}%
}

\end{table}
\end{landscape}

\begin{landscape}
\begin{table}[ht]
	\caption{\bf Stock Returns and Short-Selling Bans}\label{tab:DiD_distribution}

	\begin{footnotesize}
	This table presents difference-in-differences estimates associated with the introduction of temporary short-selling bans in European stock markets during the COVID-19 pandemic. The dependent variables are the percentage mean, median, volatility, and maximum drawdown of total return indices in dollar terms based on on a one-month window around the enactment of short-selling bans. The pre-treatment (post-treatment) period runs between February 17 and March 16, 2020 (March 17 and April 15, 2020). \emph{Country} is a dummy variable that equals one for the group of countries that introduced short-selling bans (i.e., Austria, Belgium, France, Greece, Italy, and Spain), and zero for the control group of countries (i.e., Denmark, Finland, Germany, Ireland, Netherlands, Norway,  Poland, Portugal, Sweden, Switzerland, and the United Kingdom).  \emph{Ban} is a dummy variable that equals one (zero) for the post-treatment (pre-treatment) window. The set of controls includes firm size and sovereign CDS spread. Standard errors (in parentheses) are clustered by stock dimension. *, **, ***, indicate statistical significance at the 10\%, 5\%, and 1\% level, respectively. The sample includes small, mid, and large cap stocks (micro and nano caps are excluded). Data are collected from {\it Datastream}.	
	
    \end{footnotesize}  

		\bigskip

\newcolumntype{L}[1]{>{\raggedright\let\newline\\\arraybackslash\hspace{0pt}}m{#1}}
\newcolumntype{C}[1]{>{\centering\let\newline\\\arraybackslash\hspace{0pt}}m{#1}}
\newcolumntype{R}[1]{>{\raggedleft\let\newline\\\arraybackslash\hspace{0pt}}m{#1}}
\setlength\extrarowheight{0pt}

\sisetup{
	input-symbols         = [()],
	table-format          = 1.0,
	table-space-text-post = ***,
	table-align-text-post = false,
	table-text-alignment  = center,
	group-digits          = false
}

\robustify\bfseries
\centering
\scalebox{0.90}{
	
	\begin{tabular}{L{2.8cm}
			S[table-column-width = 2.0cm]
			S[table-column-width = 2.0cm]
			S[table-column-width = 2.0cm]
			S[table-column-width = 2.0cm]
			S[table-column-width = 2.0cm]
			S[table-column-width = 2.0cm]
			S[table-column-width = 2.0cm]
			S[table-column-width = 2.0cm]
		}
		
		\toprule

	        & \multicolumn{2}{c}{Mean} 
	        & \multicolumn{2}{c}{Median} 
	        & \multicolumn{2}{c}{Volatility}
	        & \multicolumn{2}{c}{Maximum drawdown} \\
	             
	        \cmidrule(lr){2-3}
	        \cmidrule(lr){4-5}
	        \cmidrule(lr){6-7}
	        \cmidrule(l){8-9}
	        
	        & \multicolumn{1}{c}{(1)} 
	        & \multicolumn{1}{c}{(2)} 
	        & \multicolumn{1}{c}{(3)} 
	        & \multicolumn{1}{c}{(4)} 
	        & \multicolumn{1}{c}{(5)} 
	        & \multicolumn{1}{c}{(6)} 
	        & \multicolumn{1}{c}{(7)}
	        & \multicolumn{1}{c}{(8)} \\

 	\midrule
 
 \addlinespace[5pt]

	\rc{\it Ban $\times$ Country}  & \bl -0.073 & \bl-0.070 & \bl-0.224\fn{***} & \bl-0.215\fn{***} & \bl-1.203\fn{***} & \bl-1.262\fn{***} & \bl2.757\fn{***} & \bl3.105\fn{***} \\
	\rc & (0.070) & (0.070) & (0.060) & (0.061) & (0.098) & (0.103) & (0.562) & (0.570) \\
	\addlinespace[5pt]

	{\it Ban} & 2.627\fn{***} & 2.639\fn{***} & 1.602\fn{***} & 1.620\fn{***} & 1.061\fn{***} & 0.938\fn{***} & 21.462\fn{***} & 21.931\fn{***} \\
	& (0.043) & (0.043) & (0.035) & (0.035) & (0.062) & (0.061) & (0.316) & (0.327) \\
	\addlinespace[5pt]
	
	{\it Country} & 0.034 & 0.037 & 0.084\fn{*} & 0.097\fn{**} & 0.069 & -0.014 & -0.081 & 0.442 \\
	& (0.050) & (0.051) & (0.044) & (0.045) & (0.085) & (0.092) & (0.634) & (0.670) \\
	\addlinespace[5pt]

	{\it Constant} & -1.979\fn{***} & -2.225\fn{***} & -1.273\fn{***} & -1.553\fn{***} & 4.771\fn{***} & 6.788\fn{***} & -36.125\fn{***} & -42.116\fn{***} \\
	& (0.030) & (0.077) & (0.024) & (0.081) & (0.051) & (0.274) & (0.357) & (1.186) \\
	\addlinespace[5pt]
	
	{$R^2$} & 0.660 & 0.661 & 0.437 & 0.440 & 0.047 & 0.081 & 0.522 & 0.530 \\
	\cmidrule{2-9}

	{\it $\#$ obs} & {3,844} & {3,844} & {3,844} & {3,844} & {3,844} & {3,844} & {3,844} & {3,844}  \\
    \midrule
    {\it Controls}  &       &   {\checkmark}  &       &   {\checkmark}  &       &   {\checkmark} &       &   {\checkmark}      \\

		\bottomrule

\end{tabular}%
}

\end{table}
\end{landscape}
\begin{landscape}
\begin{table}[ht]
	\caption{\bf Stock Returns and Institutional Ownership}\label{tab:DiD_distribution_IO}

	\begin{footnotesize}
	This table presents difference-in-differences estimates associated with the introduction of temporary short-selling bans in European stock markets during the COVID-19 pandemic. The dependent variables are the percentage  mean, median, volatility, and maximum drawdown of total return indices in dollar terms based on a one-month window around the enactment of short-selling bans. The pre-treatment (post-treatment) period runs between February 17 and March 16, 2020 (March 17 and April 15, 2020). Panel A presents estimates for stocks with low institutional ownership (bottom tercile), whereas Panel B  for stocks with high institutional ownership (top tercile).  \emph{Country} is a dummy variable that equals one for the group of countries that introduced short-selling bans (i.e., Austria, Belgium, France, Greece, Italy, and Spain), and zero for the control group of countries (i.e., Denmark, Finland, Germany, Ireland, Netherlands, Norway,  Poland, Portugal, Sweden, Switzerland, and the United Kingdom).  \emph{Ban} is a dummy variable that equals one (zero) for the post-treatment (pre-treatment) window. The set of controls includes firm size and sovereign CDS spread. Standard errors (in parentheses) are clustered by stock dimension. *, **, ***, indicate statistical significance at the 10\%, 5\%, and 1\% level, respectively. The sample includes small, mid, and large cap stocks (micro and nano caps are excluded). Data are collected from {\it Datastream} and {\it Bloomberg}.	
	
    \end{footnotesize}

		\bigskip

\newcolumntype{L}[1]{>{\raggedright\let\newline\\\arraybackslash\hspace{0pt}}m{#1}}
\newcolumntype{C}[1]{>{\centering\let\newline\\\arraybackslash\hspace{0pt}}m{#1}}
\newcolumntype{R}[1]{>{\raggedleft\let\newline\\\arraybackslash\hspace{0pt}}m{#1}}
\setlength\extrarowheight{0pt}

\sisetup{
	input-symbols         = [()],
	table-format          = 1.0,
	table-space-text-post = ***,
	table-align-text-post = false,
	table-text-alignment  = center,
	group-digits          = false
}

\robustify\bfseries
\centering
\scalebox{0.90}{
	
	\begin{tabular}{L{2.8cm}
			S[table-column-width = 2.0cm]
			S[table-column-width = 2.0cm]
			S[table-column-width = 2.0cm]
			S[table-column-width = 2.0cm]
			S[table-column-width = 2.0cm]
			S[table-column-width = 2.0cm]
			S[table-column-width = 2.0cm]
			S[table-column-width = 2.0cm]
		}
		
		\toprule

	        & \multicolumn{4}{c}{\bf Panel A: Low Institutional Ownership} 
	        & \multicolumn{4}{c}{\bf Panel B: High Institutional Ownership}\\
	        
			\cmidrule(lr){2-5}
			\cmidrule(lr){6-9}	
	
     & 	\multicolumn{1}{c}{\multirow{2}{*}{Mean}} &  \multicolumn{1}{c}{\multirow{2}{*}{Median}} & \multicolumn{1}{c}{\multirow{2}{*}{Volatility}} & {Maximum} 
     & 	\multicolumn{1}{c}{\multirow{2}{*}{Mean}} &  \multicolumn{1}{c}{\multirow{2}{*}{Median}} & \multicolumn{1}{c}{\multirow{2}{*}{Volatility}} & {Maximum} \\

	 &  \multicolumn{1}{c}{}                       &  \multicolumn{1}{c}{}                       & \multicolumn{1}{c}{}                            & {drawdown}
	 &  \multicolumn{1}{c}{}                       &  \multicolumn{1}{c}{}                       & \multicolumn{1}{c}{}                            & {drawdown}\\

	\midrule

	\rc{\it Ban $\times$ Country}   & \bl 0.180 & \bl -0.093 & \bl -0.670\fn{***} & \bl 4.300\fn{***} & \bl -0.091 & \bl -0.291\fn{***} & \bl -1.021\fn{***} & \bl 2.077\fn{*} \\
	\rc & (0.117) & (0.106) & (0.147) & (0.987) & (0.133) & (0.107) & (0.184) & (1.116) \\
	\addlinespace[5pt]
		
	{\it Ban} & 2.289\fn{***} & 1.511\fn{***} & 0.251\fn{***} & 20.197\fn{***} & 2.632\fn{***} & 1.607\fn{***} & 0.939\fn{***} & 21.788\fn{***} \\
	& (0.082) & (0.075) & (0.095) & (0.692) & (0.084) & (0.066) & (0.092) & (0.632) \\
	\addlinespace[5pt]
		
	{\it Country} & -0.217\fn{**} & -0.044 & 0.299\fn{**} & -3.302\fn{***} & 0.034 & 0.182\fn{**} & -0.019 & 0.291 \\
	& (0.086) & (0.075) & (0.148) & (1.138) & (0.104) & (0.088) & (0.186) & (1.387) \\
	\addlinespace[5pt]
		
	{\it Constant} & -1.987\fn{***} & -1.488\fn{***} & 5.761\fn{***} & -37.308\fn{***} & -2.392\fn{***} & -1.586\fn{***} & 6.888\fn{***} & -42.927\fn{***} \\
	& (0.143) & (0.147) & (0.471) & (2.299) & (0.136) & (0.166) & (0.536) & (2.282) \\
	\addlinespace[5pt]
		
	{$R^2$} & 0.665 & 0.451 & 0.027 & 0.553 & 0.661 & 0.445 & 0.085 & 0.530\\
	\cmidrule{2-9}
	{$\#$ \it obs} & {1,036} & {1,036} & {1,036} & {1,036} & {1,036} & {1,036} & {1,036} & {1,036} \\
    \midrule
    {\it Controls}  &   {\checkmark}    &   {\checkmark}  &   {\checkmark}    &   {\checkmark}  & {\checkmark}      &   {\checkmark} &   {\checkmark}    &   {\checkmark}      \\

		\bottomrule

\end{tabular}
}

\end{table}
\end{landscape}
\begin{landscape}
\begin{table}[ht]
	\caption{\bf Bid-Ask Spreads and Institutional Ownership: Matched Sample}\label{tab:DiD_bas_IO_matched}
	
	\begin{footnotesize}	
	This table presents difference-in-differences estimates associated with the introduction of temporary short-selling bans in European stock markets during the COVID-19 pandemic. The dependent variable is the percentage bid-ask spread based on daily bid and ask prices for stocks with low institutional ownership (bottom tercile) in Panel A and high institutional ownership (top tercile) in Panel B, respectively.  \emph{Country} is a dummy variable that equals one for the group of countries that introduced short-selling bans (i.e., Austria, Belgium, France, Greece, Italy, and Spain), and zero for the control group of countries (i.e., Denmark, Finland, Germany, Ireland, Netherlands, Norway,  Poland, Portugal, Sweden, Switzerland, and the United Kingdom).   \emph{Ban} is a dummy variable that equals one (zero) for a post-treatment (pre-treatment) period of one month that goes from March 17 to April 15, 2020 (February 17 to March 16, 2020). The set of controls includes firm size and sovereign CDS spread. Specifications are complemented with stock and time (calendar date) fixed effects {\it fe}. Standard errors (in parentheses) are clustered by stock and time (calendar date) dimensions. *, **, ***, indicate statistical significance at the 10\%, 5\%, and 1\% level, respectively. The sample includes small, mid, and large cap stocks (micro and nano caps are excluded) matched by market capitalization and industry classification according to the Industry Classification Benchmark (ICB) code. Data are collected from {\it Datastream} and {\it Bloomberg}.	
	
\end{footnotesize}
	
		\bigskip

	\newcolumntype{L}[1]{>{\raggedright\let\newline\\\arraybackslash\hspace{0pt}}m{#1}}
	\newcolumntype{C}[1]{>{\centering\let\newline\\\arraybackslash\hspace{0pt}}m{#1}}
	\newcolumntype{R}[1]{>{\raggedleft\let\newline\\\arraybackslash\hspace{0pt}}m{#1}}
	\setlength\extrarowheight{0pt}

	\sisetup{
		input-symbols         = [()],
		table-format          = 1.0,
		table-space-text-post = ***,
		table-align-text-post = false,
		table-text-alignment  = center,
		group-digits          = false
	}

	\robustify\bfseries
	\centering
	\scalebox{0.90}{
		
		\begin{tabular}{L{2.9cm}
				S[table-column-width = 1.9cm]
				S[table-column-width = 1.9cm]
				S[table-column-width = 1.9cm]
				S[table-column-width = 1.9cm]
				S[table-column-width = 1.9cm]
				S[table-column-width = 1.9cm]
				S[table-column-width = 1.9cm]
				S[table-column-width = 1.9cm]
			}
	
	\toprule

			& \multicolumn{4}{c}{\bf Panel A: Low Institutional Ownership} 
			& \multicolumn{4}{c}{\bf Panel B: High Institutional Ownership} \\
			\cmidrule(r){2-5} \cmidrule(l){6-9}

		 	& \multicolumn{1}{c}{(1)} 
			& \multicolumn{1}{c}{(2)} 
		    & \multicolumn{1}{c}{(3)} 
		 	& \multicolumn{1}{c}{(4)} 
			& \multicolumn{1}{c}{(5)} 
			& \multicolumn{1}{c}{(6)} 
			& \multicolumn{1}{c}{(7)}
			& \multicolumn{1}{c}{(8)} \\
			
			\midrule
			\addlinespace[5pt]

	\rc{\it Ban $\times$ Country}   & \bl 0.121\fn{*} & \bl 0.110 & \bl 0.112 & \bl 0.133 & \bl 0.262\fn{***} & \bl 0.316\fn{***} & \bl 0.317\fn{***} & \bl 0.340\fn{***} \\
	\rc & (0.066) & (0.088) & (0.089) & (0.101) & (0.058) & (0.073) & (0.075) & (0.082) \\
	\addlinespace[5pt]
	
	{\it Ban} & 0.283\fn{***} & 0.296\fn{***} &       &       & 0.152\fn{***} & 0.150\fn{***} &       &  \\
	& (0.072) & (0.091) &       &       & (0.045) & (0.053) &       &  \\	
	\addlinespace[5pt]
	
	{\it Country} & 0.203\fn{*} &       &       &       & 0.184 &       &       &  \\
	& (0.120) &       &       &       & (0.125) &       &       &  \\
	\addlinespace[5pt]
	
	{Constant} & 0.693\fn{***} & 0.834\fn{***} & 0.977\fn{***} & -1.051 & 0.610\fn{***} & 0.700\fn{***} & 0.772\fn{***} & -0.151 \\
	& (0.110) & (0.075) & (0.029) & (1.822) & (0.092) & (0.046) & (0.020) & (1.161) \\
	\addlinespace[5pt]
	
	{$R^2$} & 0.019 & 0.444 & 0.475 & 0.475 & 0.019 & 0.567 & 0.581 & 0.581  \\
    \cmidrule{2-9}
    {\it $\#$ obs}  & {11,093} & {11,092} & {11,092} & {11,092} & {11,031} & {11,031} & {11,031} & {11,031} \\

    \midrule

	{\it Stock fe} &       & {\checkmark}     & {\checkmark}     & {\checkmark}     &       & {\checkmark}     & {\checkmark}     & {\checkmark} \\
	{\it Time fe}  &       &                  & {\checkmark}     & {\checkmark}     &       &                  & {\checkmark}     & {\checkmark} \\
	{\it Controls} &       &                  &                  & {\checkmark}     &       &                  &                  & {\checkmark} \\
	\bottomrule

\end{tabular}
}

\end{table}
\end{landscape}

\begin{landscape}
\begin{table}[ht]
	\caption{\bf Stock Returns and Institutional Ownership: Matched Sample}\label{tab:DiD_distribution_IO_matched}

	\begin{footnotesize}
	This table presents difference-in-differences estimates associated with the introduction of temporary short-selling bans in European stock markets (see \mytab{tab:data_description} for more details) during the COVID-19 pandemic. The dependent variables are the percentage  mean, median, volatility, and maximum drawdown of total return indices in dollar terms based on a one-month window around the enactment of short-selling bans. The pre-treatment (post-treatment) period runs between February 17 and March 16, 2020 (March 17 and April 15, 2020). Panel A presents estimates for stocks with low institutional ownership (bottom tercile), whereas Panel B  for stocks with high institutional ownership (top tercile). \emph{Country} is a dummy variable that equals one (zero) for the treated (control) group of European countries.  \emph{Ban} is a dummy variable that equals one (zero) for the post-treatment (pre-treatment) period. The set of controls includes firm size and sovereign CDS spread. Standard errors (in parentheses) are clustered by stock dimension. *, **, ***, indicate statistical significance at the 10\%, 5\%, and 1\% level, respectively. The sample includes small, mid, and large cap stocks (micro and nano caps are excluded) matched by market capitalization and industry classification according to the Industry Classification Benchmark (ICB) code. Data are collected from {\it Datastream} and {\it Bloomberg}.	
	
\end{footnotesize}

		\bigskip

\newcolumntype{L}[1]{>{\raggedright\let\newline\\\arraybackslash\hspace{0pt}}m{#1}}
\newcolumntype{C}[1]{>{\centering\let\newline\\\arraybackslash\hspace{0pt}}m{#1}}
\newcolumntype{R}[1]{>{\raggedleft\let\newline\\\arraybackslash\hspace{0pt}}m{#1}}
\setlength\extrarowheight{0pt}

\sisetup{
	input-symbols         = [()],
	table-format          = 1.0,
	table-space-text-post = ***,
	table-align-text-post = false,
	table-text-alignment  = center,
	group-digits          = false
}

\robustify\bfseries
\centering
\scalebox{0.90}{
	
	\begin{tabular}{L{2.8cm}
			S[table-column-width = 2.0cm]
			S[table-column-width = 2.0cm]
			S[table-column-width = 2.0cm]
			S[table-column-width = 2.0cm]
			S[table-column-width = 2.0cm]
			S[table-column-width = 2.0cm]
			S[table-column-width = 2.0cm]
			S[table-column-width = 2.0cm]
		}
		
		\toprule

	        & \multicolumn{4}{c}{\bf Panel A: Low Institutional Ownership} 
	        & \multicolumn{4}{c}{\bf Panel B: High Institutional Ownership}\\
	        
			\cmidrule(lr){2-5}
			\cmidrule(lr){6-9}	
	
     & 	\multicolumn{1}{c}{\multirow{2}{*}{Mean}} &  \multicolumn{1}{c}{\multirow{2}{*}{Median}} & \multicolumn{1}{c}{\multirow{2}{*}{Volatility}} & {Maximum} 
     & 	\multicolumn{1}{c}{\multirow{2}{*}{Mean}} &  \multicolumn{1}{c}{\multirow{2}{*}{Median}} & \multicolumn{1}{c}{\multirow{2}{*}{Volatility}} & {Maximum} \\

	 &  \multicolumn{1}{c}{}                       &  \multicolumn{1}{c}{}                       & \multicolumn{1}{c}{}                            & {drawdown}
	 &  \multicolumn{1}{c}{}                       &  \multicolumn{1}{c}{}                       & \multicolumn{1}{c}{}                            & {drawdown}\\

	\midrule

	\rc{\it Ban $\times$ Country}   & \bl 0.450\fn{**} & \bl -0.006 & \bl -0.609\fn{**} & \bl 5.502\fn{***} & \bl -0.009 & \bl -0.294\fn{**}  & \bl -0.825\fn{***} & \bl 2.522\fn{*} \\
	\rc & (0.176) & (0.161) & (0.239) & (1.658) & (0.165) & (0.131) & (0.225) & (1.458) \\	
	\addlinespace[5pt]
	
	{\it Ban} & 2.050\fn{***} & 1.429\fn{***} & 0.176 & 19.251\fn{***} & 2.570\fn{***} & 1.593\fn{***} & 0.883\fn{***} & 21.086\fn{***} \\
	& (0.152) & (0.139) & (0.208) & (1.488) & (0.125) & (0.103) & (0.162) & (1.139) \\
	\addlinespace[5pt]
	{\it Country} & -0.356\fn{***} & -0.160 & 0.475\fn{**} & -4.889\fn{***} & -0.058 & 0.144 & 0.050 & -0.984 \\
	& (0.126) & (0.097) & (0.223) & (1.693) & (0.124) & (0.103) & (0.218) & (1.661) \\
	\addlinespace[5pt]
	{\it Constant} & -1.608\fn{***} & -0.960\fn{***} & 4.220\fn{***} & -29.625\fn{***} & -2.274\fn{***} & -1.464\fn{***} & 5.769\fn{***} & -37.888\fn{***} \\
	& (0.220) & (0.187) & (0.508) & (3.087) & (0.194) & (0.225) & (0.655) & (3.276) \\
	\addlinespace[5pt]
	
	{$R^2$} & 0.692 & 0.472 & 0.011 & 0.601 & 0.697 & 0.462 & 0.036 & 0.531  \\
	\cmidrule{2-9}
	{\it $\#$ obs} & {552}   & {552}   & {552}   & {552}   & {550}   & {550}   & {550}   & {550}  \\

    \midrule
    {\it Controls}  &   {\checkmark}    &   {\checkmark}  &   {\checkmark}    &   {\checkmark}  & {\checkmark}      &   {\checkmark} &   {\checkmark}    &   {\checkmark}      \\

		\bottomrule

\end{tabular}
}

\end{table}
\end{landscape}
\begin{landscape}
\begin{table}[ht]
	\caption{\bf Bid-Ask Spreads and Institutional Ownership: Placebo}
	\label{tab:PlaceboOutcome_bas_IO_matched}

	\begin{footnotesize}		
	This table presents difference-in-differences estimates associated with the introduction of temporary short-selling bans in European stock markets during the COVID-19 pandemic. The dependent variable is the percentage bid-ask spread based on daily bid and ask prices for stocks with low institutional ownership (bottom tercile) in Panel A and stocks with high institutional ownership (top tercile) in Panel B, respectively.  \emph{Country} is a dummy variable that equals one for the group of placebo countries (i.e., Denmark, Finland, Germany, Ireland, Netherlands, and Norway), and zero for the group of control countries (i.e., Poland, Portugal, Sweden, Switzerland, and the United Kingdom). \emph{Ban} is a dummy variable that equals one (zero) for a post-treatment (pre-treatment) period of one month that goes from March 17 to April 15, 2020 (February 17 to March 16, 2020). The set of controls includes firm size and sovereign CDS spread. Specifications are complemented with stock and time (calendar date) fixed effects {\it fe}. Standard errors (in parentheses) are clustered by stock and time (calendar date) dimensions. *, **, ***, indicate statistical significance at the 10\%, 5\%, and 1\% level, respectively. The sample includes small, mid, and large cap stocks (micro and nano caps are excluded) matched by firm characteristics. Data are collected from {\it Datastream} and {\it Bloomberg}.	
	
\end{footnotesize}

		\bigskip

	\newcolumntype{L}[1]{>{\raggedright\let\newline\\\arraybackslash\hspace{0pt}}m{#1}}
	\newcolumntype{C}[1]{>{\centering\let\newline\\\arraybackslash\hspace{0pt}}m{#1}}
	\newcolumntype{R}[1]{>{\raggedleft\let\newline\\\arraybackslash\hspace{0pt}}m{#1}}
	\setlength\extrarowheight{0pt}

	\sisetup{
		input-symbols         = [()],
		table-format          = 1.0,
		table-space-text-post = ***,
		table-align-text-post = false,
		table-text-alignment  = center,
		group-digits          = false
	}

	\robustify\bfseries
	\centering
	\scalebox{0.90}{
		
		\begin{tabular}{L{2.8cm}
				S[table-column-width = 1.9cm]
				S[table-column-width = 1.9cm]
				S[table-column-width = 1.9cm]
				S[table-column-width = 1.9cm]
				S[table-column-width = 1.9cm]
				S[table-column-width = 1.9cm]
				S[table-column-width = 1.9cm]
				S[table-column-width = 1.9cm]
			}
	
	\toprule

			& \multicolumn{4}{c}{\bf Panel A: Low Institutional Ownership} 
			& \multicolumn{4}{c}{\bf Panel B: High Institutional Ownership} \\
			\cmidrule(r){2-5} \cmidrule(l){6-9}

		 	& \multicolumn{1}{c}{(1)} 
			& \multicolumn{1}{c}{(2)} 
		    & \multicolumn{1}{c}{(3)} 
		 	& \multicolumn{1}{c}{(4)} 
			& \multicolumn{1}{c}{(5)} 
			& \multicolumn{1}{c}{(6)} 
			& \multicolumn{1}{c}{(7)}
			& \multicolumn{1}{c}{(8)} \\
			
			\midrule
			\addlinespace[5pt]

	\rc{\it Ban $\times$ Country}   & \bl 0.041 & \bl0.050 & \bl0.050 & \bl-0.014 & \bl-0.070 & \bl-0.074 & \bl-0.076 & \bl-0.014 \\
	\rc & (0.091) & (0.127) & (0.128) & (0.141) & (0.048) & (0.077) & (0.078) & (0.077) \\
	\addlinespace[5pt]
	
	{\it Ban} & 0.258\fn{***} & 0.266\fn{***} &       &       & 0.181\fn{***} & 0.181\fn{***} &       &  \\
	& (0.080) & (0.089) &       &       & (0.050) & (0.065) &       &  \\
	\addlinespace[5pt]
	
	{\it Country} & -0.045 &       &       &       & -0.243 &       &       &  \\
	& (0.166) &       &       &       & (0.184) &       &       &  \\
	\addlinespace[5pt]

	{\it Constant} & 0.691\fn{***} & 0.663\fn{***} & 0.791\fn{***} & 12.257\fn{*} & 0.718\fn{***} & 0.639\fn{***} & 0.727\fn{***} & 0.062 \\
	& (0.097) & (0.058) & (0.028) & (6.241) & (0.128) & (0.042) & (0.009) & (1.499) \\
	\addlinespace[5pt]
	
	{$R^2$} & 0.008 & 0.504 & 0.517 & 0.528 & 0.012 & 0.529 & 0.540 & 0.542  \\
	\cmidrule{2-9}
	{$\#$ \it obs} & {4,090} & {4,089} & {4,089} & {4,089} & {4,252} & {4,252} & {4,252} & {4,252}  \\
	\midrule
	{\it Stock fe} &       & {\checkmark}     & {\checkmark}     & {\checkmark}     &       & {\checkmark}     & {\checkmark}     & {\checkmark} \\
	{\it Time fe}  &       &                  & {\checkmark}     & {\checkmark}     &       &                  & {\checkmark}     & {\checkmark} \\
	{\it Controls} &       &                  &                  & {\checkmark}     &       &                  &                  & {\checkmark} \\
	\bottomrule

\end{tabular}%
}

\end{table}
\end{landscape}

\begin{landscape}
\begin{table}[ht]
	\caption{\bf Stock Returns and Institutional Ownership: Placebo}\label{tab:PlaceboOutcome_distribution_IO_matched}

\begin{footnotesize}
		This table presents difference-in-differences estimates associated with the introduction of temporary short-selling bans in European stock markets during the COVID-19 pandemic. The dependent variables are the percentage mean, median, volatility, and maximum drawdown of stock returns based on a one-month window around the enactment of short-selling bans. The pre-treatment (post-treatment) period runs between February 17 and March 16, 2020 (March 17 and April 15, 2020).  Panel A presents estimates for stocks with low institutional ownership (bottom tercile), whereas Panel B  for stocks with high institutional ownership (top tercile). \emph{Country} is a dummy variable that equals one for the group of placebo countries (i.e., Denmark, Finland, Germany, Ireland, Netherlands, and Norway), and zero for the group of control countries (i.e., Poland, Portugal, Sweden, Switzerland, and the United Kingdom). \emph{Ban} is a dummy variable that equals one (zero) for the post-treatment (pre-treatment) window. The set of controls includes firm size and sovereign CDS spread. Standard errors (in parentheses) are clustered by stock dimension. *, **, ***, indicate statistical significance at the 10\%, 5\%, and 1\% level, respectively. The sample includes small, mid, and large cap stocks (micro and nano caps are excluded). Data are collected from {\it Datastream} and {\it Bloomberg}.	
		
	\end{footnotesize}

		\bigskip

\newcolumntype{L}[1]{>{\raggedright\let\newline\\\arraybackslash\hspace{0pt}}m{#1}}
\newcolumntype{C}[1]{>{\centering\let\newline\\\arraybackslash\hspace{0pt}}m{#1}}
\newcolumntype{R}[1]{>{\raggedleft\let\newline\\\arraybackslash\hspace{0pt}}m{#1}}
\setlength\extrarowheight{0pt}

\sisetup{
	input-symbols         = [()],
	table-format          = 1.0,
	table-space-text-post = ***,
	table-align-text-post = false,
	table-text-alignment  = center,
	group-digits          = false
}

\robustify\bfseries
\centering
\scalebox{0.90}{
	
	\begin{tabular}{L{2.8cm}
			S[table-column-width = 2.0cm]
			S[table-column-width = 2.0cm]
			S[table-column-width = 2.0cm]
			S[table-column-width = 2.0cm]
			S[table-column-width = 2.0cm]
			S[table-column-width = 2.0cm]
			S[table-column-width = 2.0cm]
			S[table-column-width = 2.0cm]
		}
		
		\toprule

     & \multicolumn{4}{c}{\bf Panel A: Low Institutional Ownership} 
	 & \multicolumn{4}{c}{\bf Panel B: High Institutional Ownership}\\
	        
	\cmidrule(lr){2-5}
	\cmidrule(lr){6-9}	
	
     & 	\multicolumn{1}{c}{\multirow{2}{*}{Mean}} &  \multicolumn{1}{c}{\multirow{2}{*}{Median}} & \multicolumn{1}{c}{\multirow{2}{*}{Volatility}} & {Maximum} 
     & 	\multicolumn{1}{c}{\multirow{2}{*}{Mean}} &  \multicolumn{1}{c}{\multirow{2}{*}{Median}} & \multicolumn{1}{c}{\multirow{2}{*}{Volatility}} & {Maximum} \\

	 &  \multicolumn{1}{c}{}                       &  \multicolumn{1}{c}{}                       & \multicolumn{1}{c}{}                            & {drawdown}
	 &  \multicolumn{1}{c}{}                       &  \multicolumn{1}{c}{}                       & \multicolumn{1}{c}{}                            & {drawdown}\\

	\midrule

	\rc{\it Ban $\times$ Country}   & \bl0.373 &\bl 0.508\fn{**} & \bl0.022 & \bl2.896 & \bl-0.069 & \bl-0.215 & \bl-1.062\fn{***} & \bl2.200 \\
	\rc& (0.285) & (0.252) & (0.384) & (2.838) & (0.308) & (0.223) & (0.355) & (2.870) \\
	\addlinespace[5pt]
	{\it Ban} & 1.934\fn{***} & 1.185\fn{***} & 0.335 & 17.662\fn{***} & 2.545\fn{***} & 1.662\fn{***} & 1.307\fn{***} & 19.510\fn{***} \\
	& (0.223) & (0.202) & (0.278) & (2.047) & (0.175) & (0.148) & (0.256) & (1.881) \\
	\addlinespace[5pt]
	{\it Country} & -0.328 & -0.277\fn{*} & 0.413 & -4.148 & 0.166 & 0.208 & 0.088 & 1.986 \\
	& (0.199) & (0.158) & (0.387) & (2.671) & (0.238) & (0.178) & (0.374) & (2.945) \\
	\addlinespace[5pt]
	\addlinespace[5pt]
	{\it Constant} & -1.859\fn{***} & -1.057\fn{***} & 5.665\fn{***} & -36.080\fn{***} & -2.564\fn{***} & -1.576\fn{***} & 7.989\fn{***} & -48.721\fn{***} \\
	& (0.315) & (0.276) & (0.746) & (4.391) & (0.293) & (0.356) & (1.106) & (5.028) \\
	\addlinespace[5pt]
	{$R^2$} & 0.585 & 0.459 & 0.044 & 0.502 & 0.682 & 0.511 & 0.155 & 0.498  \\
	\cmidrule{2-9}
	{\it $\#$ obs} & {200}   & {200}   & {200}   & {200}   & {202}   & {202}   & {202}   & {202} \\
	\midrule
	{\it Controls}  &   {\checkmark}    &   {\checkmark}  &   {\checkmark}    &   {\checkmark}  & {\checkmark}      &   {\checkmark} &   {\checkmark}    &   {\checkmark}      \\

	\bottomrule

\end{tabular}
}

\end{table}
\end{landscape}

\newpage 

\phantomsection
\addcontentsline{toc}{section}{Internet Appendix}
\begin{appendices}
	\begin{center}
		
		{\LARGE \texttt{Internet Appendix to}\label{app:internet_appendix} \vspace{0.5cm}\\
			\textbf{\textquotedblleft The Double-Edged Sword \\ \vspace{0.0cm} of Short-Selling Bans\textquotedblright} \vspace{0.5cm}\\
			{\large (\texttt{not for publication})} \vspace{0.5cm}\\}
		
	\end{center}

	\begin{abstract}
		\noindent This Internet Appendix presents additional technical details and results not included in the main body of the paper.
		
		\vspace{1cm}
		
	\end{abstract}

	\pagenumbering{gobble}   %
	
	\setcounter{section}{0}
	\setcounter{subsection}{0}
	\renewcommand{\thesection}{\Alph{section}}
	\renewcommand{\thesubsection}{\thesection.\arabic{subsection}}
	
	\setcounter{equation}{0}
	\renewcommand{\theequation}{\thesection.\arabic{equation}}
	
	\setcounter{table}{0}
	\renewcommand\thetable{A.\arabic{table}}
	
	\setcounter{figure}{0}
	\renewcommand\thefigure{A.\arabic{figure}}
	
	\newpage
	\pagenumbering{arabic}
	\renewcommand{\thepage}{\arabic{page}}
	\setcounter{page}{1}
	\pagestyle{plain}
	
	\newpage

\newpage

\section{Model's Proofs and Extensions} \label{app:proof_simulations}

\subsection{Proof of \mylem{lem:prob}}\label{app:proof_lemma1}
\begin{proof}
If $c$ is small enough, we have
\begin{align*}
P(q<c)&=\mathbb{E}[\mathbb{E}[\mathbbm{1}_{q<c}|\eta_s]] \nonumber  \\
&=\mathbb{E}[P(Sell| \eta_s)\mathbbm{1}_{Bid<c}] \nonumber \\
&= \mathbb{E}\left[P(Sell|\eta_s)\mid\eta_s<K\right]P(\eta_s<K),
\end{align*}
where $Bid<c$ \emph{iff} $\eta_s<K$, and $\mathbbm{1}$ is an indicator function that takes the value of one (and zero otherwise) if $q<c$. Noting that $P(Sell|\eta_s)=(1-p)g\alpha+(1-\alpha)g \eta_s$, we than have
\begin{equation}\label{eq:probq}
    P(q<c)= \left((1-p)g\alpha+(1-\alpha)g \mathbb{E}[\eta_s\mid\eta_s<K]\right)\cdot P(\eta_s<K).
\end{equation}
\end{proof}

\subsection{Proof of \myres{prop:zeta}}\label{app:proof_result1}
\begin{proof}
	The {\it regulator} only acts when   $\int_{0}^{K}\left((1-p)g\alpha+(1-\alpha)g\eta\right)f(\eta)d\eta>x$ according to \myeq{eq:probq}. Since the left-hand side is increasing in $K$, we obtain the above necessary condition. 
\end{proof}

\subsection{Proof of \mylem{lem:conprob}}\label{app:proof_lemma3}
\begin{proof}
The first part of the proof of this lemma is the same as that of \mylem{lem:prob}, with $\frac{h_I}{h_N}K$ playing the role of $K$. As for the case of the standard uniform distribution, we have:
\begin{align*}
P(\tilde{q}<c) &=\left((1-p)g h_I\alpha+(1-\alpha)g h_N \mathbb{E}[\eta_s\mid\eta_s<\frac{h_I}{h_N}K]\right)\cdot P(\eta_s<\frac{h_I}{h_N}K)\\
&=h_I \left((1-p)g \alpha+(1-\alpha)g  \frac{\mathbb{E}[\eta_s\mid\eta_s<\frac{h_I}{h_N}K]}{h_I/h_N}\right)\cdot \frac{h_I}{h_N}\frac{P(\eta_s<\frac{h_I}{h_N}K)}{h_I/h_N}\\
&=h_I \left((1-p)g \alpha+(1-\alpha)g  \mathbb{E}[\eta_s\mid\eta_s<K]\right)\cdot \frac{h_I}{h_N}P(\eta_s<K)\\
&=\frac{h_I^2}{h_N}P(q<c).
\end{align*}
We used the linearity of the cumulative distribution function of the uniform distribution to get from the second to the third line, and \myeq{eq:probq} to derive the final line.
\end{proof}

\subsection{Proof of \mylem{lem:meanmedian}}\label{app:proof_lemma4}
\begin{proof}
	As also explained in \citet{Diamond1987}, the law of iterated expectations and the risk neutrality of {\it market makers}  imply that
	\begin{equation*}
		\mathbb{E}[\tilde{q}]=\mathbb{E}[\mathbb{E}[V|\tilde{\mathcal{F}}]]=\mathbb{E}[V]=\mathbb{E}[\mathbb{E}[V|\mathcal{F}]]=\mathbb{E}[q],
	\end{equation*}
	where $\mathcal{F}$ ($\tilde{\mathcal{F}}$) denotes the information sets of the {\it market maker} without (with) short-selling bans.  The median price (denoted as $\mu_{1/2}$) coincides with the price attained with the no trade action in either cases. Under no bans, $E[V|\text{No trade}]=p$ and $\mu_{1/2}(q)=p$ since $P(q\leq p)=1-P(Buy)\geq 1/2$ and  $P(q\geq p)=1-P(Sell)\geq 1/2$. With short-selling bans, moreover, $E[V|\text{No Trade},\eta_s]=\frac{((1-g)+g(1-\alpha)\eta_s(1-h_N))p}{1-g+g((1-\alpha)\eta_s(1-h_N)+\alpha(1-p)(1-h_I))}<p$ independently of $\eta_s$. It then follows that $\mu_{1/2}(\tilde{q})<\mu_{1/2}(q)$.
\end{proof}

\subsection{General distribution $f_u(\eta)$}\label{app:technical}

Let us consider the following family of symmetric distributions in $[0,1]$, parametrized by $u\in[0,2]$:

\[ 
f_u(\eta)= \left\{
\begin{array}{ll}
     u-4(u-1)\eta  &  \eta\leq \frac{1}{2}\\
      4-3u+4(u-1)\eta & \eta>\frac{1}{2} \\ 
\end{array} 
\right. 
\]
For example, for $u=1$, we obtain the standard uniform distribution $U(0,1)$. But more generally, this is a tractable family of distributions in $[0,1]$, indexed by $u$, that can be (second-order) stochastically ordered.\footnote{We choose this family of distributions, as opposed to other families like the Beta distribution, so that we can compute the $E[\eta_s|\eta_s<K]$ in closed form.} Since these distributions are symmetric with $E[f_u(\eta)]=1/2$, it is easy to show that $f_{u_1}(\eta)\succeq f_{u_2}(\eta)$ iff $u_1<u_2$: if we consider the ratio $\frac{f_{u_1}(\eta)}{f_{u_2}(\eta)}$, this is increasing in $[0,\frac{1}{2}]$ and decreasing in $[\frac{1}{2},1]$. It, thus, follows by \citet{Ramos2000sufficient} that the two distributions are second-order stochastically ordered. 
Then using \mylem{lem:prob}, and assuming that $K$ is sufficiently small, we get:
\begin{lemma}
\label{lem:SSD}
In the unconstrained economy when $c$ is sufficiently small ($c<\frac{(1-a)p}{(1-a)+4a(1-p)}$), the likelihood of a very low price, $P(q<c)$, is increasing in the perceived variance of $\eta_s$.
\end{lemma}
\begin{proof}
Using \mylem{lem:prob}, we know that 
$$    P(q<c)=\int_{0}^{K}\left((1-p)g\alpha+(1-\alpha)g\eta\right)f_u(\eta)d\eta$$
We now have that $c<\frac{(1-a)p}{(1-a)+4a(1-p)}\Longrightarrow K<\frac{1}{4}$ and hence we can write
$$P(q<c)=\int_{0}^{K}\left((1-p)g\alpha+(1-\alpha)g\eta\right)( u-4(u-1)\eta)d\eta,$$
which increases in $u$. Moreover, because of the stochastic dominance result shown above for the specified families of distributions $f_u(\eta)$, we get that $var[\eta(u)]$ is also increasing in $u$. Thus, the more uncertain the {\it regulator} is about the sentiment of the {\it noise} traders, the higher the left tail of prices and hence the higher the likelihood of bans getting imposed.
\end{proof}

\subsection{Simulations}
In \myfigIA{fig:simulated_bid_ask}, we provide a graphical illustration of the relationship between $\eta_s$ and the prices quoted by {\it market makers}. We use simulated data based on $\alpha=0.5$, $p= 0.5$, $g= 0.9$, and $h_N=0.3$. Panel A considers the benchmark scenario where all traders can freely short the risky asset. While bid and ask prices are increasing in $\eta_s$, the no trade price is independent of $\eta_s$ since {\it market makers} cannot update their beliefs and rely solely on public information to determine the conditional expectation of $V$ (equal to $0.5$ in our simulation). In the subsequent panels, we analyze the alternative scenario where the regulatory authority enforces a prohibition on short sales, while varying the proportion of {\it informed} relative to {\it noise} traders holding the risky asset. Specifically, Panel B sets $h_I = 0.6$ so that $h_I/h_N > 1$ as in \mylem{lem:bidask}, and displays a wider bid-ask spread than the benchmark case, driven by a lower bid price and an identical ask price. Also, the no trade price is now increasing in $\eta_s$ since a higher (lower) $\eta_s$ implies a smaller (larger) number of constrained {\it informed} investors,  and {\it market makers} can adjust upward (downward) their expectations of $V$. In Panel C, we set $h_I/h_N = 1$ as in \citet{Diamond1987} and show that the bid-ask spread remains identical to the benchmark case, while the no trade price remains a function on $\eta_s$. In Panel D, finally, we set $h_I = 0.2$ so that $h_I/h_N < 1$ and report a slightly tighter bid-ask spread than the benchmark case due to a higher bid price. This happens as sell orders are more likely to originate from {\it noise} traders, thus reducing the adverse selection faced by {\it market makers}.

Moreover, to better understand the effect of bans on prices, \myfigIA{fig:model_distribution} estimates the density of prices with or without bans using simulated data under different scenarios. Without loss of generality, we choose the following baseline parameters: $a=0.5, p=0.5, \eta=0.5, g=0.9,  h_N=0.3$, and  $h_I = 0.2$ (so that $h_I^2/h_N < 1$ in Panel A) or $h_I = 0.6$ (so that $h_I^2/h_N > 1$ in Panel B). Akin to \citet{Diamond1987}, the price is equal to $E[V|\text{No trade}]$ in the case of a no trade action. 
Panel A examines the price distribution with $h_I^2/h_N < 1$ and shows that the {\it regulator} can effectively reduce the probability of an extreme left-tail event as stated in \myres{prop:condition}. When short selling is prohibited (relative to the unconstrained scenario), the weight on very small price realizations falls while the likelihood of having prices below the average (but not extreme) rises, since the no-trade event is now regarded as negative news. 
Panel B, moreover, considers the price distribution with $h_I^2/h_N > 1$ and reveals that the {\it regulator} is unable to reduce the probability of an extreme left-tail event as such events are more likely to happen when short selling is prohibited (relative to the unconstrained scenario). Why is this so? Even though the likelihood of a sell order decreases, the bid price after the ban also declines. This is because when many informed traders own the asset, the sell order becomes more informative about the payoff, and market makers adjust their valuation towards the low payoff (i.e., zero). In that case, bans are ineffective.  

\newpage

\section{Supplementary Empirical Analysis} \label{app:supplement_analysis}
\subsection{Short-Selling Bans and Market Liquidity}
Here, we present the full analysis of our study of the effect of short-sale restrictions on stock market liquidity using bid-ask spreads, following the seminal paper of \citet{Beber2013}. While other measures of market liquidity could be used, \citet{Goyenko2009} show that liquidity measures based on bid-ask prices are closely related to actual transaction costs. To assess the impact of the ban, we calculate the average daily bid-ask spread over a window that covers one month before and one month after the introduction of short-selling bans. While liquidity was lower during the ban than in the preceding period as previously discussed,  we cannot conclude that the imposition of short-selling bans \emph{caused} the rise in bid-ask spreads. There is evidence that liquidity started to decrease several weeks before the imposition of bans. To visually inspect the sensitivity of our results to the specific choice of the start and end dates of the observation windows, we examine the daily bid-ask spreads between January 2 and June 2, 2020. 

\myfigIA{fig:DiD_bas} plots the average bid-ask spread for two groups of countries, namely, countries that imposed short-selling bans (dark blue line) and countries that refrained from this policy (light blue line). Both series are smoothed on a five-day rolling window to mitigate noise. While bid-ask spreads began to widen in February when the infection started to spread in Europe, short-selling bans were not implemented until March 17, 2020. In contrast, short-selling bans during the global financial crisis were enacted almost immediately after the collapse of Lehman Brothers on September 15, 2008. The delayed policy response in 2020 is potentially beneficial for our analysis, as it allows the effects of confounding factors like  pandemic-related uncertainty to be reflected in prices prior to the imposition of the bans. As a result, this delay can improve our ability to isolate the impact of short-selling bans on market liquidity. Visually,  \myfigIA{fig:DiD_bas} shows that the bid-ask spreads for the two groups of countries moved together prior to the implementation of short-selling bans.  However, following the introduction of these restrictions (red dashed vertical line), bid-ask spreads increased more sharply in the countries that imposed bans, suggesting a potential adverse effect of short-selling bans on market liquidity.  The divergence in bid-ask spreads gradually narrows after the bans are lifted on May 18, 2020 (green dashed line).
To formalize our findings, \mytabIA{tab:DiD_bas} reports the results of a difference-in-differences regression that estimates how short-selling bans differentially affected market liquidity. The dependent variable is the daily percentage bid-ask spread.  \emph{Country} is a dummy variable that equals one for the group of countries that introduced short-selling bans (i.e., Austria, Belgium, France, Greece, Italy, and Spain), and zero for the control group of countries (i.e., Denmark, Finland, Germany, Ireland, Netherlands, Norway,  Poland, Portugal, Sweden, Switzerland, and the United Kingdom). \emph{Ban} is a dummy variable that equals one (zero) for a post-treatment (pre-treatment) period of one month that goes from March 17 to April 15, 2020 (February 17 to March 16, 2020). Standard errors are clustered by stock and time (calendar date) dimensions. 

The first specification shows that the average bid-ask spreads increased by $25$ basis points during the short-selling ban period across all stocks in the sample.  However, the difference-in-differences specification allows us to estimate the relative effect of short-selling bans on liquidity, while controlling for both unconditional differences in bid-ask spread levels and the overall rise in spreads during the ban period. The coefficient estimate on the interaction term {\it Ban} $\times$ {\it Country} indicates that  the average bid-ask spreads in countries that imposed short-selling bans widened by an additional $11.6$ basis points compared to countries without short-selling restrictions, and the result is statistically significant at the 1\% level. The second specification introduces stock-level fixed effects, while the third specification further adds time (calendar date) fixed effects. Stock fixed effects control for time-invariant unobserved heterogeneity such as the number of market makers, analyst coverage, and country characteristics. Time fixed effects capture unobserved common shocks that vary over time, including market-wide liquidity conditions, which are particularly relevant during periods of systemic stress like the COVID-19 pandemic. The second and third specifications confirm the result of the first specification in terms of statistical significance and point estimate, with the differential effect of short-selling bans on bid-ask spreads estimated at $13$ basis points. 

Stock and time fixed effects may fail to control for unobserved time-varying stock characteristics. If the decision to impose short-selling bans is correlated with country or stock characteristics, then any regression results may incorrectly attribute the effects of those characteristics to the imposition of short-selling bans. \citet{Beber2013}, for example, argue that country-level CDS spreads may be correlated with the regulatory decisions to implement short-selling bans, raising the possibility that our estimates could be confounded by country-specific default risk. To address this concern, our fourth specification includes country-specific daily CDS spreads and firm size (measured as the log of stock market capitalization) as control variables. The coefficient on our variable of interest remains statistically significant and closely aligned in magnitude ($12$ basis points) with previous specifications.  We thus conclude that our results support \myhyp{h:liquidity} and are consistent with the findings of \citet{Beber2013} during the global financial crisis.\footnote{\citet{Beber2013} estimate that short-selling bans increased bid-ask spreads by around 198 basis points for covered bans, with the effect reduced by 65 basis points in jurisdictions with short-sale disclosure. In contrast, our estimates range from 11 to 13 basis points, a discrepancy we attribute to sample differences. First, our analysis excludes micro  and nano caps, which are more illiquid and likely to magnify the effect. As a result, pre-ban bid-ask spreads in our sample are significantly lower across all countries. Second, in 2008 several countries imposed bans on naked short selling, which has been permanently prohibited in Europe since 2012. Take together, these differences may explain the smaller impact of short-selling bans on liquidity observed in 2020.}

In our baseline specifications, we cluster standard errors by stock and time dimensions, which may fail to fully capture cross-sectional dependence arising from shocks that simultaneously affect multiple firms within the same country on the same day. For instance, a policy announcement like in our case could generate correlated responses among firms within a country, beyond what is captured by time fixed effects alone. To address this concern, we extend our clustering strategy by computing standard errors clustered by stock and country-time dimensions,  accounting for cross-sectional correlation within each country at a given time. The results, presented in \mytabIA{tab:DiD_bas_sctcls}, reveal that the standard error on the slope coefficient of interest, i.e., the coefficient measuring the differential change in bid-ask spreads in countries with bans relative to those without bans following the short-selling restrictions, is slightly wider when stock and time fixed effects are omitted, but becomes slightly tighter when both are included.  Overall, the results remain qualitatively unchanged, confirming our baseline findings.

\subsection{Short-Selling Bans and Endogeneity}
A key concern with the estimates in \mytabIA{tab:DiD_bas} is the potential endogeneity of short-selling bans. If policymakers tend to impose bans during periods of rising volatility and falling liquidity, the relationship between short-selling bans and market liquidity could not be interpreted as a causal relationship. To address this concern, we employ an instrumental variables strategy, where the first stage determines the likelihood of a ban and the second stage its effect on liquidity. Valid instruments must be correlated with the decision to impose a ban but uncorrelated with the residuals of the bid-ask spread regression. Since short-selling bans are market-wide policies, the instruments must also be market-level variables that vary over time to avoid collinearity with stock fixed effects.

Similar to \citet{Beber2013}, we employ the monthly average sovereign CDS spread and the monthly Financial Stress Index of \citet{Duprey2017} in logs as instruments. Both variables are measured over the previous month at the country level, and daily values are constructed by forward filling, i.e., holding the most recent monthly observation constant until a new value becomes available.\footnote{Country-level data on the Financial Stress Index are available through the {\it ECB Data Portal}. For Norway and Switzerland, which are not included in the ECB dataset, we replicate the index in the spirit of \citet{Duprey2017}.} The CDS spread is a market-based proxy for the default risk of a country, while the Financial Stress Index captures systemic financial pressures. We expect countries in which these sources of risk are greater to be more inclined to impose short-selling restrictions. Importantly, because the instruments are lagged, they should not be correlated with contemporaneous bid-ask spreads at the individual stock level if default and systemic risk are already fully priced in contemporaneous bid-ask spreads.

We report our results in \mytabIA{tab:IV_bas}. In the first-stage regression, we regress the dummy variable indicating the ban period on both instruments, which exhibit strong explanatory power, even after controlling for firm size and firm-level volatility, and including stock and time (calendar date) fixed effects. In the second stage regression, we regress percentage bid-ask spreads on the instrumented short-selling ban indicator, controlling for firm size, firm-level volatility, and including stock and time (calendar date) fixed effects. We estimate our specifications over the full sample from January 2 to June 2, 2020. While our difference-in-differences analysis relies on a narrow one-month window around the introduction of the bans, the full sample for the instrumental variables approach provides greater variation in the instruments and helps improve the predictive power of the first-stage regression. The coefficient on the estimated ban, ranging between $0.126$ and $0.135$,  is statistically significant and closely aligned with the magnitudes reported in \mytabIA{tab:DiD_bas}. 

Finally, we assess the validity of our instruments using standard diagnostic tests. First, the Kleibergen-Paap LM statistic rejects the null hypothesis of underidentification, confirming that the instruments are relevant and the model is identified. Second, both the Kleibergen-Paap and Cragg-Donald Wald F-tests exceed the critical values proposed by \citet{stock_yogo:2015}, indicating that the instruments are not weak. Finally, we test the validity of the overidentifying restrictions, which require that the instruments are correlated with the endogenous variables but uncorrelated with the error term. The Hansen J statistic supports this assumption, indicating that the instruments are unrelated to the residual error term. Overall, our instrumental variable estimates reinforce the baseline results and strengthen the causal interpretation of the ban's impact on market liquidity.

\newpage

\begin{figure}
	\begin{center}
		\includegraphics[scale = 0.55, angle = 0, trim = 00mm 00mm 00mm 00mm]{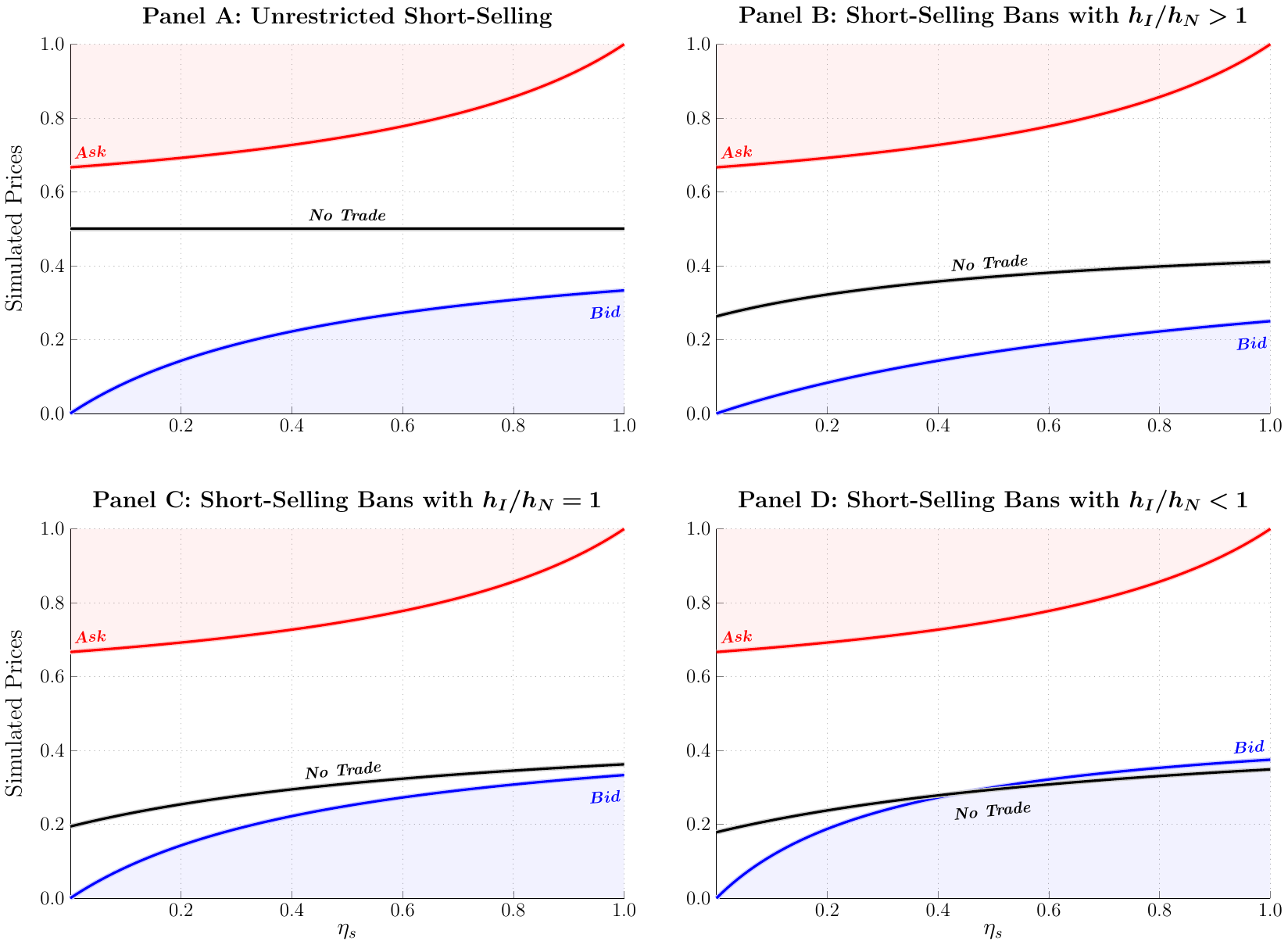}
		\vspace{0.5cm}
		\caption{\bf Simulated Bid-Ask Spreads under Different Scenarios} 
		\label{fig:simulated_bid_ask}
	\end{center}
	
\begin{footnotesize}
This figure displays the relationship between $\eta_{s}$ and the prices set by {\it market makers} under different scenarios. Panel A illustrates the scenario with no short-selling ban, whereas Panels B, C, and D focus on the case with short-selling bans, while varying the proportion of {\it informed} to {\it noise} traders holding the risky asset  ($h_I/h_N$). The simulation assumes $\alpha=0.5$, $p= 0.5$, $g= 0.9$, $h_N=0.3$, $h_I = 0.6$ in Panel B, $h_I = 0.3$ in Panel C, and $h_I = 0.2$ in Panel D. A description of these parameters is presented in \mytabIA{tab:summary_notation}.

\end{footnotesize}
\end{figure}

\begin{figure}

	\begin{center}		
		\includegraphics[scale=0.70, angle = 0, trim = 05mm 00mm 00mm 00mm]{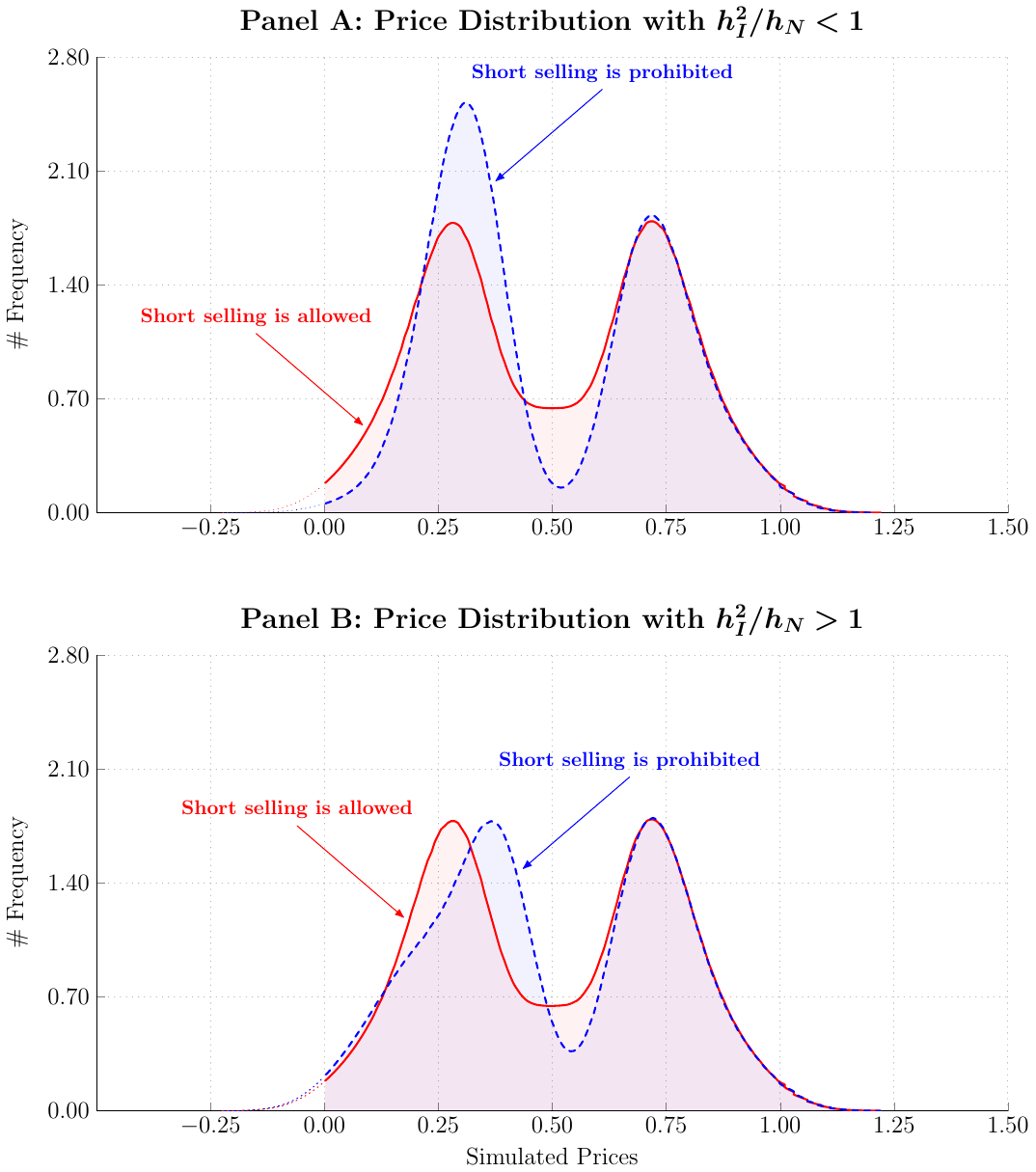}
		\vspace{0.5cm}
		\caption{\bf Distribution of Simulated Prices} 
		\label{fig:model_distribution}
	\end{center}

	\begin{footnotesize}
		This figure displays the distribution of simulated stock prices under the benchmark scenario (short selling is permitted) and the alternative scenario (short selling is restricted). In Panel A, the price distribution under short-selling restrictions  is characterized by  $h_I^2/h_N < 1$, meaning that the proportion of {\it informed} investors holding the stock ($h_I$) is substantially lower than the fraction of {\it noised} investors holding the stock ($h_N$).  In Panel B, instead, the price distribution under short-selling restrictions  is characterized by  $h_I^2/h_N > 1$, meaning that the proportion of {\it informed} investors holding the stock is substantially higher than the fraction of {\it noised} investors holding the stock. We simulate 200 prices using $\alpha=0.5$, $p= 0.5$, $g= 0.9$, $h_N=0.3$, and $h_I = 0.2$ (so that $\smash{h_I^2/h_N < 1}$) or $h_I = 0.6$ (so that $\smash{h_I^2/h_N >1}$). A description of these parameters is presented in \mytabIA{tab:summary_notation}.	
			
	\end{footnotesize}
\end{figure}
\begin{landscape}
\begin{figure}
		\begin{center}
		\includegraphics[scale=1.00, angle = 0, trim = 00mm 00mm 00mm 00mm]{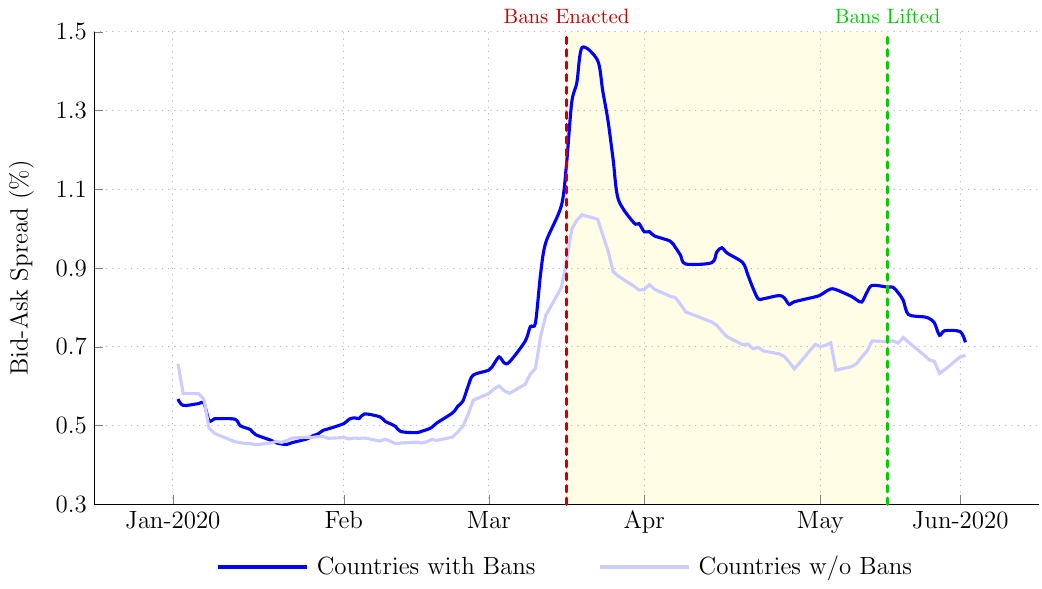}
		\vspace{0.5cm}
		\caption{\bf Bid-Ask Spreads and Short-Selling Bans} \label{fig:DiD_bas}
	\end{center}

	\begin{footnotesize}
		This figure shows the average percentage bid-ask spread of stocks traded in European markets around the introduction of temporary short-selling bans during the COVID-19 pandemic. The group of countries that restricted short sales includes Austria, Belgium, France, Greece, Italy, and Spain, whereas the group of countries that allowed short sales consists of Denmark, Finland, Germany, Ireland, Netherlands, Norway,  Poland, Portugal, Sweden, Switzerland, and the United Kingdom. Short-selling bans were introduced on March 17 and remained in place until May 18, 2020. Daily bid-ask spreads are averaged across countries with and without short-selling bans nd then smoothed using a five-day rolling average. The sample includes small, mid, and large cap stocks (micro and nano caps are excluded), and runs between January 2 and June 2, 2020. Data are collected from {\it Datastream}.	 
							
	\end{footnotesize}

\end{figure}
\end{landscape}
\begin{landscape}
	
	\begin{figure}
		\begin{center}
		\includegraphics[scale=1.00, angle = 0, trim = 00mm 00mm 00mm 00mm]{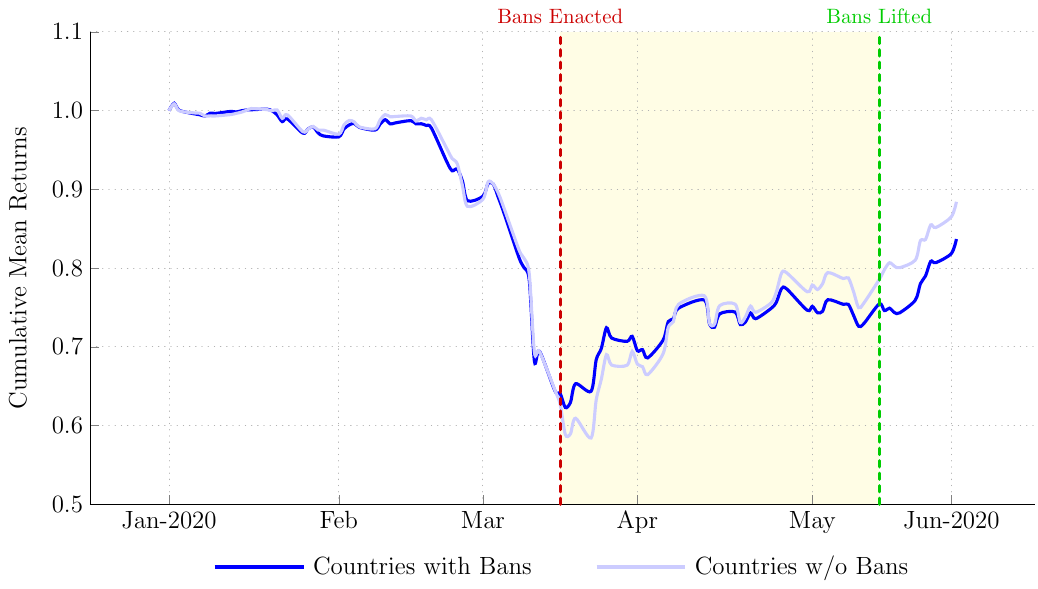}
		\vspace{0.5cm}
		\caption{\bf Mean Returns and Short-Selling Bans} \label{fig:DiD_mean}
	\end{center}

	\begin{footnotesize}
	This figure shows the cumulative mean return of stocks traded in European markets around the introduction of temporary short-selling bans during the COVID-19 pandemic. The group of countries that restricted short sales includes Austria, Belgium, France, Greece, Italy, and Spain, whereas the group of countries that allowed short sales consists of Denmark, Finland, Germany, Ireland, Netherlands, Norway,  Poland, Portugal, Sweden, Switzerland, and the United Kingdom. Short-selling bans were introduced on March 17 and remained in place until May 18, 2020. Returns are based on daily total return indices expressed in dollar terms before taking the mean value across all stocks in countries with and without short-selling bans, respectively. The sample includes small, mid, and large cap stocks (micro and nano caps are excluded), and runs between January 2 and June 2, 2020. Data are collected from {\it Datastream}.			 
	
\end{footnotesize}

\end{figure}
\end{landscape}
\begin{landscape}
	
	\begin{figure}
		\begin{center}
		\includegraphics[scale=1.00, angle = 0, trim = 00mm 00mm 00mm 00mm]{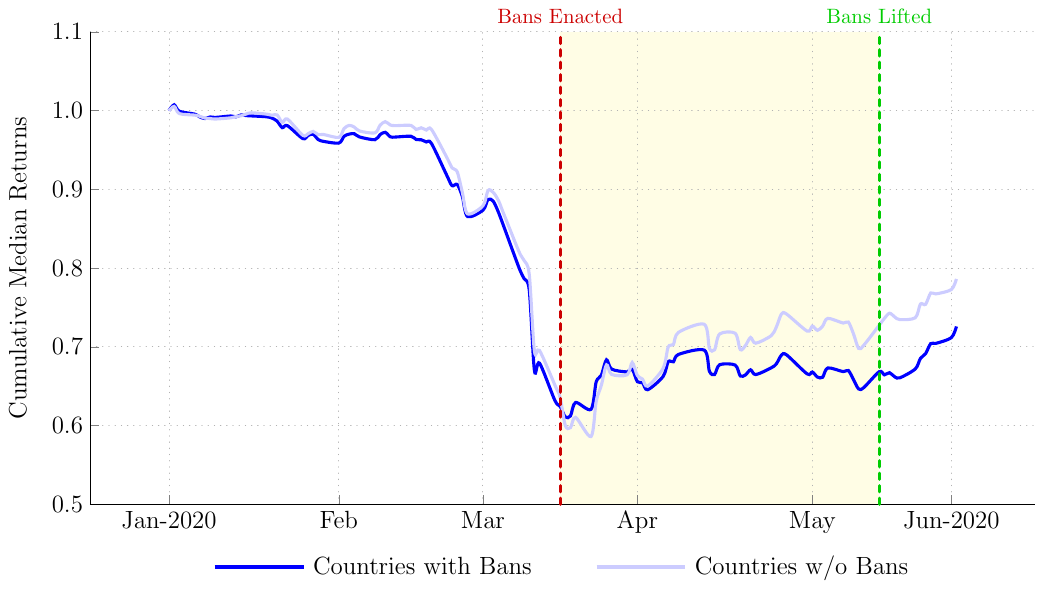}
		\vspace{0.5cm}
		\caption{\bf Median Returns and Short-Selling Bans} \label{fig:DiD_median}
	\end{center}

	\begin{footnotesize}
	This figure shows the cumulative median return of stocks traded in European markets around the introduction of temporary short-selling bans during the COVID-19 pandemic. The group of countries that restricted short sales includes Austria, Belgium, France, Greece, Italy, and Spain, whereas the group of countries that allowed short sales consists of Denmark, Finland, Germany, Ireland, Netherlands, Norway,  Poland, Portugal, Sweden, Switzerland, and the United Kingdom. Short-selling bans were introduced on March 17 and remained in place until May 18, 2020. Returns are based on daily total return indices expressed in dollar terms before taking the median value across all stocks in countries with and without short-selling bans, respectively. The sample includes small, mid, and large cap stocks (micro and nano caps are excluded), and runs between January 2 and June 2, 2020. Data are collected from {\it Datastream}.	
	
\end{footnotesize}

\end{figure}
\end{landscape}
\begin{landscape}
\begin{figure}
		\begin{center}
		\includegraphics[scale=1.00, angle = 0, trim = 00mm 00mm 00mm 00mm]{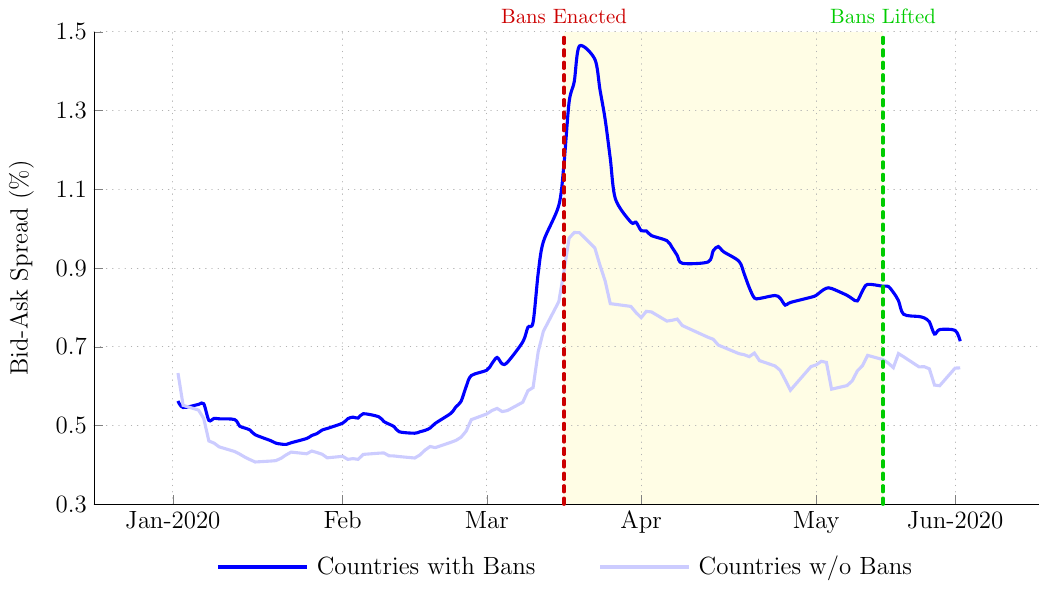}
		\vspace{0.5cm}
		\caption{\bf Bid-Ask Spreads and Short-Selling Bans: Matched Sample} \label{fig:DiD_bas_matched}
	\end{center}

	\begin{footnotesize}
	This figure shows the average percentage bid-ask spread of stocks traded in European markets around the introduction of temporary short-selling bans during the COVID-19 pandemic. The group of countries that restricted short sales includes Austria, Belgium, France, Greece, Italy, and Spain, whereas the group of countries that allowed short sales consists of Denmark, Finland, Germany, Ireland, Netherlands, Norway,  Poland, Portugal, Sweden, Switzerland, and the United Kingdom. Short-selling bans were introduced on March 17 and remained in place until May 18, 2020. Daily bid-ask spreads are averaged across countries with and without short-selling bans and then smoothed using a five-day rolling average. The sample includes small, mid, and large cap stocks (micro and nano caps are excluded) matched by market capitalization and industry classification according to the Industry Classification Benchmark (ICB) code, and runs between January 2 and June 2, 2020. Data are collected from {\it Datastream}.	 
	
\end{footnotesize}

\end{figure}
\end{landscape}
\begin{landscape}
	
	\begin{figure}
		\begin{center}
		\includegraphics[scale=1.00, angle = 0, trim = 00mm 00mm 00mm 00mm]{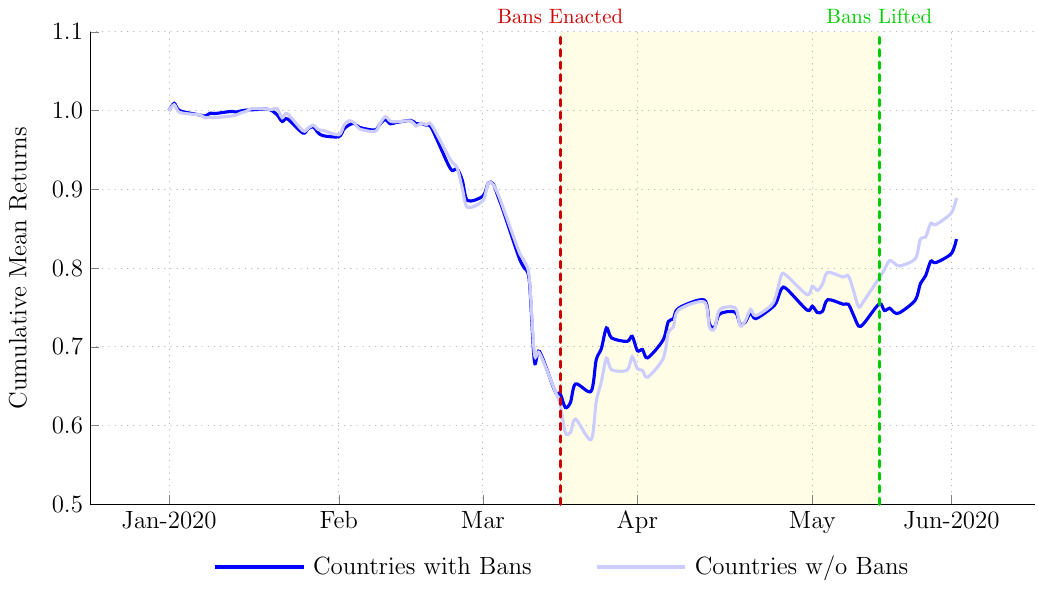}
		\vspace{0.5cm}
		\caption{\bf Mean Returns and Short-Selling Bans: Matched Sample} \label{fig:DiD_mean_matched}
	\end{center}

	\begin{footnotesize}
	This figure shows the cumulative mean return of stocks traded in European markets around the introduction of temporary short-selling bans during the COVID-19 pandemic. The group of countries that restricted short sales includes Austria, Belgium, France, Greece, Italy, and Spain, whereas the group of countries that allowed short sales consists of Denmark, Finland, Germany, Ireland, Netherlands, Norway,  Poland, Portugal, Sweden, Switzerland, and the United Kingdom. Short-selling bans were introduced on March 17 and remained in place until May 18, 2020. Returns are based on daily total return indices expressed in dollar terms before taking the mean value across all stocks in countries with and without short-selling bans, respectively. The sample includes small, mid, and large cap stocks (micro and nano caps are excluded) matched by market capitalization and industry classification according to the Industry Classification Benchmark (ICB) code, and runs between January 2 and June 2, 2020. Data are collected from {\it Bloomberg} and {\it Datastream}.
	
\end{footnotesize}

\end{figure}
\end{landscape}
\begin{landscape}
	
	\begin{figure}
		\begin{center}
		\includegraphics[scale=1.00, angle = 0, trim = 00mm 00mm 00mm 00mm]{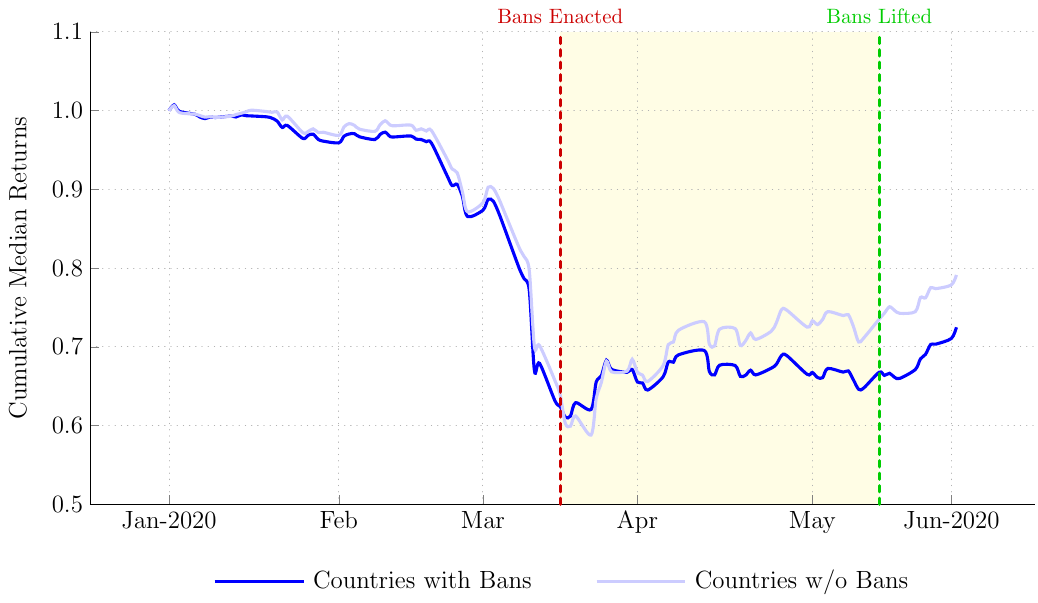}
		\vspace{0.5cm}
		\caption{\bf Median Returns and Short-Selling Bans: Matched Sample} \label{fig:DiD_median_matched}
	\end{center}

	\begin{footnotesize}
	This figure shows the cumulative median return of stocks traded in European markets around the introduction of temporary short-selling bans during the COVID-19 pandemic. The group of countries that restricted short sales includes Austria, Belgium, France, Greece, Italy, and Spain, whereas the group of countries that allowed short sales consists of Denmark, Finland, Germany, Ireland, Netherlands, Norway,  Poland, Portugal, Sweden, Switzerland, and the United Kingdom. Short-selling bans were introduced on March 17 and remained in place until May 18, 2020. Returns are based on daily total return indices expressed in dollar terms before taking the median value across all stocks in countries with and without short-selling bans, respectively. The sample includes small, mid, and large cap stocks (micro and nano caps are excluded) matched by market capitalization and industry classification according to the Industry Classification Benchmark (ICB) code, and runs between January 2 and June 2, 2020. Data are collected from {\it Bloomberg} and {\it Datastream}.
	
	\end{footnotesize}

\end{figure}
\end{landscape}
\begin{landscape}
\begin{figure}
		\begin{center}
		\includegraphics[scale=1.00, angle = 0, trim = 00mm 00mm 00mm 00mm]{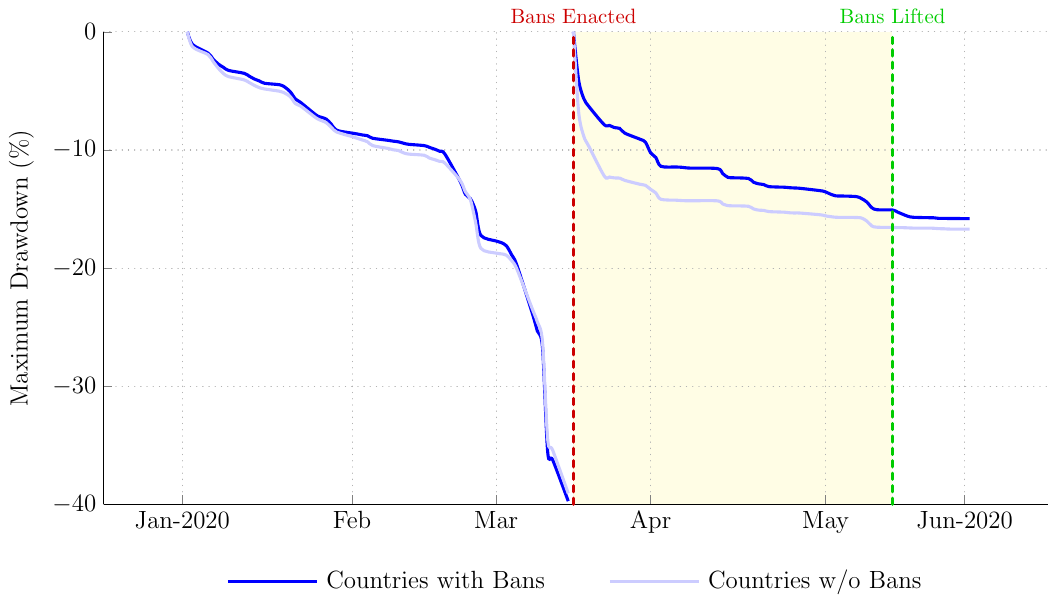}
		\vspace{0.5cm}
		\caption{\bf Maximum Drawdown and Short-Selling Bans: Matched Sample} \label{fig:DiD_mdd_matched}
	\end{center}

	\begin{footnotesize}
		This figure shows the average maximum drawdown of stocks traded in European markets around the introduction of  temporary short-selling bans during the COVID-19 pandemic. The group of countries that restricted short sales includes Austria, Belgium, France, Greece, Italy, and Spain, whereas the group of countries that allowed short sales consists of Denmark, Finland, Germany, Ireland, Netherlands, Norway,  Poland, Portugal, Sweden, Switzerland, and the United Kingdom. Short-selling bans were introduced on March 17 and remained in place until May 18, 2020. The maximum drawdown employs an expanding window from January 2 to March 16, 2020, and from March 17 to June 2, 2020, respectively, and is based on daily total return indices expressed in dollar terms before taking the mean value across all stocks in countries with and without short-selling bans, respectively. The sample includes small, mid, and large cap stocks (micro and nano caps are excluded) matched by market capitalization and industry classification according to the Industry Classification Benchmark (ICB) code, and runs between January 2 and June 2, 2020. Data are collected from {\it Bloomberg} and {\it Datastream}.	
		
	\end{footnotesize}

\end{figure}
\end{landscape}
\begin{landscape}
\begin{figure}
		\begin{center}
		\includegraphics[scale=1.00, angle = 0, trim = 00mm 00mm 00mm 00mm]{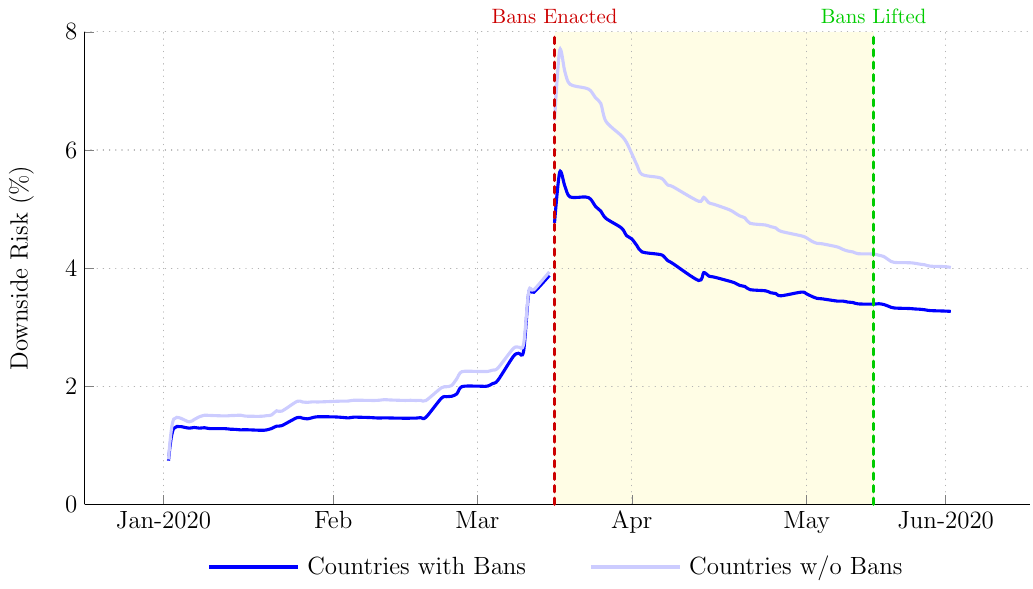}
		\vspace{0.5cm}
		\caption{\bf Downside Risk and Short-Selling Bans} \label{fig:DiD_down}
	\end{center}

	\begin{footnotesize}
	This figure shows the average downside risk of stocks traded in European markets around the introduction of  temporary short-selling bans during the COVID-19 pandemic. The group of countries that restricted short sales includes Austria, Belgium, France, Greece, Italy, and Spain, whereas the group of countries that allowed short sales consists of Denmark, Finland, Germany, Ireland, Netherlands, Norway,  Poland, Portugal, Sweden, Switzerland, and the United Kingdom.  Short-selling bans were introduced on March 17 and remained in place until May 18, 2020. The downside risk employs an expanding window from January 2 to March 16, 2020, and from March 17 to June 2, 2020, respectively, and is constructed as the volatility of negative returns on daily total return indices expressed in dollar terms before taking the mean value across all stocks in countries with and without short-selling bans, respectively. The sample includes small, mid, and large cap stocks (micro and nano caps are excluded), and runs between January 2 and June 2, 2020. Data are collected from {\it Datastream}.	
	
\end{footnotesize}
\end{figure}
\end{landscape}
\begin{landscape}
\begin{figure}
		\begin{center}
		\includegraphics[scale=1.00, angle = 0, trim = 00mm 00mm 00mm 00mm]{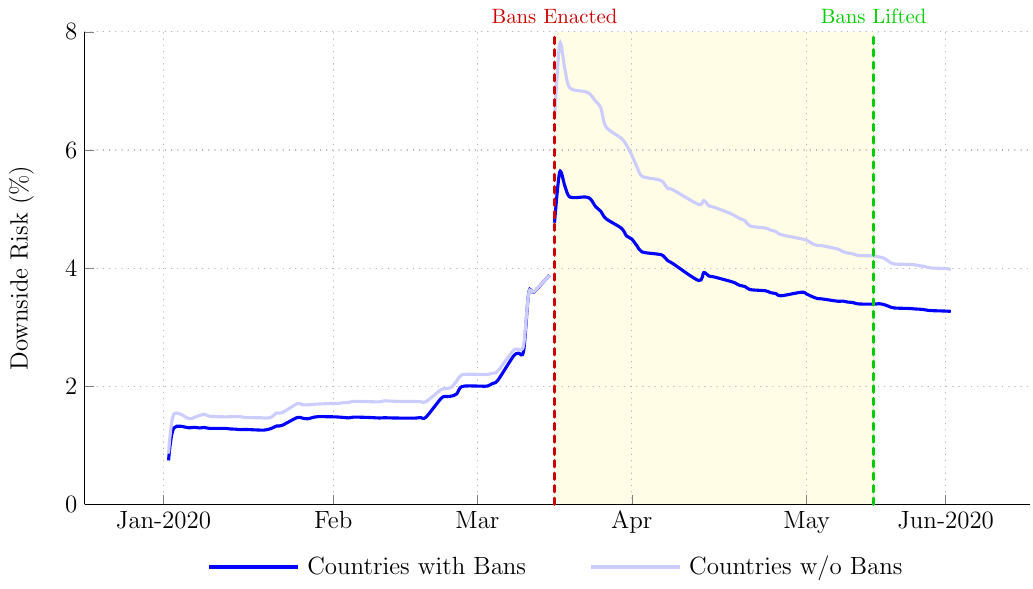}
		\vspace{0.5cm}
		\caption{\bf Downside Risk and Short-Selling Bans: Matched Sample} \label{fig:DiD_down_matched}
	\end{center}

	\begin{footnotesize}
		This figure shows the average downside risk of stocks traded in European markets around the introduction of  temporary short-selling bans during the COVID-19 pandemic. The group of countries that restricted short sales includes Austria, Belgium, France, Greece, Italy, and Spain, whereas the group of countries that allowed short sales consists of Denmark, Finland, Germany, Ireland, Netherlands, Norway,  Poland, Portugal, Sweden, Switzerland, and the United Kingdom.  Short-selling bans were introduced on March 17 and remained in place until May 18, 2020. The downside risk employs an expanding window from January 2 to March 16, 2020, and from March 17 to June 2, 2020, respectively, and is constructed as the volatility of negative returns on daily total return indices expressed in dollar terms before taking the mean value across all stocks in countries with and without short-selling bans, respectively. The sample includes small, mid, and large cap stocks (micro and nano caps are excluded) matched by market capitalization and industry classification according to the Industry Classification Benchmark (ICB) code, and runs between January 2 and June 2, 2020. Data are collected from {\it Bloomberg} and {\it Datastream}.
		
	\end{footnotesize}

\end{figure}
\end{landscape}

\begin{table}
\caption{\bf Summary of the Notation} \label{tab:summary_notation}
\begin{footnotesize}
This table summarizes the model's notation used in \mysec{sec:bans_model} and \mysecIA{app:proof_simulations}.

    \end{footnotesize}
\bigskip

\newcolumntype{L}[1]{>{\raggedright\let\newline\\\arraybackslash\hspace{0pt}}m{#1}}
\newcolumntype{C}[1]{>{\centering\let\newline\\\arraybackslash\hspace{0pt}}m{#1}}
\newcolumntype{R}[1]{>{\raggedleft\let\newline\\\arraybackslash\hspace{0pt}}m{#1}}
\setlength\extrarowheight{2pt}

\robustify\bfseries
\centering
\scalebox{0.90}{
	
	\begin{tabular}{C{2.0cm}m{0.9\linewidth}}
	    \toprule
	    \multicolumn{1}{c}{\bf Variable} & \multicolumn{1}{c}{\bf Definition} \\
 		\midrule 		
 		\addlinespace[5pt]
 		
 		$V$ & {Liquidation value of the risky asset, either zero or one.} \\
 		\addlinespace[5pt]
 		
 		$p$   & {Probability that $V$ is equal to one.} \\
 		\addlinespace[5pt]

 		$1- p$   & {Probability that $V$ is equal to zero.} \\
        \addlinespace[5pt]

		$g$   & {Probability that a trader is active either for information or liquidity reasons. He submits a buy order or a sell-or-short order of a single stock.} \\
		\addlinespace[5pt]

		$1-g$   & {Probability that a trader does not actively participate in the market, so no trade is observed.} \\
\addlinespace[5pt]

		$\alpha$   & {Probability that a trader is {\it informed} (or the fraction of {\it informed} traders among those active in the market).} \\
		\addlinespace[5pt]
		
		$1-\alpha$   & {Probability that a given trader is {\it noise} (or the fraction of {\it noise} traders among those active in the market).} \\
        \addlinespace[5pt]

 		$h_I$    & {Probability that an {\it informed} trader owns the stock.}  \\
 		\addlinespace[5pt]
 		
 		$h_N$ & {Probability that an {\it noise} trader owns the stock.}  \\
 		\addlinespace[5pt]

 		$\eta_s$ & {Probability that an {\it noise} trader submits a sell-or-short order, which is known to {\it market makers}.} \\
 		\addlinespace[5pt] 		
 		
 		$f(\eta_s)$ & {Probability distribution of $\eta_s$, which is known to the {\it regulator}.} \\
 		\addlinespace[5pt]

 		$q$ & {Equilibrium price of the risky asset from the {\it regulator}'s perspective when short selling is allowed.} \\ %
         
         \addlinespace[5pt]

		$\tilde{q}$ & {Equilibrium price of the risky asset from the {\it regulator}'s perspective when short selling is prohibited.} \\ %

\addlinespace[5pt]

 		$c$ & {A sufficiently low threshold for the stock price set by the {\it regulator}.} \\

        \addlinespace[5pt]
 		
 		$P(q<c)$ & {Probability that $q$ falls below $c$, as perceived by the {\it regulator}, when short selling is allowed.}\\
 		
 		\addlinespace[5pt]

  		$P(\tilde{q}<c)$ & {Probability that $\tilde{q}$ falls below $c$, as perceived by the {\it regulator}, when short selling is prohibited.}\\

        \addlinespace[5pt]

$x$ & {Confidence level set by the {\it regulator}.} \\

	    \bottomrule
	
\end{tabular}
}
\end{table}

\begin{table}
\caption{\bf Timeline of Short-Selling Bans} \label{tab:timeline_restrictions}
	\begin{footnotesize}
	This table describes the decisions taken by the National Regulatory Authority of Austria, Belgium, France, Greece, Italy, and Spain as well as the corresponding opinions released by the  European Securities and Markets Authority (ESMA) between March and May 2020. 
	
    \end{footnotesize}
\bigskip

\newcolumntype{L}[1]{>{\raggedright\let\newline\\\arraybackslash\hspace{0pt}}m{#1}}
\newcolumntype{C}[1]{>{\centering\let\newline\\\arraybackslash\hspace{0pt}}m{#1}}
\newcolumntype{R}[1]{>{\raggedleft\let\newline\\\arraybackslash\hspace{0pt}}m{#1}}
\setlength\extrarowheight{2pt}

\robustify\bfseries
\centering
\scalebox{0.83}{
	
	\begin{tabular}{C{2.5cm}m{1\linewidth}}
		\toprule
		\multicolumn{1}{c}{\bf Country} & \multicolumn{1}{c}{\bf Actions Taken Between March and May 2020} \\
		\midrule 		
		\addlinespace[5pt]
		
		{\bf Austria} & {On March 18, the Austrian Financial Market Authority (FMA) issued a temporary ban on stocks traded on the Vienna Stock Exchange. The ban began on the same date and was expected to expire on April 18. FMA's decision is available \href{https://www.fma.gv.at/en/fma-issues-a-regulation-prohibiting-short-selling-in-certain-financial-instruments-that-are-listed-on-the-vienna-stock-exchange/}{here} and ESMA's opinion \href{https://www.esma.europa.eu/document/opinion-fma-emergency-measure-under-ssr}{here}. Later, the ban was renewed until May 18 with ESMA's opinion available \href{https://www.esma.europa.eu/document/opinion-fma-emergency-measure-under-ssr-0}{here}.} \\
		\addlinespace[5pt]
		\midrule
		\addlinespace[5pt]		
		{\bf Belgium} & {On March 16, the Belgian Financial Services and Markets Authority (FSMA) issued a temporary ban on stocks traded on Belgian trading venues (Euronext Brussels and Euronext Growth). The ban began on March 17 and was expected to expire on April 17. FSMA's decision is available \href{https://www.fsma.be/en/news/short-selling}{here} and ESMA's opinion  \href{https://www.esma.europa.eu/document/opinion-fsma-emergency-measure-under-ssr}{here}. Later, the ban was renewed until May 18 with ESMA's opinion available \href{https://www.esma.europa.eu/document/opinion-fsma-emergency-measure-under-ssr-0}{here}.} \\
		\addlinespace[5pt]
		\midrule
		\addlinespace[5pt]		
		{\bf France} & {On March 17, the Autorit\'e des March\'e Financiers (AMF) issued a temporary ban on stocks traded on French  trading venues (Euronext Paris, Euronext Growth Paris, and Euronext Access Paris). The ban began on the same date and was expected to expire on April 16. AMF's decision is available \href{https://www.amf-france.org/en/news-publications/news-releases/amf-news-releases/amf-announces-temporary-short-selling-ban-certain-shares-until-end-trading-day-march-17}{here} and ESMA's opinion \href{https://www.esma.europa.eu/document/opinion-amf-emergency-measure-under-ssr}{here}. Later, the ban was renewed until May 18 with ESMA's opinion available \href{https://www.esma.europa.eu/document/opinion-amf-emergency-measure-under-ssr-0}{here}.} \\
		\addlinespace[5pt]
		\midrule
		\addlinespace[5pt]		
		{\bf Greece} & {On March 17, the Hellenic Capital Markets Commission (HCMC) issued a temporary ban on stocks traded on the Athens Stock Exchange. The ban came into force on March 18 and was expected to expire on April 24. HCMC's decision is available \href{http://www.hcmc.gr/vdrv/elib/af9dc6d61-5b88-4379-866e-f701707eeedc-92668751-0}{here} and ESMA's opinion \href{https://www.esma.europa.eu/document/opinion-hcmc-emergency-measure-under-ssr}{here}. Later, the ban was renewed until May 18 with ESMA's opinion available \href{https://www.esma.europa.eu/document/opinion-hcmc-emergency-measure-under-ssr-0}{here}.} \\
		\addlinespace[5pt]
		\midrule
		\addlinespace[5pt]		
		{\bf Italy} & {On March 17, the Italian Commissione Nazionale per le Societ\'a e la Borsa (CONSOB) issued a temporary ban on stocks traded on the Mercato Telematico Azionario (MTA). The ban came into force on March 18 and was expected to expire on June 20. CONSOB's decision is available \href{https://www.consob.it/web/consob-and-its-activities/bullettin/documenti/english/resolutions/res21303.htm}{here} and ESMA's opinion \href{https://www.esma.europa.eu/document/opinion-consob-emergency-measure-under-ssr}{here}. On May 15, CONSOB decided to terminate the ban early on May 18, in line with the other European countries as stated \href{https://www.consob.it/web/consob-and-its-activities/bullettin/documenti/english/resolutions/res21367.htm}{here}.} \\
		\addlinespace[5pt]
		\midrule
		\addlinespace[5pt]		
		{\bf Spain} & {On March 16, the Comisi\'on Nacional del Mercado de Valores (CNMV) issued a temporary ban on stocks traded on Spanish trading venues (Bolsa de Madrid, Bolsa de Barcelona, Bolsa de Valencia, Bolsa de Bilbao, and Mercado Alternativo Burs\'atil). The ban came into force on March 17 and was expected to expire on April 17. CNMV's decision is available \href{https://www.cnmv.es/portal/verDoc.axd?t=\%7b5baf609e-ed4e-4dad-a697-80c55548e181\%7d}{here} and ESMA's opinion \href{https://www.esma.europa.eu/document/opinion-cnmv-emergency-measure-under-ssr}{here}. Later, the ban was renewed until May 18 with ESMA's opinion available \href{https://www.esma.europa.eu/document/opinion-cnmv-emergency-measure-under-ssr-0}{here}.} \\

	    \bottomrule

\end{tabular}
}
\end{table}
\begin{table}[ht]
	\caption{\bf  Bid-Ask Spreads and Short-Selling Bans}\label{tab:DiD_bas}
	\begin{footnotesize}
		This table presents difference-in-differences estimates associated with the introduction of temporary short-selling bans in European stock markets during the COVID-19 pandemic. The dependent variable is the percentage bid-ask spread based on daily bid and ask prices. \emph{Country} is a dummy variable that equals one for the group of countries that introduced short-selling bans (i.e., Austria, Belgium, France, Greece, Italy, and Spain), and zero for the control group of countries (i.e., Denmark, Finland, Germany, Ireland, Netherlands, Norway,  Poland, Portugal, Sweden, Switzerland, and the United Kingdom). \emph{Ban} is a dummy variable that equals one (zero) for a post-treatment (pre-treatment) period of one month that goes from March 17 to April 15, 2020 (February 17 to March 16, 2020). The set of controls includes firm size and sovereign CDS spread. Specifications are complemented with stock and time (calendar date) fixed effects {\it fe}. Standard errors (in parentheses) are clustered by stock and time (calendar date) dimensions. *, **, ***, indicate statistical significance at the 10\%, 5\%, and 1\% level, respectively. The sample includes small, mid, and large cap stocks (micro and nano caps are excluded). Data are collected from {\it Datastream}.
				
	\end{footnotesize}

		\bigskip

	\newcolumntype{L}[1]{>{\raggedright\let\newline\\\arraybackslash\hspace{0pt}}m{#1}}
	\newcolumntype{C}[1]{>{\centering\let\newline\\\arraybackslash\hspace{0pt}}m{#1}}
	\newcolumntype{R}[1]{>{\raggedleft\let\newline\\\arraybackslash\hspace{0pt}}m{#1}}
	\setlength\extrarowheight{0pt}

	\sisetup{
		input-symbols         = [()],
		table-format          = 1.0,
		table-space-text-post = ***,
		table-align-text-post = false,
		table-text-alignment  = center,
		group-digits          = false
	}

	\robustify\bfseries
	\centering
	\scalebox{1.00}{
		
		\begin{tabular}{L{6.0cm}
				S[table-column-width = 2.0cm]
				S[table-column-width = 2.0cm]
				S[table-column-width = 2.0cm]
				S[table-column-width = 2.0cm]
			}
	
	\toprule
	
	  \multicolumn{1}{l}{}
	& \multicolumn{1}{c}{(1)} 
	& \multicolumn{1}{c}{(2)} 
	& \multicolumn{1}{c}{(3)}  
	& \multicolumn{1}{c}{(4)} \\
	\midrule
	
	\addlinespace[5pt]

		\rc{\it Ban $\times$ Country}   & \bl 0.116\fn{***} & \bl 0.130\fn{***} & \bl 0.130\fn{***} & \bl 0.120\fn{**} \\
		\rc & (0.040) & (0.044) & (0.044) & (0.048) \\		
		\addlinespace[5pt]
		
		 {\it Ban} & 0.247\fn{***} & 0.246\fn{***} &       &  \\
		& (0.045) & (0.046) &       &  \\		
		\addlinespace[5pt]
		
		{\it Country} & 0.106\fn{**} &       &       &  \\
		& (0.051) &       &       &  \\
		\addlinespace[5pt]

		{\it Constant} & 0.619\fn{***} & 0.647\fn{***} & 0.767\fn{***} & 1.646\fn{**} \\
		& (0.044) & (0.042) & (0.006) & (0.625) \\
		\addlinespace[5pt]
		
		{$R^2$} & 0.012 & 0.562 & 0.574 & 0.574  \\
		\cmidrule{2-5}
		{$\#$ \it obs} & {77,667} & {77,665} & {77,665} & {77,665} \\
		\midrule	
		{\it Stock fe} &       & {\checkmark}      & {\checkmark}      &     {\checkmark}  \\
		{\it Time fe}  &       &                   & {\checkmark}      &     {\checkmark}  \\
		{\it Controls} &       &                   &                   &     {\checkmark}  \\

	\bottomrule

\end{tabular}%
}

\end{table}
 
\begin{table}[ht]
	\caption{\bf Bid-Ask Spreads and IV Regressions}\label{tab:IV_bas}

	\begin{footnotesize}		
	This table presents instrumental variable regressions associated with the introduction of temporary short-selling bans in European stock markets during the COVID-19 pandemic. In Panel A, the dependent variable is \emph{Ban}, a dummy that equals one during the ban period that goes from March 17 to May 15, 2020, and zero otherwise. In Panel B, the dependent variable is the percentage bid-ask spread based on daily bid and ask prices. The instruments for the ban decision are the monthly financial stress index in logs and the monthly average sovereign CDS spread, both measured over the previous month. The controls are the firm size (stock market capitalization in logs) and firm volatility (the standard deviation of daily stock returns over a 21-day rolling window). Both Panels include stock and time (calendar date) fixed effects {\it fe}. Standard errors (in parentheses) are clustered by stock and time (calendar date) dimensions. *, **, ***, indicate statistical significance at the 10\%, 5\%, and 1\% level, respectively. The critical values for the Kleibergen-Paap and Cragg-Donald Wald F statistics are from \citet{stock_yogo:2015}. The significance of Kleibergen-Paap LM statistic is based on the $\chi^2_2$. The constants are unreported to save space. The sample includes small, mid, and large cap stocks (micro and nano caps are excluded), and runs daily between January 2 and June 2, 2020. Data are collected from {\it Datastream}, {\it Bloomberg} and the {\it ECB Data Portal}.
		 
\end{footnotesize}  

		\bigskip
		
	\newcolumntype{L}[1]{>{\raggedright\let\newline\\\arraybackslash\hspace{0pt}}m{#1}}
	\newcolumntype{C}[1]{>{\centering\let\newline\\\arraybackslash\hspace{0pt}}m{#1}}
	\newcolumntype{R}[1]{>{\raggedleft\let\newline\\\arraybackslash\hspace{0pt}}m{#1}}
	\setlength\extrarowheight{0pt}

	\sisetup{
		input-symbols         = [()],
		table-format          = 1.0,
		table-space-text-post = ***,
		table-align-text-post = false,
		table-text-alignment  = center,
		group-digits          = false
	}

	\robustify\bfseries
	\centering
	\scalebox{0.85}{
		
		\begin{tabular}{L{8cm}
				S[table-column-width = 1.9cm]
				S[table-column-width = 1.9cm]
				S[table-column-width = 1.9cm]
			}
	
	\toprule
	\multicolumn{4}{c}{\bf Panel A: First-Stage Regression} \\
	\midrule
	& \multicolumn{1}{c}{(1)} 
	& \multicolumn{1}{c}{(2)} 
	& \multicolumn{1}{c}{(3)}  \\
	\cmidrule(l){2-4}

	\rc{\it Sovereign CDS Spread} & \bl0.005\fn{***} & \bl0.005\fn{***} & \bl0.005\fn{***} \\
	\rc& (0.001) & (0.001) & (0.001) \\	
	\addlinespace[0pt]

	\rc{\it Financial Stress Index} & \bl0.263\fn{***} & \bl0.261\fn{***} & \bl0.261\fn{***} \\
	\rc& (0.028) & (0.029) & (0.028) \\	
	\addlinespace[0pt]

	{\it Firm Size} &       & 0.049\fn{***} & 0.027\fn{*} \\
	&       & (0.014) & (0.015) \\
	\addlinespace[0pt]
	
	{\it Firm Volatility} &       &       & -0.795\fn{***} \\
	&       &       & (0.248) \\
	\addlinespace[0pt]

	{$R^2$} & 0.618 & 0.619 & 0.620  \\
	\cmidrule{2-4}
	{\it $\#$ obs} & {207,418} & {207,418} & {207,382}  \\
	\midrule
	\multicolumn{4}{c}{\bf Panel B: Second-Stage Regression} \\
	\midrule
	\addlinespace[2pt]

	\rc{\it Ban (estimated)} & \bl0.126\fn{**} & \bl0.129\fn{**} & \bl0.135\fn{***} \\
   \rc	& (0.052) & (0.051) & (0.051) \\
	\addlinespace[0pt]
	
	{\it Firm Size} &       & -0.212\fn{***} & -0.178\fn{***} \\
	&       & (0.058) & (0.055) \\
	\addlinespace[0pt]
	
	{\it Firm Volatility} &       &       & 1.207 \\
	&       &       & (0.753) \\
	\addlinespace[0pt]
	
	{$R^2$} & 0.001 & 0.003 & 0.003  \\
	\cmidrule{2-4}
	{\it $\#$ obs} & {199,087} & {199,087} & {199,051}  \\
	\midrule
	{\it Stock fe} & {\checkmark}     & {\checkmark}     & {\checkmark}     \\
	{\it Time fe}  & {\checkmark}     & {\checkmark}     & {\checkmark}     \\
	\midrule 
	
	{\it Kleibergen-Paap LM statistic} & 43.371\fn{***} & 42.857\fn{***} & 43.066\fn{***} \\
	{\it Kleibergen-Paap Wald F  statistic} & 40.947\fn{***} & 40.167\fn{***} & 40.425\fn{***} \\
	{\it Cragg-Donald Wald F statistic} & 24{,}074\fn{***} & 24{,}068\fn{***} & 23{,}898\fn{***} \\
	{\it Hansen J statistic} & 0.288 & 1.335 & 1.043 \\
	{\it Hansen $\chi^2_1$ pval} & (0.592) & (0.248) & (0.307) \\

	\bottomrule

\end{tabular}%
}

\end{table}

\begin{landscape}
	
\begin{table}[ht]
	\caption{\bf Descriptive Statistics on Bid-Ask Spreads}	\label{tab:data_summary_bas}
	
	\begin{footnotesize}
		This figure table presents summary statistics of percentage bid-ask spreads for European stock markets around the introduction of temporary short-selling bans during the COVID-19 pandemic. The group of countries that restricted short sales (i.e., Austria, Belgium, France, Greece, Italy, and Spain) are highlighted in gray. The pre-ban (post-ban) period covers one month between February 17 and March 16, 2020 (March 17 and April 15, 2020). Bid-ask spreads are first computed using end-of-day quotes and then averaged across stocks in each country. AC(1) denotes the first-order serial correlation. The sample includes small, mid, and large cap stocks (micro and nano caps are excluded). Data are collected from {\it Datastream}.
				
	\end{footnotesize}

	\bigskip

	\newcolumntype{L}[1]{>{\raggedright\let\newline\\\arraybackslash\hspace{0pt}}m{#1}}
	\newcolumntype{C}[1]{>{\centering\let\newline\\\arraybackslash\hspace{0pt}}m{#1}}
	\newcolumntype{R}[1]{>{\raggedleft\let\newline\\\arraybackslash\hspace{0pt}}m{#1}}
	\setlength\extrarowheight{2pt}

	\sisetup{
		input-symbols         = [()],
		table-format          = 1.0,
		table-space-text-post = ***,
		table-align-text-post = false,
		table-text-alignment  = center,
		group-digits          = false
	}

	\robustify\bfseries
	\centering
	\scalebox{0.90}{
		
    \begin{tabular}{L{3.5cm}
    		S[table-column-width = 1.3cm]
    		S[table-column-width = 1.3cm]
    		S[table-column-width = 1.3cm]
    		S[table-column-width = 1.3cm]
    		S[table-column-width = 1.3cm]
    		S[table-column-width = 1.3cm]
    		S[table-column-width = 1.3cm]
    		S[table-column-width = 1.3cm]
    		S[table-column-width = 1.3cm]
    		S[table-column-width = 1.3cm]
    	}
    	\toprule

		& \multicolumn{5}{c}{\bf Panel A: Pre-Ban}               & \multicolumn{5}{c}{\bf Panel B: Post-Ban} \\
		\cmidrule(r){2-6} \cmidrule(l){7-11}

		\multicolumn{1}{l}{Country}
		& \multicolumn{1}{c}{mean} 
		& \multicolumn{1}{c}{sdev} 
		& \multicolumn{1}{c}{skew} 
		& \multicolumn{1}{c}{kurt} 
		& \multicolumn{1}{c}{AC(1)} 
		& \multicolumn{1}{c}{mean} 
		& \multicolumn{1}{c}{sdev} 
		& \multicolumn{1}{c}{skew} 
		& \multicolumn{1}{c}{kurt} 
		& \multicolumn{1}{c}{AC(1)} \\		
		\midrule

	\rc \bl{Austria} & \bl0.70 & \bl0.88 & \bl3.74 & \bl23.82 & \bl0.49 & \bl1.08 & \bl1.13 & \bl3.36 & \bl19.75 & \bl0.36 \\
	\rc \bl{Belgium} & \bl0.88 & \bl1.48 & \bl5.41 & \bl57.57 & \bl0.60 & \bl1.42 & \bl2.01 & \bl3.34 & \bl18.80 & \bl0.54 \\
	{Denmark} & 0.36  & 0.52  & 5.55  & 50.16 & 0.51  & 0.47  & 0.55  & 5.52  & 68.25 & 0.44 \\
	{Finland} & 0.34  & 0.49  & 7.93  & 109.53 & 0.34  & 0.49  & 0.58  & 3.96  & 25.95 & 0.36 \\
	\rc \bl{France} & \bl0.83 & \bl1.52 & \bl4.61 & \bl36.55 & \bl0.61 & \bl1.27 & \bl1.98 & \bl4.08 & \bl29.14 & \bl0.65 \\
	{Germany} & 0.69  & 1.34  & 7.15  & 78.67 & 0.41  & 0.92  & 1.48  & 5.29  & 43.33 & 0.49 \\
	\rc \bl{Greece} & \bl1.17 & \bl1.54 & \bl4.07 & \bl32.02 & \bl0.42 & \bl1.34 & \bl1.52 & \bl2.75 & \bl13.58 & \bl0.43 \\
	{Ireland} & 1.97  & 3.25  & 3.17  & 14.52 & 0.58  & 2.86  & 4.27  & 2.52  & 9.35  & 0.69 \\
	\rc \bl{Italy} & \bl0.58 & \bl1.33 & \bl7.18 & \bl76.59 & \bl0.37 & \bl0.83 & \bl1.21 & \bl3.54 & \bl21.56 & \bl0.44 \\
	{Netherlands} & 0.29  & 0.54  & 5.25  & 44.84 & 0.50  & 0.47  & 0.80  & 4.93  & 38.80 & 0.49 \\
	{Norway} & 0.90  & 1.57  & 3.68  & 20.98 & 0.57  & 1.11  & 1.85  & 4.54  & 32.00 & 0.62 \\
	{Poland} & 0.85  & 1.10  & 3.70  & 24.92 & 0.46  & 0.81  & 0.96  & 3.07  & 17.61 & 0.46 \\
	{Portugal} & 0.67  & 1.71  & 4.76  & 28.47 & 0.66  & 1.13  & 3.04  & 4.72  & 28.69 & 0.61 \\
	\rc \bl{Spain} & \bl0.41 & \bl0.86 & \bl6.13 & \bl57.02 & \bl0.49 & \bl0.63 & \bl0.99 & \bl3.88 & \bl25.86 & \bl0.40 \\
	{Sweden} & 0.38  & 0.53  & 7.33  & 104.86 & 0.33  & 0.57  & 0.67  & 4.30  & 33.18 & 0.43 \\
	{Switzerland} & 0.45  & 0.87  & 8.31  & 132.62 & 0.39  & 0.67  & 1.12  & 4.62  & 36.42 & 0.46 \\
	{United Kingdom} & 0.71  & 1.38  & 3.29  & 16.60 & 0.70  & 1.05  & 2.05  & 3.74  & 22.15 & 0.74 \\
	\midrule
	{Total} & 0.65  & 1.28  & 5.51  & 51.82 & 0.58  & 0.93  & 1.66  & 4.84  & 38.04 & 0.65 \\
		\bottomrule
	\end{tabular}
}

\end{table}

\end{landscape}

\begin{landscape}
	
\begin{table}[ht]
	\caption{\bf Descriptive Statistics on Stock Returns}	\label{tab:data_summary_returns}
	
	\begin{footnotesize}
		This figure table presents summary statistics of percentage stock returns for European stock markets around the introduction of temporary short-selling bans during the COVID-19 pandemic. The group of countries that restricted short sales (i.e., Austria, Belgium, France, Greece, Italy, and Spain) are highlighted in gray. The pre-ban (post-ban) period covers one month between February 17 and March 16, 2020 (March 17 and April 15, 2020). Returns are based on daily total return indices expressed in dollar terms before taking the mean value across all stocks in each country. AC(1) denotes the first-order serial correlation. The sample includes small, mid, and large cap stocks (micro and nano caps are excluded). Data are collected from {\it Datastream}.	
			
	\end{footnotesize}
	
	\bigskip

	\newcolumntype{L}[1]{>{\raggedright\let\newline\\\arraybackslash\hspace{0pt}}m{#1}}
	\newcolumntype{C}[1]{>{\centering\let\newline\\\arraybackslash\hspace{0pt}}m{#1}}
	\newcolumntype{R}[1]{>{\raggedleft\let\newline\\\arraybackslash\hspace{0pt}}m{#1}}
	\setlength\extrarowheight{2pt}

	\sisetup{
		input-symbols         = [()],
		table-format          = 1.0,
		table-space-text-post = ***,
		table-align-text-post = false,
		table-text-alignment  = center,
		group-digits          = false
	}

	\robustify\bfseries
	\centering
	\scalebox{0.90}{
		
    \begin{tabular}{L{3.5cm}
    		S[table-column-width = 1.3cm]
    		S[table-column-width = 1.3cm]
    		S[table-column-width = 1.3cm]
    		S[table-column-width = 1.3cm]
    		S[table-column-width = 1.3cm]
    		S[table-column-width = 1.3cm]
    		S[table-column-width = 1.3cm]
    		S[table-column-width = 1.3cm]
    		S[table-column-width = 1.3cm]
    		S[table-column-width = 1.3cm]
    	}
    	\toprule

		& \multicolumn{5}{c}{\bf Panel A: Pre-Ban}               & \multicolumn{5}{c}{\bf Panel B: Post-Ban} \\
		\cmidrule(r){2-6} \cmidrule(l){7-11}

		\multicolumn{1}{l}{Country}
		& \multicolumn{1}{c}{mean} 
		& \multicolumn{1}{c}{sdev} 
		& \multicolumn{1}{c}{skew} 
		& \multicolumn{1}{c}{kurt} 
		& \multicolumn{1}{c}{AC(1)} 
		& \multicolumn{1}{c}{mean} 
		& \multicolumn{1}{c}{sdev} 
		& \multicolumn{1}{c}{skew} 
		& \multicolumn{1}{c}{kurt} 
		& \multicolumn{1}{c}{AC(1)} \\		
		\midrule

	\rc \bl{Austria} & \bl-2.21 & \bl4.48 & \bl-1.53 & \bl7.22 & \bl0.02 & \bl0.85 & \bl5.35 & \bl0.63 & \bl5.81 & \bl0.08 \\
	\rc \bl{Belgium} & \bl-1.61 & \bl4.13 & \bl-1.23 & \bl6.81 & \bl-0.09 & \bl0.60 & \bl4.66 & \bl0.50 & \bl6.51 & \bl0.01 \\
	{Denmark} & -1.52 & 3.92  & -0.98 & 7.92  & 0.07  & 0.53  & 4.11  & 0.01  & 4.32  & 0.00 \\
	{Finland} & -1.66 & 3.51  & -1.20 & 6.11  & 0.03  & 0.53  & 4.41  & 0.06  & 4.48  & -0.03 \\
	\rc \bl{France} & \bl-1.92 & \bl4.41 & \bl-1.47 & \bl11.45 & \bl0.00 & \bl0.58 & \bl5.34 & \bl0.51 & \bl10.21 & \bl0.01 \\
	{Germany} & -1.77 & 4.23  & -1.00 & 7.97  & -0.04 & 0.70  & 5.33  & 0.78  & 8.53  & 0.02 \\
	\rc \bl{Greece} & \bl-2.44 & \bl6.16 & \bl-0.59 & \bl4.81 & \bl-0.17 & \bl1.04 & \bl5.73 & \bl0.60 & \bl5.14 & \bl-0.20 \\
	{Ireland} & -1.63 & 3.86  & -1.31 & 5.16  & 0.00  & 0.02  & 6.52  & 0.31  & 7.46  & 0.12 \\
	\rc \bl{Italy} & \bl-2.02 & \bl5.08 & \bl-0.87 & \bl6.77 & \bl-0.27 & \bl0.57 & \bl4.29 & \bl0.52 & \bl4.94 & \bl0.07 \\
	{Netherlands} & -2.04 & 4.67  & -1.57 & 8.45  & -0.07 & 0.41  & 5.46  & 0.69  & 6.18  & 0.02 \\
	{Norway} & -2.65 & 6.31  & -1.34 & 9.02  & -0.04 & 0.68  & 6.86  & 1.45  & 23.26 & 0.14 \\
	{Poland} & -1.86 & 5.27  & -0.79 & 6.38  & 0.06  & 0.68  & 4.93  & 1.20  & 9.01  & -0.10 \\
	{Portugal} & -1.76 & 3.81  & -1.46 & 7.05  & -0.04 & 0.58  & 4.19  & 0.61  & 6.34  & -0.04 \\
	\rc \bl{Spain} & \bl-1.88 & \bl4.88 & \bl-0.79 & \bl10.31 & \bl-0.06 & \bl0.52 & \bl4.96 & \bl0.95 & \bl8.85 & \bl-0.02 \\
	{Sweden} & -2.01 & 4.55  & -1.30 & 7.37  & -0.02 & 0.76  & 5.75  & 0.35  & 6.12  & -0.02 \\
	{Switzerland} & -1.35 & 3.49  & -1.63 & 12.30 & 0.06  & 0.52  & 4.22  & 0.52  & 7.93  & -0.03 \\
	{United Kingdom} & -2.31 & 5.27  & -1.09 & 25.62 & 0.02  & 0.71  & 8.49  & 1.35  & 21.07 & 0.06 \\
	\midrule
	{Total} & -1.97 & 4.73  & -1.24 & 15.38 & -0.02 & 0.64  & 6.11  & 1.19  & 22.70 & 0.03 \\

		\bottomrule
	\end{tabular}
}

\end{table}

\end{landscape}

\begin{table}[ht]
	\thispagestyle{plain}
	\caption{\bf Bid-Ask Spreads and Country-Time Clustering}\label{tab:DiD_bas_sctcls}
	\begin{footnotesize}
		This table presents difference-in-differences estimates associated with the introduction of temporary short-selling bans in European stock markets during the COVID-19 pandemic. The dependent variable is the percentage bid-ask spread based on daily bid and ask prices. \emph{Country} is a dummy variable that equals one for the group of countries that introduced short-selling bans (i.e., Austria, Belgium, France, Greece, Italy, and Spain), and zero for the control group of countries (i.e., Denmark, Finland, Germany, Ireland, Netherlands, Norway,  Poland, Portugal, Sweden, Switzerland, and the United Kingdom). \emph{Ban} is a dummy variable that equals one (zero) for a post-treatment (pre-treatment) period of one month that goes from March 17 to April 15, 2020 (February 17 to March 16, 2020). The set of controls includes firm size and sovereign CDS spread. Specifications are complemented with stock and time (calendar date) fixed effects {\it fe}. Standard errors (in parentheses) are clustered by stock and country-time (calendar date) dimensions. *, **, ***, indicate statistical significance at the 10\%, 5\%, and 1\% level, respectively. The sample includes small, mid, and large cap stocks (micro and nano caps are excluded). Data are collected from {\it Datastream}.
				
	\end{footnotesize}

		\bigskip

	\newcolumntype{L}[1]{>{\raggedright\let\newline\\\arraybackslash\hspace{0pt}}m{#1}}
	\newcolumntype{C}[1]{>{\centering\let\newline\\\arraybackslash\hspace{0pt}}m{#1}}
	\newcolumntype{R}[1]{>{\raggedleft\let\newline\\\arraybackslash\hspace{0pt}}m{#1}}
	\setlength\extrarowheight{0pt}

	\sisetup{
		input-symbols         = [()],
		table-format          = 1.0,
		table-space-text-post = ***,
		table-align-text-post = false,
		table-text-alignment  = center,
		group-digits          = false
	}

	\robustify\bfseries
	\centering
	\scalebox{1.00}{
		
		\begin{tabular}{L{6.0cm}
				S[table-column-width = 2.0cm]
				S[table-column-width = 2.0cm]
				S[table-column-width = 2.0cm]
				S[table-column-width = 2.0cm]
			}
	
	\toprule
	
	  \multicolumn{1}{l}{}
	& \multicolumn{1}{c}{(1)} 
	& \multicolumn{1}{c}{(2)} 
	& \multicolumn{1}{c}{(3)}  
	& \multicolumn{1}{c}{(4)} \\
	\midrule
	
	\addlinespace[5pt]

	\rc{\it Ban $\times$ Country}   & \bl 0.116* &  \bl0.130** & \bl0.130*** & \bl0.120*** \\
	\rc & (0.063) & (0.053) & (0.037) & (0.040) \\
	\addlinespace[5pt]
	
	{\it Ban} & 0.247*** & 0.246*** &       &  \\
	& (0.033) & (0.025) &       &  \\
	\addlinespace[5pt]
	
	{\it Country} & 0.106* &       &       &  \\
	& (0.061) &       &       &  \\
	\addlinespace[5pt]

	{\it Constant} & 0.619*** & 0.647*** & 0.767*** & 1.646*** \\
	& (0.032) & (0.015) & (0.007) & (0.632) \\
	\addlinespace[5pt]
	{$R^2$} & 0.012 & 0.562 & 0.574 & 0.574  \\ 
	
	\cmidrule{2-5}
	{$\#$ \it obs}  & {77,667} & {77,665} & {77,665} & {77,665} \\
	\midrule	
	{\it Stock fe} &       & {\checkmark}      & {\checkmark}      &     {\checkmark}  \\
	{\it Time fe}  &       &                   & {\checkmark}      &     {\checkmark}  \\
	{\it Controls} &       &                   &                   &     {\checkmark}  \\
	\bottomrule

\end{tabular}%
}

\end{table} 
\begin{landscape}
\begin{table}[ht]
	\caption{\bf Institutional Ownership and Country-Time Clustering}\label{tab:DiD_bas_IO_sctcls}
	
	\begin{footnotesize}		
		This table presents difference-in-differences estimates associated with the introduction of temporary short-selling bans in European stock markets during the COVID-19 pandemic. The dependent variable is the percentage bid-ask spread based on daily bid and ask prices for stocks with low institutional ownership (bottom tercile) in Panel A and stocks with high institutional ownership (top tercile) in Panel B, respectively. \emph{Country} is a dummy variable that equals one for the group of countries that introduced short-selling bans (i.e., Austria, Belgium, France, Greece, Italy, and Spain), and zero for the control group of countries (i.e., Denmark, Finland, Germany, Ireland, Netherlands, Norway,  Poland, Portugal, Sweden, Switzerland, and the United Kingdom). \emph{Ban} is a dummy variable that equals one (zero) for a post-treatment (pre-treatment) period of one month that goes from March 17 to April 15, 2020 (February 17 to March 16, 2020). The set of controls includes firm size and sovereign CDS spread. Specifications are complemented with stock and time (calendar date) fixed effects {\it fe}. Standard errors (in parentheses) are clustered by stock and country-time (calendar date) dimensions. *, **, ***, indicate statistical significance at the 10\%, 5\%, and 1\% level, respectively. The sample includes small, mid, and large cap stocks (micro and nano caps are excluded). Data are collected from {\it Datastream} and {\it Bloomberg}.	
			
   	\end{footnotesize}
	
		\bigskip

	\newcolumntype{L}[1]{>{\raggedright\let\newline\\\arraybackslash\hspace{0pt}}m{#1}}
	\newcolumntype{C}[1]{>{\centering\let\newline\\\arraybackslash\hspace{0pt}}m{#1}}
	\newcolumntype{R}[1]{>{\raggedleft\let\newline\\\arraybackslash\hspace{0pt}}m{#1}}
	\setlength\extrarowheight{0pt}

	\sisetup{
		input-symbols         = [()],
		table-format          = 1.0,
		table-space-text-post = ***,
		table-align-text-post = false,
		table-text-alignment  = center,
		group-digits          = false
	}

	\robustify\bfseries
	\centering
	\scalebox{0.90}{
		
		\begin{tabular}{L{2.9cm}
				S[table-column-width = 1.9cm]
				S[table-column-width = 1.9cm]
				S[table-column-width = 1.9cm]
				S[table-column-width = 1.9cm]
				S[table-column-width = 1.9cm]
				S[table-column-width = 1.9cm]
				S[table-column-width = 1.9cm]
				S[table-column-width = 1.9cm]
			}
	
	\toprule

			& \multicolumn{4}{c}{\bf Panel A: Low Institutional Ownership} 
			& \multicolumn{4}{c}{\bf Panel B: High Institutional Ownership} \\
			\cmidrule(r){2-5} \cmidrule(l){6-9}

		 	& \multicolumn{1}{c}{(1)} 
			& \multicolumn{1}{c}{(2)} 
		    & \multicolumn{1}{c}{(3)} 
		 	& \multicolumn{1}{c}{(4)} 
			& \multicolumn{1}{c}{(5)} 
			& \multicolumn{1}{c}{(6)} 
			& \multicolumn{1}{c}{(7)}
			& \multicolumn{1}{c}{(8)} \\
			
			\midrule
			\addlinespace[5pt]

	\rc{\it Ban $\times$ Country}   & \bl 0.079 & \bl0.077 & \bl0.079 & \bl0.058 & \bl0.229** &\bl 0.285*** & \bl0.285*** & \bl0.278*** \\
	\rc & (0.091) & (0.084) & (0.060) & (0.068) & (0.090) & (0.079) & (0.069) & (0.074) \\	
	\addlinespace[5pt]
	
	{\it  Ban} & 0.307*** & 0.310*** &       &       & 0.206*** & 0.209*** &       &  \\
	& (0.057) & (0.051) &       &       & (0.040) & (0.031) &       &  \\
	\addlinespace[5pt]
	
	{\it Country} & 0.063 &       &       &       & 0.181 &       &       &  \\
	& (0.104) &       &       &       & (0.121) &       &       &  \\
	\addlinespace[5pt]
	
	{\it Constant} & 0.807*** & 0.833*** & 0.983*** & 0.186 & 0.613*** & 0.653*** & 0.754*** & 0.643 \\
	& (0.069) & (0.028) & (0.014) & (1.080) & (0.057) & (0.017) & (0.009) & (0.807) \\
	
	\addlinespace[5pt]
	
	{$R^2$} & 0.011 & 0.490 & 0.513 & 0.513 & 0.014 & 0.618 & 0.628 & 0.628 \\
	\cmidrule{2-9}
	{\it $\#$ obs} & {20,826} & {20,825} & {20,825} & {20,825} & {20,921} & {20,921} & {20,921} & {20,921}  \\
	\midrule
	{\it Stock fe} &       & {\checkmark}     & {\checkmark}     & {\checkmark}     &       & {\checkmark}     & {\checkmark}     & {\checkmark} \\
	{\it Time fe}  &       &                  & {\checkmark}     & {\checkmark}     &       &                  & {\checkmark}     & {\checkmark} \\
	{\it Controls} &       &                  &                  & {\checkmark}     &       &                  &                  & {\checkmark} \\
	\bottomrule

\end{tabular}%
}

\end{table}
\end{landscape}

\begin{landscape}
	\begin{table}[ht]
		\caption{\bf Data Description: Matched Sample} \label{tab:data_description_matched}
		\begin{footnotesize}
		This table describes stock market data for European venues. The countries highlighted in gray, i.e., Austria, Belgium, France, Greece, Italy, and Spain, introduced temporary short-selling bans between March 17 and May, 18, 2020. $\#$Small denotes the number of small-cap firm with an average market capitalization between 250 million and 2 billion US dollars, $\#$Mid refers to the number of mid-cap firms with an average market capitalization between 2 billion and 10 billion US dollars, whereas $\#$Large indicates the number of large-cap firm with an average market capitalization of \$10 billion US dollars or more.  The sample includes small, mid, and large cap stocks (micro and nano caps are excluded) matched by market capitalization and industry classification according to the Industry Classification Benchmark (ICB) code, and runs daily between January 2 and June 2, 2020. Data are collected from {\it Datastream} and {\it Bloomberg}.	
				
		\end{footnotesize}
		\bigskip

		\newcolumntype{L}[1]{>{\raggedright\let\newline\\\arraybackslash\hspace{0pt}}m{#1}}
		\newcolumntype{C}[1]{>{\centering\let\newline\\\arraybackslash\hspace{0pt}}m{#1}}
		\newcolumntype{R}[1]{>{\raggedleft\let\newline\\\arraybackslash\hspace{0pt}}m{#1}}
		\setlength\extrarowheight{2pt}

		\sisetup{
			input-decimal-markers  =  .,
			input-ignore           = {,},
			table-number-alignment = right,
			group-separator        ={,}, 
			group-four-digits      = true,
			input-symbols         = [()],
			table-text-alignment  = center,
			table-space-text-post = ***,
			table-align-text-post = false,
		}

		\robustify\bfseries
		\centering
		\scalebox{0.90}{
			
			\begin{tabular}{L{3.5cm}
					S[table-column-width = 3.0cm, table-format = 7.0]
					S[table-column-width = 2.2cm, table-format = 6.0]
					S[table-column-width = 2.2cm, table-format = 6.0]
					S[table-column-width = 2.2cm, table-format = 6.0]
					S[table-column-width = 2.2cm, table-format = 6.0]
				}
				\toprule
									   
			   	 &       &       &       \multicolumn{3}{c}{Market Capitalization} \\
			   	 \cmidrule{4-6}
			   	 \multicolumn{1}{l}{ Country}	
			   	 & \multicolumn{1}{c}{$\#$Obs} 
			   	 & \multicolumn{1}{c}{$\#$Stock} 
			   	 & \multicolumn{1}{c}{$\#$Small} 
			   	 & \multicolumn{1}{c}{$\#$Mid}   
			   	 & \multicolumn{1}{c}{$\#$Large}  \\			   	 
			   	 \cmidrule{1-6}
  
					\rc{Austria} & \bl3,672 & \bl34 & \bl16 & \bl16 & \bl2 \\
					\rc{Belgium} & \bl7,344 & \bl68 & \bl40 & \bl24 & \bl4 \\
					{Denmark} & 1,944 & 18    & 4     & 8     & 6 \\
					{Finland} & 2,052 & 19    & 13    & 4     & 2 \\
					\rc{France} & \bl24,038 & \bl223 & \bl118 & \bl61 & \bl44 \\
					{Germany} & 6,372 & 59    & 27    & 22    & 10 \\
					\rc{Greece} & \bl2,916 & \bl27 & \bl23 & \bl4  & \bl0 \\
					{Ireland} & 108   & 1     & 0     & 1     & 0 \\
					\rc{Italy} & \bl12,636 & \bl117 & \bl68 & \bl38 & \bl11 \\
					{Netherlands} & 2,268 & 21    & 9     & 7     & 5 \\
					{Norway} & 2,700 & 25    & 18    & 5     & 2 \\
					{Poland} & 1,512 & 14    & 11    & 3     & 0 \\
					{Portugal} & 432   & 4     & 2     & 0     & 2 \\
					\rc{Spain} & \bl9,072 & \bl84 & \bl46 & \bl23 & \bl15 \\
					{Sweden} & 5,940 & 55    & 32    & 19    & 4 \\
					{Switzerland} & 5,613 & 52    & 32    & 14    & 6 \\
					{United Kingdom} & 15,118 & 140   & 86    & 39    & 15 \\
					\midrule
					{Total} & 103,737 & 961   & 545   & 288   & 128 \\ 
				\bottomrule
			
			\end{tabular}
		}
	\end{table}

\end{landscape}

\begin{landscape}
	
\begin{table}[ht]
	\caption{\bf Descriptive Statistics on Bid-Ask Spreads: Matched Sample}	\label{tab:data_summary_bas_matched}
	
	\begin{footnotesize}
		This figure table presents summary statistics of percentage bid-ask spreads for European stock markets around the introduction of temporary short-selling bans during the COVID-19 pandemic. The group of countries that restricted short sales (i.e., Austria, Belgium, France, Greece, Italy, and Spain) are highlighted in gray. The pre-ban (post-ban) period covers one month between February 17 and March 16, 2020 (March 17 and April 15, 2020). Bid-ask spreads are computed using daily bid and ask prices before taking the mean value across all stocks in each country. AC(1) denotes the first-order serial correlation. The sample includes small, mid, and large cap stocks (micro and nano caps are excluded) matched by market capitalization and industry classification according to the Industry Classification Benchmark (ICB) code, and runs daily between January 2 and June 12, 2020. Data are collected from {\it Datastream} and {\it Bloomberg}.	
				
	\end{footnotesize}

	\bigskip

	\newcolumntype{L}[1]{>{\raggedright\let\newline\\\arraybackslash\hspace{0pt}}m{#1}}
	\newcolumntype{C}[1]{>{\centering\let\newline\\\arraybackslash\hspace{0pt}}m{#1}}
	\newcolumntype{R}[1]{>{\raggedleft\let\newline\\\arraybackslash\hspace{0pt}}m{#1}}
	\setlength\extrarowheight{2pt}

	\sisetup{
		input-symbols         = [()],
		table-format          = 1.0,
		table-space-text-post = ***,
		table-align-text-post = false,
		table-text-alignment  = center,
		group-digits          = false
	}

	\robustify\bfseries
	\centering
	\scalebox{0.90}{
		
    \begin{tabular}{L{3.5cm}
    		S[table-column-width = 1.3cm]
    		S[table-column-width = 1.3cm]
    		S[table-column-width = 1.3cm]
    		S[table-column-width = 1.3cm]
    		S[table-column-width = 1.3cm]
    		S[table-column-width = 1.3cm]
    		S[table-column-width = 1.3cm]
    		S[table-column-width = 1.3cm]
    		S[table-column-width = 1.3cm]
    		S[table-column-width = 1.3cm]
    	}
    	\toprule

		& \multicolumn{5}{c}{\bf Panel A: Pre-Ban}               & \multicolumn{5}{c}{\bf Panel B: Post-Ban} \\
		\cmidrule(r){2-6} \cmidrule(l){7-11}

		\multicolumn{1}{l}{Country}
		& \multicolumn{1}{c}{mean} 
		& \multicolumn{1}{c}{sdev} 
		& \multicolumn{1}{c}{skew} 
		& \multicolumn{1}{c}{kurt} 
		& \multicolumn{1}{c}{AC(1)} 
		& \multicolumn{1}{c}{mean} 
		& \multicolumn{1}{c}{sdev} 
		& \multicolumn{1}{c}{skew} 
		& \multicolumn{1}{c}{kurt} 
		& \multicolumn{1}{c}{AC(1)} \\		
		\midrule

	\rc \bl{Austria} & \bl0.70 & \bl0.88 & \bl3.74 & \bl23.82 & \bl0.49 & \bl1.08 & \bl1.13 & \bl3.36 & \bl19.75 & \bl0.36 \\
	\rc \bl{Belgium} & \bl0.90 & \bl1.49 & \bl5.38 & \bl57.10 & \bl0.60 & \bl1.44 & \bl2.02 & \bl3.32 & \bl18.64 & \bl0.53 \\
	{Denmark} & 0.17  & 0.14  & 2.20  & 8.26  & 0.63  & 0.26  & 0.21  & 2.30  & 10.17 & 0.56 \\
	{Finland} & 0.39  & 0.51  & 6.16  & 59.60 & 0.33  & 0.52  & 0.52  & 3.63  & 23.62 & 0.39 \\
	\rc \bl{France} & \bl0.83 & \bl1.52 & \bl4.59 & \bl36.30 & \bl0.61 & \bl1.28 & \bl1.98 & \bl4.07 & \bl28.97 & \bl0.64 \\
	{Germany} & 0.71  & 1.39  & 7.11  & 77.78 & 0.42  & 0.89  & 1.32  & 4.84  & 38.15 & 0.49 \\
	\rc \bl{Greece} & \bl1.06 & \bl1.41 & \bl4.82 & \bl44.47 & \bl0.38 & \bl1.23 & \bl1.36 & \bl3.04 & \bl17.62 & \bl0.39 \\
	{Ireland} & 0.49  & 0.72  & 2.04  & 6.59  & 0.16  & 0.41  & 0.36  & 1.03  & 2.76  & -0.28 \\
	\rc \bl{Italy} & \bl0.58 & \bl1.34 & \bl7.16 & \bl76.06 & \bl0.37 & \bl0.84 & \bl1.21 & \bl3.53 & \bl21.43 & \bl0.44 \\
	{Netherlands} & 0.25  & 0.40  & 3.21  & 16.55 & 0.44  & 0.41  & 0.53  & 2.11  & 7.72  & 0.39 \\
	{Norway} & 0.70  & 1.43  & 5.38  & 43.36 & 0.55  & 1.04  & 2.25  & 5.22  & 34.22 & 0.75 \\
	{Poland} & 0.82  & 0.90  & 3.42  & 20.97 & 0.29  & 0.75  & 0.82  & 3.30  & 21.09 & 0.36 \\
	{Portugal} & 0.56  & 1.33  & 5.73  & 42.08 & 0.47  & 1.13  & 2.87  & 5.16  & 34.85 & 0.30 \\
	\rc \bl{Spain} & \bl0.41 & \bl0.86 & \bl6.13 & \bl57.02 & \bl0.49 & \bl0.63 & \bl0.99 & \bl3.88 & \bl25.86 & \bl0.40 \\
	{Sweden} & 0.39  & 0.64  & 8.59  & 124.29 & 0.21  & 0.57  & 0.71  & 4.22  & 31.05 & 0.45 \\
	{Switzerland} & 0.47  & 0.91  & 10.28 & 182.14 & 0.31  & 0.69  & 1.11  & 4.52  & 32.52 & 0.42 \\
	{United Kingdom} & 0.74  & 1.38  & 2.88  & 11.92 & 0.64  & 1.05  & 2.11  & 3.74  & 21.09 & 0.73 \\
	\midrule
	{Total} & 0.66  & 1.28  & 5.52  & 53.45 & 0.54  & 0.97  & 1.64  & 4.59  & 35.23 & 0.62 \\

		\bottomrule
	\end{tabular}
}

\end{table}

\end{landscape}

\begin{landscape}
	
\begin{table}[ht]
	\caption{\bf Descriptive Statistics on Stock Returns: Matched Sample}	\label{tab:data_summary_returns_matched}
	
	\begin{footnotesize}
		This figure table presents summary statistics of percentage stock returns for European stock markets around the introduction of temporary short-selling bans during the COVID-19 pandemic. The group of countries that restricted short sales (i.e., Austria, Belgium, France, Greece, Italy, and Spain) are highlighted in gray. The pre-ban (post-ban) period covers one month between February 17 and March 16, 2020 (March 17 and April 15, 2020). Returns are based on daily total return indices expressed in dollar terms before taking the mean value across all stocks in in each country.  AC(1) denotes the first-order serial correlation. The sample includes small, mid, and large cap stocks (micro and nano caps are excluded) matched by market capitalization and industry classification according to the Industry Classification Benchmark (ICB) code, and runs daily between January 2 and June 12, 2020. Data are collected from {\it Datastream} and {\it Bloomberg}.	
					
	\end{footnotesize}
	
	\bigskip

	\newcolumntype{L}[1]{>{\raggedright\let\newline\\\arraybackslash\hspace{0pt}}m{#1}}
	\newcolumntype{C}[1]{>{\centering\let\newline\\\arraybackslash\hspace{0pt}}m{#1}}
	\newcolumntype{R}[1]{>{\raggedleft\let\newline\\\arraybackslash\hspace{0pt}}m{#1}}
	\setlength\extrarowheight{2pt}

	\sisetup{
		input-symbols         = [()],
		table-format          = 1.0,
		table-space-text-post = ***,
		table-align-text-post = false,
		table-text-alignment  = center,
		group-digits          = false
	}

	\robustify\bfseries
	\centering
	\scalebox{0.90}{
		
    \begin{tabular}{L{3.5cm}
    		S[table-column-width = 1.3cm]
    		S[table-column-width = 1.3cm]
    		S[table-column-width = 1.3cm]
    		S[table-column-width = 1.3cm]
    		S[table-column-width = 1.3cm]
    		S[table-column-width = 1.3cm]
    		S[table-column-width = 1.3cm]
    		S[table-column-width = 1.3cm]
    		S[table-column-width = 1.3cm]
    		S[table-column-width = 1.3cm]
    	}
    	\toprule

		& \multicolumn{5}{c}{\bf Panel A: Pre-Ban}               & \multicolumn{5}{c}{\bf Panel B: Post-Ban} \\
		\cmidrule(r){2-6} \cmidrule(l){7-11}

		\multicolumn{1}{l}{Country}
		& \multicolumn{1}{c}{mean} 
		& \multicolumn{1}{c}{sdev} 
		& \multicolumn{1}{c}{skew} 
		& \multicolumn{1}{c}{kurt} 
		& \multicolumn{1}{c}{AC(1)} 
		& \multicolumn{1}{c}{mean} 
		& \multicolumn{1}{c}{sdev} 
		& \multicolumn{1}{c}{skew} 
		& \multicolumn{1}{c}{kurt} 
		& \multicolumn{1}{c}{AC(1)} \\		
		\midrule

	\rc \bl{Austria} & \bl-2.21 & \bl4.48 & \bl-1.53 & \bl7.22 & \bl0.02 & \bl0.85 & \bl5.35 & \bl0.63 & \bl5.81 & \bl0.08 \\
	\rc \bl{Belgium} & \bl-1.59 & \bl4.11 & \bl-1.23 & \bl6.87 & \bl-0.09 & \bl0.59 & \bl4.64 & \bl0.50 & \bl6.58 & \bl0.01 \\
	{Denmark} & -1.48 & 3.73  & -0.55 & 6.97  & 0.13  & 0.55  & 3.82  & -0.20 & 4.33  & -0.03 \\
	{Finland} & -1.42 & 3.41  & -1.02 & 5.88  & -0.01 & 0.56  & 4.25  & 0.38  & 5.74  & 0.01 \\
	\rc \bl{France} & \bl-1.92 & \bl4.42 & \bl-1.47 & \bl11.43 & \bl0.00 & \bl0.57 & \bl5.35 & \bl0.51 & \bl10.20 & \bl0.01 \\
	{Germany} & -1.66 & 4.18  & -1.00 & 7.27  & -0.06 & 0.67  & 5.15  & 0.48  & 5.87  & 0.03 \\
	\rc \bl{Greece} & \bl-2.51 & \bl6.25 & \bl-0.56 & \bl4.68 & \bl-0.18 & \bl1.08 & \bl5.81 & \bl0.58 & \bl5.00 & \bl-0.20 \\
	{Ireland} & -2.89 & 5.00  & -1.76 & 5.27  & -0.35 & -1.91 & 6.85  & 0.15  & 2.53  & 0.39 \\
	\rc \bl{Italy} & \bl-2.02 & \bl5.07 & \bl-0.86 & \bl6.79 & \bl-0.26 & \bl0.57 & \bl4.28 & \bl0.53 & \bl4.97 & \bl0.07 \\
	{Netherlands} & -2.25 & 4.62  & -1.87 & 7.95  & -0.07 & 0.35  & 5.58  & 0.68  & 5.74  & 0.04 \\
	{Norway} & -2.81 & 6.27  & -1.22 & 6.84  & 0.03  & 0.82  & 7.41  & 2.82  & 38.43 & 0.12 \\
	{Poland} & -1.63 & 4.97  & -0.81 & 7.66  & 0.14  & 0.61  & 5.11  & 2.18  & 14.62 & -0.04 \\
	{Portugal} & -1.68 & 4.23  & -2.22 & 10.10 & 0.04  & 0.52  & 5.16  & 0.84  & 7.49  & -0.20 \\
	\rc \bl{Spain} & \bl-1.88 & \bl4.88 & \bl-0.79 & \bl10.31 & \bl-0.06 & \bl0.52 & \bl4.96 & \bl0.95 & \bl8.85 & \bl-0.02 \\
	{Sweden} & -2.06 & 4.57  & -1.19 & 6.76  & -0.02 & 0.84  & 5.73  & 0.05  & 4.41  & 0.00 \\
	{Switzerland} & -1.29 & 3.24  & -1.57 & 7.62  & 0.08  & 0.55  & 3.94  & 0.65  & 5.28  & -0.08 \\
	{United Kingdom} & -2.24 & 5.06  & -1.68 & 10.06 & 0.04  & 0.57  & 8.33  & 2.17  & 38.88 & 0.03 \\
	\midrule
	{Total} & -1.95 & 4.67  & -1.29 & 9.17  & -0.04 & 0.61  & 5.68  & 1.48  & 33.04 & 0.02 \\

		\bottomrule
	\end{tabular}
}

\end{table}

\end{landscape}

\begin{landscape}
\begin{table}[ht]
	\caption{\bf Bid-Ask Spreads and Institutional Ownership: Two-month window}\label{tab:DiD_bas_IO_60days}

	\begin{footnotesize}	
	This table presents difference-in-differences estimates associated with the introduction of temporary short-selling bans in European stock markets during the COVID-19 pandemic. The dependent variable is the percentage bid-ask spread based on daily bid and ask prices for stocks with low institutional ownership (bottom tercile) in Panel A and stocks with high institutional ownership (top tercile) in Panel B, respectively.  \emph{Country} is a dummy variable that equals one (zero) for the treated (control) group of European countries.  \emph{Ban} is a dummy variable that equals one (zero) for a post-treatment (pre-treatment) period of two months that goes from March 17 to May 15, 2020 (January 17 to March 16, 2020). The set of controls includes firm size and sovereign CDS spread. Specifications are complemented with stock and time (calendar date) fixed effects {\it fe}. Standard errors (in parentheses) are clustered by stock and time (calendar date) dimensions. *, **, ***, indicate statistical significance at the 10\%, 5\%, and 1\% level, respectively. The sample includes small, mid, and large cap stocks (micro and nano caps are excluded). Data are collected from {\it Datastream} and {\it Bloomberg}.	
	
\end{footnotesize} 

		\bigskip
		
	\newcolumntype{L}[1]{>{\raggedright\let\newline\\\arraybackslash\hspace{0pt}}m{#1}}
	\newcolumntype{C}[1]{>{\centering\let\newline\\\arraybackslash\hspace{0pt}}m{#1}}
	\newcolumntype{R}[1]{>{\raggedleft\let\newline\\\arraybackslash\hspace{0pt}}m{#1}}
	\setlength\extrarowheight{0pt}

	\sisetup{
		input-symbols         = [()],
		table-format          = 1.0,
		table-space-text-post = ***,
		table-align-text-post = false,
		table-text-alignment  = center,
		group-digits          = false
	}

	\robustify\bfseries
	\centering
	\scalebox{0.90}{
		
		\begin{tabular}{L{2.8cm}
				S[table-column-width = 1.9cm]
				S[table-column-width = 1.9cm]
				S[table-column-width = 1.9cm]
				S[table-column-width = 1.9cm]
				S[table-column-width = 1.9cm]
				S[table-column-width = 1.9cm]
				S[table-column-width = 1.9cm]
				S[table-column-width = 1.9cm]
			}
	
	\toprule

			& \multicolumn{4}{c}{\bf Panel A: Low Institutional Ownership} 
			& \multicolumn{4}{c}{\bf Panel B: High Institutional Ownership} \\
			\cmidrule(r){2-5} \cmidrule(l){6-9}

		 	& \multicolumn{1}{c}{(1)} 
			& \multicolumn{1}{c}{(2)} 
		    & \multicolumn{1}{c}{(3)} 
		 	& \multicolumn{1}{c}{(4)} 
			& \multicolumn{1}{c}{(5)} 
			& \multicolumn{1}{c}{(6)} 
			& \multicolumn{1}{c}{(7)}
			& \multicolumn{1}{c}{(8)} \\
			
			\midrule
			\addlinespace[5pt]

	\rc{\it Ban $\times$ Country}   & \bl 0.104\fn{**} & \bl 0.102\fn{**} & \bl 0.103\fn{**} & \bl 0.072 & \bl 0.211\fn{***} & \bl 0.237\fn{***} & \bl 0.241\fn{***} & \bl 0.242\fn{***} \\
    \rc &   (0.044) &   (0.049) &   (0.050) &   (0.063) &   (0.059) &  (0.063) &   (0.063) &  (0.078) \\
	\addlinespace[5pt]
	
	{\it Ban} & 0.286\fn{***} & 0.289\fn{***} &       &       & 0.197\fn{***} & 0.199\fn{***} &       &  \\
	& (0.052) & (0.055) &       &       & (0.033) & (0.036) &       &  \\
	\addlinespace[5pt]
	
	{\it Country} & 0.023 &       &       &       & 0.133 &       &       &  \\
	& (0.074) &       &       &       & (0.097) &       &       &  \\
	\addlinespace[5pt]

	{\it Constant} & 0.691\fn{***} & 0.700\fn{***} & 0.842\fn{***} & 2.062\fn{*} & 0.544\fn{***} & 0.574\fn{***} & 0.672\fn{***} & 1.295\fn{**} \\
	& (0.060) & (0.040) & (0.010) & (1.061) & (0.053) & (0.025) & (0.008) & (0.578) \\
	\addlinespace[5pt]
	
	{$R^2$} & 0.014 & 0.507 & 0.531 & 0.531 & 0.014 & 0.632 & 0.641 & 0.641  \\
	\cmidrule{2-9}
	{$\#$ \it obs} & {42,275} & {42,275} & {42,275} & {42,275} & {42,461} & {42,461} & {42,461} & {42,461} \\
	\midrule
	{\it Stock fe} &       & {\checkmark}     & {\checkmark}     & {\checkmark}     &       & {\checkmark}     & {\checkmark}     & {\checkmark} \\
	{\it Time fe}  &       &                  & {\checkmark}     & {\checkmark}     &       &                  & {\checkmark}     & {\checkmark} \\
	{\it Controls} &       &                  &                  & {\checkmark}     &       &                  &                  & {\checkmark} \\
	\bottomrule

\end{tabular}%
}

\end{table}
\end{landscape}

\begin{landscape}
\begin{table}[ht]
	\caption{\bf Stock Returns and Institutional Ownership: Two-month window}\label{tab:DiD_distribution_IO_60days}

	\begin{footnotesize}
	This table presents difference-in-differences estimates associated with the introduction of temporary short-selling bans in European stock markets during the COVID-19 pandemic. The dependent variables are the percentage mean, median, volatility, and maximum drawdown of total return indices in dollar terms based on a two-month window around the enactment of short-selling bans. The pre-treatment (post-treatment) period runs between January 17 and March 16, 2020 (March 17 and May 15, 2020). Panel A presents estimates for stocks with low institutional ownership (bottom tercile), whereas Panel B  for stocks with high institutional ownership (top tercile). \emph{Country} is a dummy variable that equals one for the group of countries that introduced short-selling bans (i.e., Austria, Belgium, France, Greece, Italy, and Spain), and zero for the control group of countries (i.e., Denmark, Finland, Germany, Ireland, Netherlands, Norway,  Poland, Portugal, Sweden, Switzerland, and the United Kingdom).  \emph{Ban} is a dummy variable that equals one (zero) for the post-treatment (pre-treatment) period. The set of controls includes firm size and sovereign CDS spread. Standard errors (in parentheses) are clustered by stock dimension. *, **, ***, indicate statistical significance at the 10\%, 5\%, and 1\% level, respectively. The sample includes small, mid, and large cap stocks (micro and nano caps are excluded). Data are collected from {\it Datastream} and {\it Bloomberg}.
	
\end{footnotesize} 

		\bigskip

\newcolumntype{L}[1]{>{\raggedright\let\newline\\\arraybackslash\hspace{0pt}}m{#1}}
\newcolumntype{C}[1]{>{\centering\let\newline\\\arraybackslash\hspace{0pt}}m{#1}}
\newcolumntype{R}[1]{>{\raggedleft\let\newline\\\arraybackslash\hspace{0pt}}m{#1}}
\setlength\extrarowheight{0pt}

\sisetup{
	input-symbols         = [()],
	table-format          = 1.0,
	table-space-text-post = ***,
	table-align-text-post = false,
	table-text-alignment  = center,
	group-digits          = false
}

\robustify\bfseries
\centering
\scalebox{0.90}{
	
	\begin{tabular}{L{2.8cm}
			S[table-column-width = 2.0cm]
			S[table-column-width = 2.0cm]
			S[table-column-width = 2.0cm]
			S[table-column-width = 2.0cm]
			S[table-column-width = 2.0cm]
			S[table-column-width = 2.0cm]
			S[table-column-width = 2.0cm]
			S[table-column-width = 2.0cm]
		}
		
		\toprule

	        & \multicolumn{4}{c}{\bf Panel A: Low Institutional Ownership} 
	        & \multicolumn{4}{c}{\bf Panel B: High Institutional Ownership}\\
	        
			\cmidrule(lr){2-5}
			\cmidrule(lr){6-9}	
	
     & 	\multicolumn{1}{c}{\multirow{2}{*}{Mean}} &  \multicolumn{1}{c}{\multirow{2}{*}{Median}} & \multicolumn{1}{c}{\multirow{2}{*}{Volatility}} & {Maximum} 
     & 	\multicolumn{1}{c}{\multirow{2}{*}{Mean}} &  \multicolumn{1}{c}{\multirow{2}{*}{Median}} & \multicolumn{1}{c}{\multirow{2}{*}{Volatility}} & {Maximum} \\

	 &  \multicolumn{1}{c}{}                       &  \multicolumn{1}{c}{}                       & \multicolumn{1}{c}{}                            & {drawdown}
	 &  \multicolumn{1}{c}{}                       &  \multicolumn{1}{c}{}                       & \multicolumn{1}{c}{}                            & {drawdown}\\

	\midrule

	\rc{\it Ban $\times$ Country}   & \bl 0.077 & \bl-0.006 & \bl-0.367\fn{***} & \bl3.691\fn{***} & \bl-0.112 & \bl-0.120** & \bl-0.705\fn{***} & \bl0.841 \\
	\rc & (0.068) & (0.056) & (0.122) & (1.047) & (0.076) & (0.061) & (0.135) & (1.155) \\
	\addlinespace[5pt]
	{\it Ban} & 1.240\fn{***} & 0.624\fn{***} & 0.456\fn{***} & 20.612\fn{***} & 1.446\fn{***} & 0.715\fn{***} & 0.961\fn{***} & 22.548\fn{***} \\
	& (0.048) & (0.038) & (0.072) & (0.734) & (0.048) & (0.036) & (0.070) & (0.637) \\
	\addlinespace[5pt]
	{\it Country} & -0.112\fn{**} & -0.056 & 0.175 & -3.287\fn{***} & 0.012 & 0.059 & -0.127 & 0.934 \\
	& (0.050) & (0.037) & (0.109) & (1.190) & (0.059) & (0.044) & (0.133) & (1.443) \\
	\addlinespace[5pt]
	{\it Constant} & -1.014\fn{***} & -0.775\fn{***} & 4.521\fn{***} & -41.203\fn{***} & -1.398\fn{***} & -1.071\fn{***} & 5.524\fn{***} & -49.376\fn{***} \\
	& (0.088) & (0.076) & (0.355) & (2.558) & (0.099) & (0.109) & (0.410) & (2.671) \\
	\addlinespace[5pt]
	{$R^2$} & 0.609 & 0.342 & 0.037 & 0.500 & 0.608 & 0.322 & 0.129 & 0.502 \\	
	\addlinespace[5pt]
	
	\cmidrule{2-9}
	{\it $\#$ obs} & {1,036} & {1,036} & {1,036} & {1,036} & {1,036} & {1,036} & {1,036} & {1,036}  \\
	\midrule
    {\it Controls}  &   {\checkmark}    &   {\checkmark}  &   {\checkmark}    &   {\checkmark}  & {\checkmark}      &   {\checkmark} &   {\checkmark}    &   {\checkmark}      \\

	\bottomrule

\end{tabular}
}

\end{table}
\end{landscape}

\end{appendices}
\end{document}